%% file: SQCLP_TPLP2010TR.tex
\documentclass{tlp_tr}

\pdfoutput=1

\usepackage{stmaryrd, amssymb, extarrows}
\usepackage{verbatim} 
\usepackage{fancyvrb}
\usepackage{hyperref}

\hypersetup{
  pdftitle = {A Transformation-based Implementation for CLP with Qualification and Proximity (Preliminary Version) [TR SIC-4-10]},
  pdfkeywords = {Constraint Logic Programming, Program Transformation, Qualification Domains and Values, Similarity and Proximity Relations, Flexible Information Retrieval},
  pdfauthor = {R. Caballero, M. Rodríguez-Artalejo and C. A. Romero-Díaz}
}

\newcommand{\nc}{\newcommand}

\nc{\prox}{$\sim$}

\nc{\SL}{\mathcal{T}} 
\nc{\GL}{\mathcal{G}} 

\nc{\trans}[2]{S_{#1}(#2)} 

\nc{\extended}[2]{H_{#1}(#2)} 
\nc{\abstracted}[2]{{#1}_{#2}} 

\nc{\MAProg}{\mathcal{P}_{E, \simrel}}

\nc{\linear}[1]{#1_\ell}

\nc{\Set}{\mathcal{S}}
\nc{\transform}[1]{#1^{\mathcal{T}}}

\nc{\transQ}[1]{\transform{#1}}   
\nc{\transRen}[1]{\widehat{#1}}       
\nc{\transSim}[1]{{#1}_{S}}       
\nc{\transUni}[1]{{#1}_{\sim}}    
\nc{\transPay}[1]{{#1}_{\mathit{pay}}}    

\nc{\elimS}[1]{\mathrm{elim}_{\simrel}(#1)} 
\nc{\elimD}[1]{\mathrm{elim}_{\qdom}(#1)} 

\nc{\Con}[1]{\mbox{Con}_{#1}}  

\nc{\pc}[3]{\mathsf{#1}_{#2}(#3)} 
\nc{\qval}[1]{\pc{qVal}{}{#1}} 
\nc{\qbound}[1]{\pc{qBound}{}{#1}} 
\nc{\qvali}[2]{\pc{qVal}{#1}{#2}} 
\nc{\qboundi}[2]{\pc{qBound}{#1}{#2}} 

\nc{\encode}[1]{\ulcorner#1\urcorner} 

\nc{\spair}[2]{\llparenthesis\, #1, #2 \,\rrparenthesis}

\newcommand{\schemenp}[1]{\mbox{#1}} 
\newcommand{\scheme}[2]{\schemenp{#1}(#2)} 

\newcommand{\clp}[1]{\scheme{CLP}{#1}} 

\newcommand{\qclp}[2]{\scheme{QCLP}{#1,#2}} 

\newcommand{\sqclp}[3]{\scheme{SQCLP}{#1,#2,#3}} 

\newcommand{\cdom}{\mathcal{C}} 

\newcommand{\rdom}{\mathcal{R}} 

\newcommand{\qdom}{\mathcal{D}} 
\newcommand{\aqdomd}[1]{D_{#1} \setminus \{\bt\}} 
\newcommand{\bqdomd}[1]{(D_{#1} \setminus \{\bt\}) \uplus \{?\}} 
\newcommand{\aqdom}{\aqdomd{}} 
\newcommand{\bqdom}{\bqdomd{}} 

\newcommand{\bt}{\mathrm{\mathbf{b}}} 
\newcommand{\tp}{\mathrm{\mathbf{t}}} 
\newcommand{\dleq}{\trianglelefteqslant} 
\newcommand{\dlt}{\vartriangleleft} 
\newcommand{\dgeq}{\trianglerighteqslant} 

\newcommand{\B}{\mathcal{B}} 
\newcommand{\U}{\mathcal{U}} 
\newcommand{\W}{\mathcal{W}} 

\newcommand{\simrel}{\mathcal{S}} 
\newcommand{\sid}{\simrel_{\mathrm{id}}} 

\newcommand{\Prog}{\mathcal{P}} 

\newcommand{\Var}{\mathcal V\!ar} 
\newcommand{\War}{\mathcal W\!ar} 

\newcommand{\set}[2]{\mathrm{#1}(#2)} 
\newcommand{\domset}[1]{\set{dom}{#1}} 
\newcommand{\varset}[1]{\set{var}{#1}} 
\newcommand{\warset}[1]{\set{war}{#1}} 

\nc{\Exp}{\mbox{Exp}_{\bot}(\Sigma,B,\Var)} 
\nc{\TExp}{\mbox{Exp}(\Sigma,B,\Var)} 
\nc{\GExp}{\mbox{Exp}_{\bot}(\Sigma,B)} 
\nc{\TGExp}{\mbox{Exp}(\Sigma,B)} 
\nc{\Term}{\mbox{Term}_{\bot}(\Sigma,B,\Var)} 
\nc{\TTerm}{\mbox{Term}(\Sigma,B,\Var)} 
\nc{\GTerm}{\mbox{Term}_{\bot}(\Sigma,B)} 
\nc{\TGTerm}{\mbox{Term}(\Sigma,B)} 

\nc{\At}{\mbox{At}(\Sigma,B,\Var)} 
\nc{\GAt}{\mbox{GAt}(\Sigma,B)} 
\nc{\PAt}{\mbox{PAt}(\Sigma,B,\Var)} 
\nc{\GPAt}{\mbox{GPAt}(\Sigma,B)} 

\nc{\Atz}{\mbox{At}_{\Sigma}} 
\nc{\QAtz}{\mbox{At}_{\Sigma}(\qdom)} 

\nc{\sust}{\mbox{Subst}_\Sigma} 
\nc{\Sust}{\mbox{Subst}(\Sigma,B,\Var)} 
\nc{\GSust}{\mbox{GSubst}(\Sigma,B)} 

\nc{\Soln}[3]{\mbox{Sol}_{#1}^{#2}(#3)} 
\nc{\Sol}[2]{\Soln{#1}{}{#2}} 
\nc{\GSol}[2]{\mbox{GSol}_{#1}(#2)}
\nc{\Solc}[1]{\Sol{\cdom}{#1}}
\nc{\CAns}[2]{\mbox{C\!Ans}_{#1}(#2)} 

\newcommand{\qat}[2]{#1 \sharp #2} 
\newcommand{\cat}[2]{#1 \Leftarrow #2} 
\newcommand{\cqat}[3]{\qat{#1}{#2} \Leftarrow #3} 

\newcommand{\qgets}[1]{\xleftarrow{#1}} 

\nc{\closure}[1]{\mbox{cl}_{#1}} 

\newcommand{\model}[1]{~{\models_{#1}}~} 

\newcommand{\M}[1]{\mathcal{M}_{#1}} 
\newcommand{\Mp}{\M{\Prog}} 

\newcommand{\eqdef}{~{=_{\mathrm{def}}}~} 
\newcommand{\infi}{\bigsqcap} 
\newcommand{\entail}[1]{~{\succcurlyeq_{#1}}~} 

\newcommand{\sep}{\talloblong} 

\newcommand{\tup}[1]{\overline{#1}}   
\newcommand{\ntup}[2]{\tup{#1}_{#2}}  

\newcommand{\infx}[2]{\ {\vdash}_{\!#1}^{\!#2}\ } 

\newcommand{\CHL}{\mbox{CHL}} 
\newcommand{\chln}[2]{\infx{#1}{#2}} 
\newcommand{\chl}[1]{\chln{#1}{}} 
\newcommand{\chlc}{\chln{\cdom}{}} 

\newcommand{\QCHL}{\mbox{QCHL}} 
\newcommand{\qchln}[3]{\infx{#1,#2}{#3}} 
\newcommand{\qchldc}{\qchln{\qdom}{\cdom}{}} 
\newcommand{\qchldcn}[1]{\qchln{\qdom}{\cdom}{#1}} 

\newcommand{\SQCHL}{\mbox{SQCHL}} 
\newcommand{\sqchln}[4]{\infx{#1,#2,#3}{#4}} 
\newcommand{\sqchlrdc}{\sqchln{\simrel}{\qdom}{\cdom}{}} 
\newcommand{\sqchlrdcn}[1]{\sqchln{\simrel}{\qdom}{\cdom}{#1}} 

\title[A Transformation-based Implementation for SQCLP (Preliminary Version)]
    {A Transformation-based Implementation\\ for CLP with Qualification and Proximity \thanks{This work has been partially supported by the Spanish projects STAMP (TIN2008-06622-C03-01), PROMETIDOS--CM (S2009TIC-1465) and GPD--UCM (UCM--BSCH--GR58/08-910502).}\\
    {\large Preliminary Version (Technical Report SIC-4-10)}}

  \author[R. Caballero, M. Rodr\'iguez-Artalejo and C. A. Romero-D\'iaz]
    {R. CABALLERO, M. RODR\'IGUEZ-ARTALEJO and C. A. ROMERO-D\'IAZ\\
    Departamento de Sistemas Inform\'aticos y Computaci\'on, Universidad Complutense\\
    Facultad de Inform\'atica, 28040 Madrid, Spain\\
    \email{\{rafa,mario\}@sip.ucm.es, cromdia@fdi.ucm.es}}

\begin{document}
\maketitle
\thispagestyle{empty}

\begin{abstract}
Uncertainty in logic programming has been widely investigated in the last decades,
leading to multiple extensions of the classical LP paradigm.
However, few of these are designed as
extensions of the well-established and powerful CLP scheme for Constraint Logic Programming.
In a previous work we have proposed the SQCLP ({\em proximity-based qualified constraint logic programming}) scheme as a
quite expressive extension of CLP with support for qualification values and proximity relations
as generalizations of uncertainty values and similarity relations, respectively.
In this paper we provide a transformation technique for transforming SQCLP programs and goals into semantically equivalent CLP programs and goals, and a practical Prolog-based implementation of some particularly useful instances of the SQCLP scheme.
We also illustrate, by showing some simple---and working---examples, how the prototype can be effectively used as a tool for solving problems where qualification values and proximity relations play a key role.
Intended use of SQCLP includes flexible information retrieval applications.
\end{abstract}


\newtheorem{defn}{Definition}[section]
\newtheorem{thm}{Theorem}[section]
\newtheorem{lem}{Lemma}[section]
\newtheorem{prop}{Proposition}[section]
\newtheorem{exmp}{Example}[section]
\newtheorem{cor}{Corollary}[section]


\input{J1_0} 

\input{J2_0} 

\input{J3_0} 

\input{J4_0} 

\input{J5_0} 

\input{J6_0} 


\bibliographystyle{acmtrans}
\bibliography{../biblio}



\end{document}

%% file: J1_0.tex
\section{Introduction}
\label{sec:introduction}


Many extensions of LP ({\em logic programming}) to deal with uncertain knowledge and uncertainty have been proposed in the last decades.
These extensions have been proposed from different and somewhat unrelated perspectives, leading to multiple approaches in the way of using uncertain knowledge and understanding uncertainty.


A recent work by us \cite{RR10} focuses on the declarative semantics of a new proposal for an extension of the CLP scheme supporting qualification values and proximity relations.
More specifically, this work defines a new generic scheme SQCLP ({\em proximity-based qualified constraint logic programming}) whose instances $\sqclp{\simrel}{\qdom}{\cdom}$ are parameterized by a proximity relation $\simrel$, a qualification domain $\qdom$ and a constraint domain $\cdom$.
The current paper is intended as a continuation of \cite{RR10} with the aim of providing a semantically correct program transformation technique that allows us to implement a sound and complete implementation of some useful instances of SQCLP on top of existing CLP systems like {\em SICStus Prolog} or {\em SWI-Prolog}.
In the introductory section of \cite{RR10} we have already summarized some related approaches of SQCLP with a special emphasis on their declarative semantics and their main semantic differences with SQCLP.
In the next paragraphs we present a similar overview but, this time, putting the emphasis on the goal resolution procedures and system implementation techniques, when available.


Within the extensions of LP using annotations in program clauses we can find the seminal proposal of {\em quantitative logic programming} by \cite{VE86} that inspired later works such as the GAP ({\em generalized annotated programs}) framework by \cite{KS92} and the QLP ({\em qualified logic programming}) scheme by us \cite{RR08}.
In the proposal of van Emden, one can find a primitive goal solving procedure based in and/or trees (these are similar to the alpha-beta trees used in game theory), used to prune the search space when proving some specific ground atom for some certainty value in the real interval $[0,1]$.
In the case of GAP, the goal solving procedure uses constrained SLD resolution in conjunction with a---costly---computation of so-called {\em reductants} between variants of program clauses.
In contrast, QLP goal solving uses a more efficient resolution procedure called SLD($\qdom$) resolution, implemented by means of real domain constraints, used to compute the qualification value of the head atom based on the attenuation factor of the program clause and the previously computed qualification values of the body atoms.
Admittedly, the gain in efficiency of SLD($\qdom$) w.r.t. GAP's goal solving procedure is possible because QLP focuses on a more specialized class of annotated programs.
While in all these three approaches there are some results of soundness and completeness, the results for the QLP scheme are the stronger ones (again, thanks to its also more focused scope w.r.t. GAP).


From a different viewpoint, extensions of LP supporting uncertainty can be roughly classified into two major lines: approaches based in fuzzy logic \cite{Zad65,Haj98} and approaches based in similarity relations.
Historically, Fuzzy LP languages were motivated by expert knowledge representation applications.
Early Fuzzy LP languages implementing the resolution principle introduced in \cite{Lee72} include Prolog-Elf \cite{IK85}, Fril Prolog \cite{BMP95} and F-Prolog \cite{LL90}.  More recent approaches such as the Fuzzy LP languages in \cite{Voj01,GMV04} and Multi-Adjoint LP (MALP for short) in the sense of \cite{MOV01a} use clause annotations and a fuzzy interpretation of the connectives
and aggregation operators occurring in program clauses and goals.
The Fuzzy Prolog system proposed in \cite{GMV04} is implemented by means of real constrains on top of a CLP($\rdom$) system, using a syntactic expansion of the source code during the Prolog compilation.
A complete procedural semantics for MALP using reductants has been presented in \cite{MOV01b}.
A method for translating a MALP like program into standard Prolog has been described in \cite{JMP09}.


The second line of research mentioned in the previous paragraph
was motivated by applications in the field of flexible query answering.
Classical LP is extended to Similar\-i\-ty-based LP (SLP for short),
leading to languages which keep the classical syntax of LP clauses but use a similarity relation over a set of symbols $S$ to allow ``flexible'' unification of syntactically different symbols with a certain approximation degree.
Similarity relations over a given set $S$ have been defined in \cite{Zad71,Ses02} and related literature as fuzzy relations represented by mappings $\simrel : S \times S \to [0,1]$ which satisfy reflexivity, symmetry and transitivity axioms analogous to those required for classical equivalence relations.
Resolution with flexible unification can be used as a sound and complete goal solving procedure for SLP languages as shown e.g. in \cite{Ses02}.
SLP languages include {\em Likelog} \cite{AF99,Arc02} and more recently {\em SiLog} \cite{LSS04},
which has been implemented by means of an extended Prolog interpreter and proposed as a useful tool for web knowledge discovery.


In the last years, the SLP approach has been extended in various ways. The SQLP ({\em similarity-based qualified logic programming}) scheme proposed in \cite{CRR08} extended SLP by allowing program clause annotations in QLP style and generalizing similarity relations to mappings $\simrel : S \times S \to D$ taking values in a qualification domain not necessarily identical to the real interval $[0,1]$. As implementation technique for SQLP,
\cite{CRR08} proposed a semantically correct program transformation into QLP, whose goal solving procedure has been described above.
Other related works on transformation-based implementations of SLP languages include \cite{Ses01,MOV04}.
More recently, the SLP approach has been generalized to work with {\em proximity relations} in the sense of \cite{DP80} represented by mappings $\simrel : S \times S \to [0,1]$ which satisfy reflexivity and symmetry axioms but do not always satisfy transitivity.
SLP like languages using proximity relations include {\sf Bousi$\sim$Prolog} \cite{JR09} and the SQCLP scheme \cite{RR10}.
Two prototype implementations of {\sf Bousi$\sim$Prolog} are available:
a low-level implementation \cite{JR09b} based on an adaptation of the classical WAM (called {\em Similarity WAM}) implemented in {\sc Java} and able to execute a Prolog program in the context of a similarity relation defined on the first order alphabet induced by that program;
and a high-level implementation \cite{JRG09} done on top of {\em SWI-Prolog} by means of a program transformation from {\sf Bousi$\sim$Prolog} programs into a so-called {\em Translated BPL code} than can be executed according to the weak SLD resolution principle by a meta-interpreter.


Let us now refer to approaches related to constraint solving and CLP.
An analogy of proximity relations in the context of partial constraint satisfaction can be found in \cite{FW92},
where several metrics are proposed to measure the proximity between the solution sets of two different constraint satisfaction problems.
Moreover, some extensions of LP supporting uncertain reasoning use constraint solving as implementation technique,
as discussed in the previous paragraphs.
However, we are only aware of three approaches which have been conceived as extensions of the classical CLP scheme proposed for the first time in \cite{JL87}.
These three approaches are:
\cite{Rie98phd} that extends the formulation of CLP by \cite{HS88} with quantitative LP in the sense of \cite{VE86} and adapts van Emden's idea of and/or trees to obtain a goal resolution procedure;
\cite{BMR01} that proposes a semiring-based approach to CLP, where constraints are solved in a soft way with levels of consistency represented by values of the semiring, and is implemented with {\tt clp(FD,S)} for a particular class of semirings which enable to use local consistency algorithms,
as described  in \cite{GC98};
and the SQCLP scheme proposed in our previous work \cite{RR10}, which was designed as a common extension of SQLP and CLP.


As we have already said at the beginning of this introduction, this paper deals with transformation-based implementations of the SQCLP scheme.
Our main results include: a) a transformation technique for transforming SQCLP programs into semantically equivalent CLP programs via two specific program transformations named elim$_\simrel$ and elim$_\qdom$; and b) and a practical Prolog-based implementation which relies on the aforementioned program transformations and supports several useful SQCLP instances.
As far as we know, no previous work has dealt with the implementation of extended LP languages for uncertain reasoning which are able to support clause annotations, proximity relations and CLP style programming.
In particular, our previous paper \cite{CRR08} only presented a transformation analogous to elim$_\simrel$ for a programming scheme less expressive than SQCLP,  which supported neither non-transitive proximity relations nor CLP programming.
Moreover, the transformation-based implementation reported in \cite{CRR08} was not implemented in a system.


The reader is assumed to be familiar with the semantic foundations of LP  \cite{Llo87,Apt90} and CLP \cite{JL87,JMM+98}.
The rest of the paper is structured as follows:
Section \ref{sec:sqclp} presents a brief overview of the semantics of the  SQCLP scheme, focusing on the essential notions needed  to understand the following sections and concluding with an abstract discussion of goal solving procedures for SQCLP.
Section \ref{sec:qclpclp} briefly discusses two specializations of SQCLP, namely QCLP and CLP, which are used as the targets
of the program transformations elim$_\simrel$ and elim$_\qdom$, respectively.
Section \ref{sec:implemen} presents these two program transformations along with mathematical results which prove their semantic correctness, relying on the declarative semantics of the SQCLP, QCLP and CLP schemes.
Section \ref{sec:practical} presents a Prolog-based prototype system which relies on  the transformations proposed in the previous section and implements several useful SQCLP instances.
Finally, Section \ref{sec:conclusions} summarizes conclusions and points to some lines of planned future research.

%% file: J2_0.tex
\section{The Scheme SQCLP and its Declarative Semantics}
\label{sec:sqclp}

We present in this section a short overview of the declarative semantics of the SQCLP scheme originally presented in \cite{RR10},
focusing on the essential notions needed  to understand the following sections.
Interested readers are referred to \cite{RR10} and its extended version \cite{RR10TR} for
a full-fledged exposition of SQCLP semantics and a discussion of various extended LP languages for uncertain reasoning which can
be obtained as specializations and instances of SQCLP.
Some technical notions and results from \cite{RR10TR} will be cited along this paper when needed to support mathematical proofs.


{\em Constraint domains} $\cdom$, sets of constraints $\Pi$ and their solutions, as well as terms, atoms and substitutions over a given $\cdom$ are well known notions underlying  the CLP scheme. The reader is referred  \cite{RR10TR} for a relational formalization of constraint domains and some examples, including the real constraint domain $\rdom$.
We assume the following classification of atomic $\cdom$-constraints: defined atomic constraints $p(\ntup{t}{n})$, where $p$ is a program-defined predicate symbol; primitive constraints $r(\ntup{t}{n})$ where $r$ is a $\cdom$-specific primitive predicate symbol; and equations $t == s$.

We use $\Con{\cdom}$ as a notation for the set of all $\cdom$-constraints and
$\kappa$ as a notation for an atomic primitive constraint.
Constraints are interpreted by means of {\em $\cdom$-valuations} $\eta \in \mbox{Val}_{\cdom}$,
which are ground substitutions.
The set $\Solc{\Pi}$ of solutions of $\Pi \subseteq \Con{\cdom}$ includes all the valuations $\eta$ such that $\Pi\eta$
is true when interpreted in $\cdom$.
$\Pi \subseteq \Con{\cdom}$ is called {\em satisfiable} if $\Solc{\Pi} \neq \emptyset$ and {\em unsatisfiable} otherwise.
$\pi \in \Con{\cdom}$ {\em is entailed} by $\Pi \subseteq \Con{\cdom}$
(noted  $\Pi \model{\cdom} \pi$) iff $\Solc{\Pi} \subseteq \Solc{\pi}$.


{\em Qualification domains} were first introduced in \cite{RR08} with the aim of providing elements, called qualification values,  which can be attached to computed answers. They are  defined as structures $\qdom = \langle D, \dleq, \bt, \tp, \circ \rangle$ verifying the following requirements:
\begin{enumerate}
\item
$\langle D, \dleq, \bt, \tp \rangle$ is a lattice with extreme points $\bt$ (called {\em infimum} or {\em bottom} element) and $\tp$ (called {\em maximum} or {\em top} element) w.r.t. the partial ordering $\dleq$ (called {\em qualification ordering}). For given elements  $d, e \in D$, we  write $d \sqcap e$ for the {\em greatest lower bound} ($glb$) of $d$ and $e$, and $d \sqcup e$ for the {\em least upper bound} ($lub$) of $d$ and $e$. We also write $d \dlt e$ as abbreviation for $d \dleq e \land d \neq e$.
    \item $\circ : D \times D \rightarrow D$, called {\em attenuation operation}, verifies the following axioms:
        \begin{enumerate}
            \item $\circ$ is associative, commutative and monotonic w.r.t. $\dleq$.
            \item $\forall d \in D : d \circ \tp = d$ and $d \circ \bt = \bt$.
            \item $\forall d, e \in D  : d \circ e \dleq e$ and even $\bt \neq d \circ e \dleq e$ if $d,e \in \aqdom$.
            \item $\forall d, e_1, e_2 \in D : d \circ (e_1 \sqcap e_2) = (d \circ e_1) \sqcap (d \circ e_2)$.
        \end{enumerate}
\end{enumerate}
For any $S = \{e_1, e_2, \ldots, e_n\} \subseteq D$, the $glb$ (also called {\em infimum} of $S$)
exists and can be computed as $\infi S =e_1 \sqcap e_2 \sqcap \cdots \sqcap e_n$
(which reduces to $\tp$ in the case $n = 0$).
The dual claim concerning $lub$s is also true.
As an easy consequence of the axioms, one gets the identity $d \circ \infi S =  \infi \{d \circ e \mid e \in S\}$.

Technical details, explanations and examples can be found in  \cite{RR10TR},
including: the qualification domain $\B$ of classical boolean values,
the qualification domain $\U$ of uncertainty values,
the qualification domain $\W$ of weight values,
and other qualification domains built from these by means of the strict cartesian product operation $\otimes$.
The following definition is borrowed from \cite{RR10}:

\begin{defn} [Expressing $\qdom$ in $\cdom$]
\label{dfn:expressible} A qualification domain $\qdom$ is expressible in a constraint domain $\cdom$
if there is an injective embedding mapping $\imath : \aqdom \to C$ and moreover:
  \begin{enumerate}
     \item
     There is a $\cdom$-constraint $\qval{X}$ such that
     $\Solc{\qval{X}}$ is the set of all $\eta \in \mbox{Val}_\cdom$  verifying  $\eta(X) \in ran(\imath)$.
     \item
    There is a $\cdom$-constraint $\qbound{X,Y,Z}$ encoding ``$x \dleq y \circ z$'' in the following sense:
    any $\eta \in \mbox{Val}_\cdom$ such that $\eta(X) = \iota(x)$, $\eta(Y) = \iota(y)$ and $\eta(Z) = \iota(z)$
    verifies $\eta \in \Solc{\qbound{X,Y,Z}}$ iff $x \dleq y \circ z$.
  \end{enumerate}
In addition, if $\qval{X}$ and $\qbound{X,Y,Z}$ can be chosen as existential constraints
of the form $\exists X_1 \ldots \exists X_n(B_1 \land \ldots \land B_m)$---where $B_j ~ (1 \leq j \leq m)$ are atomic---we say that $\qdom$ is {\em existentially expressible} in $\cdom$. \mathproofbox
\end{defn}

It can be proved that $\B$, $\U$, $\W$ and and any qualification domain built from these with the help of $\otimes$ are existentially expressible in any constraint domain $\cdom$ that includes the basic values and computational features of $\rdom$.


{\em Admissible triples} $\langle \simrel, \qdom, \cdom \rangle$ consist of a constraint domain $\cdom$, a qualification domain $\qdom$ and a proximity relation $\simrel : S \times S \to D$---where $D$ is the carrier set of $\qdom$ and $S$ is the set of all variables, basic values  and signature symbols available in $\cdom$---satisfying the following properties:
\begin{itemize}
\item $\forall x\in S : \simrel(x,x) = \tp$ (reflexivity).
\item $\forall x,y\in S : \simrel(x,y) = \simrel(y,x)$ (symmetry).
\item Some additional technical conditions explained in \cite{RR10TR}.
\end{itemize}
A proximity relation $\simrel$ is called {\em similarity} iff it satisfies the additional property $\forall x,y,z\in S : \simrel(x,z) \dgeq \simrel(x,y) \sqcap \simrel(y,z)$ (transitivity).
The scheme SQCLP has instances $\sqclp{\simrel}{\qdom}{\cdom}$ where $\langle \simrel,\qdom,\cdom \rangle$ is an admissible triple.


A $\sqclp{\simrel}{\qdom}{\cdom}$-program is a set $\Prog$ of  \emph{qualified
program rules} (also called \emph{qualified clauses})
$C : A \qgets{\alpha} \qat{B_1}{w_1}, \ldots, \qat{B_m}{w_m}$, where $A$ is a defined atom,
$\alpha \in \aqdom$ is  called the {\em attenuation factor} of the clause and
each $\qat{B_j}{w_j} ~ (1 \le j \le m)$ is an atom $B_j$
annotated with a so-called {\em threshold value} $w_j \in \bqdom$.
The intended meaning of $C$ is as follows:
if for all $1 \leq j \leq m$ one has $\qat{B_j}{e_j}$ (meaning that $B_j$ holds with qualification value $e_j$)
for some $e_j \dgeq^? w_j$,
then $\qat{A}{d}$ (meaning that $A$ holds with qualification value $d$)
can be inferred for any $d \in \aqdom$ such that $d \dleq \alpha \circ \infi_{j = 1}^m e_j$.
By convention, $e_j \dgeq^? w_j$ means $e_j \dgeq w_j$ if $w_j ~{\neq}~?$ and is identically true otherwise.
In practice threshold values equal to `?' and attenuation values equal to $\tp$ can be omitted.

\begin{figure}[ht]
\figrule
\footnotesize\it
\renewcommand{\arraystretch}{1.4}
\begin{tabular}{rl}
& \% Book representation: book( ID, Title, Author, Lang, Genre, VocLvl, Pages ). \\
\tiny 1 & library([ book(1, `Tintin', `Herg\'e', french, comic, easy, 65), \\
\tiny 2 & $\quad$ book(2, `Dune', `F.P.~Herbert', english,\ sciFi,\ medium,\ 345), \\
\tiny 3 & $\quad$ book(3, `Kritik der reinen Vernunft', `I.~Kant', german, philosophy, difficult, 1011), \\
\tiny 4 & $\quad$ book(4, `Beim Hauten der Zwiebel', `G.~Grass', german, biography, medium, 432) ]) \\[3mm]
& \% Auxiliary predicate for computing list membership: \\
\tiny 5 & member(B, [B$\mid$\_]) \\
\tiny 6 & member(B, [\_$\mid$T]) $\gets$ member(B, T) \\[3mm]
& \% Predicates for getting the explicit attributes of a given book: \\
\tiny 7 & getId(book(ID, \_Title, \_Author, \_Lang, \_Genre, \_VocLvl, \_Pages), ID) \\
\tiny 8 & getTitle(book(\_ID, Title, \_Author, \_Lang, \_Genre, \_VocLvl, \_Pages), Title) \\
\tiny 9 & getAuthor(book(\_ID, \_Title, Author, \_Lang, \_Genre, \_VocLvl, \_Pages), Author) \\
\tiny 10 & getLanguage(book(\_ID, \_Title, \_Author, Lang, \_Genre, \_VocLvl, \_Pages), Lang) \\
\tiny 11 & getGenre(book(\_ID, \_Title, \_Author, \_Lang, Genre, \_VocLvl, \_Pages), Genre) \\
\tiny 12 & getVocLvl(book(\_ID, \_Title, \_Author, \_Lang, \_Genre, VocLvl, \_Pages), VocLvl) \\
\tiny 13 & getPages(book(\_ID, \_Title, \_Author, \_Lang, \_Genre, \_VocLvl, Pages), Pages) \\[3mm]
& \% Function for guessing the reader level of a given book: \\
\tiny 14 & guessRdrLvl(B, basic) $\gets$ getVocLvl(B, easy), getPages(B, N), N $<$ 50 \\
\tiny 15 & guessRdrLvl(B, intermediate) $\qgets{0.8}$ getVocLvl(B, easy), getPages(B, N), N $\ge$ 50 \\
\tiny 16 & guessRdrLvl(B, basic) $\qgets{0.9}$ getGenre(B, children) \\
\tiny 17 & guessRdrLvl(B, proficiency) $\qgets{0.9}$ getVocbLvl(B, difficult), getPages(B, N), N $\ge$ 200 \\
\tiny 18 & guessRdrLvl(B, upper) $\qgets{0.8}$ getVocLvl(B, difficult), getPages(B, N), N $<$ 200 \\
\tiny 19 & guessRdrLvl(B, intermediate) $\qgets{0.8}$ getVocLvl(B, medium) \\
\tiny 20 & guessRdrLvl(B, upper) $\qgets{0.7}$ getVocLvl(B, medium) \\[3mm]
& \% Function for answering a particular kind of user queries: \\
\tiny 21 & search(Lang, Genre, Level, Id)  $\gets$ library(L)\#1.0, member(B, L)\#1.0, \\
\tiny 22 & $\quad$ getLanguage(B, Lang), getGenre(B, Genre), \\
\tiny 23 & $\quad$ guessRdrLvl(B, Level), getId(B, Id)\#1.0 \\[3mm]
& \% Proximity relation $\simrel_s$: \\
\tiny 24 & $\simrel_s$(sciFi, fantasy) = $\simrel_s$(fantasy, sciFi) = 0.9 \\
\tiny 25 & $\simrel_s$(adventure, fantasy) = $\simrel_s$(fantasy, adventure) = 0.7 \\
\tiny 26 & $\simrel_s$(essay, philosophy) = $\simrel_s$(philosophy, essay) = 0.8 \\
\tiny 27 & $\simrel_s$(essay, biography) = $\simrel_s$(biography, essay) = 0.7 \\
\end{tabular}
\normalfont
\caption{$\sqclp{\simrel_s}{\,\U}{\rdom}$-program $\Prog_{\!s}$ ({\em Library with books in different languages})}
\label{fig:library}
\figrule
\vspace*{-4mm}
\end{figure}

Figure \ref{fig:library} shows a simple $\sqclp{\simrel_s}{\,\U}{\rdom}$-program $\Prog_{\!s}$ which illustrates the expressivity of the SQCLP scheme to deal with problems involving flexible information retrieval. Predicate \mbox{\it search} can be used to answer queries asking for books in the library matching some desired language, genre and reader level.
Predicate \mbox{\it guessRdrLvl} takes advantage of attenuation factors to encode heuristic rules to compute reader levels on the basis of vocabulary level and other book features. The other predicates compute book features in the natural way,
and the proximity relation $\simrel_s$ allows flexibility in any unification (i.e. solving of equality constraints) arising during the invocation of the program predicates.


The declarative semantics of a given $\sqclp{\simrel}{\qdom}{\cdom}$-program $\Prog$ relies on {\em qualified constrained atoms} (briefly {\em qc-atoms}) of the form $\cqat{A}{d}{\Pi}$, intended to assert that the validity of atom  $A$ with qualification degree $d \in D$ is entailed by the constraint set $\Pi$.
A qc-atom is called {\em defined}, {\em primitive} or {\em equational} according to the syntactic form of $A$; and it is called {\em observable} iff $d \in \aqdom$ and $\Pi$ is satisfiable.

Program interpretations are defined as sets of observable qc-atoms which obey a natural closure condition.
The results proved in \cite{RR10} show two equivalent ways to characterize declarative semantics, using a fix-point approach and a proof-theoretical approach, respectively.
For the purposes of the present paper it suffices to consider the proof theoretical approach,
that relies on a formal inference system called {\em Proximity-based Qualified Constrained Horn Logic}---in symbols, $\SQCHL(\simrel,\qdom,\cdom)$---intended to infer observable qc-atoms from $\Prog$ and consisting of the three inference rules displayed in Figure \ref{fig:sqchl}. Rule {\bf SQEA} depends on a relation $\approx_{d,\Pi}$ between terms that is defined in the following way: $t \approx_{d,\Pi} s$ iff there exist two terms $\hat{t}$ and $\hat{s}$ such that
$\Pi \model{\cdom} t == \hat{t}$,
$\Pi \model{\cdom} s == \hat{s}$ and
$\bt \neq d \dleq \simrel(\hat{t},\hat{s})$.
This allows to deduce equations from $\Pi$ in a flexible way,
taking the proximity relation $\simrel$ into account.
The reader is referred to \cite{RR10TR} for more motivating comments on $\SQCHL(\simrel,\qdom,\cdom)$
and some technical properties of the $\approx_{d,\Pi}$ relation.


\begin{figure}[ht]
  \figrule
  \centering
  \vspace{2mm}
  \begin{tabular}{l}
  \textbf{SQDA} ~ $\displaystyle\frac
    {~ (~ \cqat{(t'_i == t_i\theta)}{d_i}{\Pi} ~)_{i=1 \ldots n} \quad (~ \cqat{B_j\theta}{e_j}{\Pi} ~)_{j=1 \ldots m} ~}
    {\cqat{p'(\ntup{t'}{n})}{d}{\Pi}}$ \\ \\
    \hspace{13mm} if $(p(\ntup{t}{n}) \qgets{\alpha} \qat{B_1}{w_1}, \ldots, \qat{B_m}{w_m}) \in \Prog$\!,~
        $\theta$ subst.,~ $\simrel(p',p) = d_0 \neq \bt$, \\
    \hspace{13mm} $e_j \dgeq^? w_j ~ (1 \le j \le m)$ and
        $d \dleq \bigsqcap_{i = 0}^{n}d_i \sqcap \alpha \circ \bigsqcap_{j = 1}^m e_j$. \\
  \\
  \textbf{SQEA} ~ $\displaystyle\frac
    {}
    {\quad \cqat{(t == s)}{d}{\Pi} \quad}$ ~
  if $t \approx_{d, \Pi} s$.
  ~~ \textbf{SQPA} ~ $\displaystyle\frac
    {}
    {\quad \cqat{\kappa}{d}{\Pi} \quad}$ ~
  if $\Pi \model{\cdom} \kappa$. \\
  \end{tabular}
  \vspace{2mm}
  \caption{Proximity-based Qualified Constrained Horn Logic}
  \label{fig:sqchl}
  \figrule
  \vspace{-3mm}
\end{figure}

We will write $\Prog \sqchlrdc \varphi$ to indicate that $\varphi$ can be deduced from $\Prog$ in
$\SQCHL(\simrel,$ $\qdom,\cdom)$, and $\Prog \sqchlrdcn{k} \varphi$ in the case that the deduction can be performed with exactly $k$ {\bf SQDA} inference steps.
As usual in formal inference systems, $\SQCHL(\simrel,\qdom,\cdom)$ proofs can be represented as {\em proof trees} whose nodes correspond to qc-atoms, each node being inferred from its children by means of some $\SQCHL(\simrel,\qdom,\cdom)$ inference step.
The following theorem, proved in \cite{RR10TR}, characterizes least program models in the scheme SQCLP.
This result allows to use $\SQCHL(\simrel,\qdom,\cdom)$-derivability as a logical criterion for proving the semantic correctness of program transformations, as we will do in Section \ref{sec:implemen}.

\begin{thm}[Logical characterization of least program models in SQCHL]
\label{thm:SQCHL-leastmodel} For any $\sqclp{\simrel}{\qdom}{\cdom}$-program $\Prog$, its least
model can be characterized as: $$\Mp = \{\varphi \mid \varphi \mbox{ is an observable defined
qc-atom and }\Prog \sqchlrdc \varphi\}. \enspace \mathproofbox$$
\end{thm}

%


Let us now discuss goals and their solutions.
Goals for a given $\sqclp{\simrel}{\qdom}{\cdom}$-program $\Prog$ have the form
$$G ~:~ \qat{A_1}{W_1},~ \ldots,~ \qat{A_m}{W_m} \sep W_1 \dgeq^? \!\beta_1,~ \ldots,~ W_m \dgeq^? \!\beta_m$$
abbreviated as $(\qat{A_i}{W_i},~ W_i \dgeq^? \!\beta_i)_{i = 1 \ldots m}$.
The $\qat{A_i}{W_i}$ are called {\em annotated atoms}.
The pairwise different variables $W_i \in \War$ are called qualification variables;
they are taken from a  set $\War$ assumed to be disjoint from the set $\Var$ of data variables used in terms.
The conditions $W_i \dgeq^? \!\beta_i$ (with $\beta_i \in \bqdom$)
are called {\em threshold conditions} and their intended meaning (relying on the notations `?' and `$\dgeq^?$') is as already explained when introducing
program clauses above.
In the sequel, $\warset{o}$ will denote the set of all qualification variables occurring in the syntactic object $o$. In particular, for a goal $G$ as displayed above, $\warset{G}$ denotes the set $\{W_i \mid 1 \leq i \leq m\}$.
In the case $m = 1$ the goal is called  {\em atomic}.
The following definition relies on  $\SQCHL(\simrel,\qdom,\cdom)$-derivability to provide a natural
declarative notion of goal solution:


\begin{defn}[Goal Solutions]
\label{dfn:goalsol}
Assume a given $\sqclp{\simrel}{\qdom}{\cdom}$-program $\Prog$ and a goal $G$ for $\Prog$ with the syntax displayed above.
Then:
\begin{enumerate}
 \item
 A {\em solution} for $G$ is any triple $\langle \sigma, \mu, \Pi \rangle$ such that $\sigma$ is a $\cdom$-substitution, $W\!\mu \in \aqdom$ for all $W \in \domset{\mu}$, $\Pi$ is a satisfiable and finite set of atomic $\cdom$-constraints and the following two conditions hold for all $i = 1 \ldots m$:
$W_i\mu = d_i \dgeq^? \!\beta_i$ and $\Prog \sqchlrdc \cqat{A_i\sigma}{W_i\mu}{\Pi}$.
The set of all solutions for $G$ w.r.t. $\Prog$ is noted $\Sol{\Prog}{G}$.
\item
A solution $\langle \eta, \rho, \Pi \rangle$ for $G$ is called {\em ground} iff $\Pi = \emptyset$ and $\eta \in \mbox{Val}_\cdom$ is a variable valuation such that $A_i\eta$ is a ground atom for all $i = 1 \ldots m$.
The set of all ground solutions for $G$ w.r.t. $\Prog$ is noted $\GSol{\Prog}{G} \subseteq \Sol{\Prog}{G}$.
\item
A ground solution $\langle \eta, \rho, \emptyset \rangle \in \GSol{\Prog}{G}$ is {\em subsumed} by $\langle \sigma, \mu, \Pi \rangle$ iff there is some $\nu \in \Solc{\Pi}$ s.t. $\eta =_{\varset{G}} \sigma\nu$ and $W_i\rho \dleq W_i\mu$ for $i = 1 \ldots m$.\mathproofbox
\end{enumerate}
\end{defn}


A possible goal $G_s$ for the library program displayed in Figure \ref{fig:library} is
\begin{center}
$G_s$ ~:~ \it search(german, essay, intermediate, ID)\#W $\sep$ \!W $\ge$ 0.65
\end{center}
and one solution for $G_s$ is
$\langle \{\textit{ID} \mapsto 4\}, \{\textit{W} \mapsto 0.7\}, \emptyset \rangle$.
In this simple case, the constraint set $\Pi$ within the solution is empty.
Other examples of goal solutions can be found in \cite{RR10TR} and Sections \ref{sec:implemen} and \ref{sec:practical} below.


In practice, users of SQCLP languages will rely on  some available {\em goal solving system} for computing goal solutions.
The following definition specifies two important abstract properties of goal solving systems
which will be taken as a reference for the implementation presented in this paper.

\begin{defn}[Correct Abstract Goal Solving Systems]
\label{dfn:goalsolsys}
An {\em abstract goal solving system} for $\sqclp{\simrel}{\qdom}{\cdom}$ is any device that takes a program $\Prog$ and a goal $G$ as input and yields various triples $\langle \sigma, \mu, \Pi \rangle$, called {\em computed answers}, as outputs.
Such a goal solving system is called:
\begin{enumerate}
\item
{\em Sound} iff every computed answer is a solution $\langle \sigma, \mu, \Pi \rangle \in \mbox{Sol}_\Prog(G)$.
\item
{\em Weakly complete} iff every ground solution $\langle \eta, \rho, \emptyset \rangle \in \mbox{GSol}_\Prog(G)$ is subsumed by some computed answer.
\item
{\em Correct} iff it is both sound and weakly complete. \mathproofbox
\end{enumerate}
\end{defn}


Every goal solving system for a SQCLP instance should be sound and ideally also weakly complete.
In principle, goal solving systems with these properties for extensions of the classical LP paradigm can be formalized as extensions of the well-known SLD-resolution method \cite{Llo87,Apt90}.
A sound and complete extensions of SLD-resolution for the CLP scheme can be found e.g. in \cite{JMM+98},
and several extensions of SLD resolution for  LP languages  aiming at uncertain reasoning SQCLP scheme have been mentioned in Section \ref{sec:introduction}.


Our aim in this paper is to present an implementation based on a semantically correct program transformation from SQCLP into CLP, rather than developing a sound and complete extension of SLD resolution.
Nevertheless, both our implementation and SLD-based approaches for SLP languages in the line of \cite{Ses02} must share the ability to solve unification problems w.r.t. to a proximity relation $\simrel : S \times S \to [0,1]$ over signature symbols, which is assumed to be transitive in \cite{Ses02} but not in our setting.
The lack of transitivity makes a crucial difference w.r.t. the behavior of unification algorithms.
In the rest of this section we briefly discuss the problem by means of a simple example.

\cite{Ses02} presents a flexible unification algorithm for solving unification problems represented as systems of the form $S\sep\alpha$, where $S$ is a set of equations between terms and $\alpha$ is a certainty degree.
A solution of such a system is any substitution $\theta$ which verifies $\simrel(s\theta, t\theta) \geq \alpha$ for all equations $s == t$ belonging to $S$.
This notion of solution is consistent with the declarative semantics of the SQCLP scheme
(more specifically, with Definition \ref{dfn:goalsol}), even in the case that $\simrel$ is a non-transitive proximity relation.
Following a traditional approach, Sessa presents the flexible unification algorithm as set of transformation rules which convert systems $S\sep\alpha$ into solved form systems which represent unifiers.
The transformations are similar to those presented in e.g. Section 4.6 of \cite{BN98} for the case of classical syntactic unification,
extended with suitable computations to update $\alpha$ during the process, taking the given similarity relation $\simrel$ into account.
One of the transformations allows to transform a system of the form $X == t, S \sep \alpha$ into $S\{X \mapsto t\}\sep \alpha$ (provided that $X$ is not identical to $t$ and does not occur in $t$, the so-called occurs check). Unfortunately, this transformation can lose solutions in case that $\simrel$ is not transitive. Consider for instance the following example:

\begin{exmp}
\label{exmp:incomplete-sld}
Assume constants $a$, $b$, $c$ and a non-transitive proximity relation $\simrel$ such that
$\simrel(a,b) = \simrel(b,a) = 0.7$; $\simrel(a,c) = \simrel(c,a) = 0.8$; $\simrel(b,c) = \simrel(c,b) = 0$.
Then, the substitution $\theta = \{X \mapsto a\}$ is obviously a solution of the unification problem
$X == b$, $X == c \sep 0.7$.
Nevertheless, the unification algorithm presented in \cite{Ses02} and related papers fails without computing any solution:
$$X==b,\,X==c \sep 0.7 ~\Longrightarrow~ X==c\ \{X \mapsto b\} \sep 0.7 ~\Longrightarrow~ f\!ail \enspace $$
The second transformation step leads to {\em fail} because $X==c\ \{X \mapsto b\} \sep 0.7$ is the same as
$b==c \sep 0.7$ and $\simrel(b,c) = 0 < 0.7$.
Should $\simrel$ satisfy transitivity, then $\simrel(b,c) = \simrel(c,b) \geq 0.7$, and Sessa's unification algorithm would compute the unifier $\sigma = \{X \mapsto b\}$ as follows:
$$X==b,\,X==c \sep 0.7 ~\Longrightarrow~ X==c\ \{X \mapsto b\} \sep 0.7 ~\Longrightarrow~ \{X \mapsto b\} \sep 0.7$$
Note that $\sigma$ is more general than $\theta$ in the sense that $\simrel(\theta,\sigma\theta) = \simrel(\theta,\sigma) \geq 0.7$.
Therefore this example does not contradict the completeness of Sessa's unification algorithm for the case of (transitive) similarity relations. \mathproofbox
\end{exmp}

Even in the case that $\simrel$ is transitive, we have found examples showing that a goal solving system based on Sessa's unification algorithm can fail to compute some valid solutions for $\sqclp{\simrel}{\qdom}{\cdom}$-programs whose clauses use attenuation factors other than $\tp$.
The unification algorithm underlaying the implementations presented in Section \ref{sec:practical}---based on the program transformations from Section \ref{sec:implemen}---avoids the problematic transformation step $X == t, S \sep \alpha \Longrightarrow S\{X \mapsto t\}\sep \alpha$, that might cause incompleteness;
instead, Prolog's backtracking is used to implement the effect of a non-deterministic choice between
 several transformation steps $X == c(\ntup{t}{n}),S\sep\alpha \Longrightarrow X_1==t_1, \ldots, X_n==t_n,S\mu\sep\alpha$, where $X_1, \ldots, X_n$ are fresh variables and $\mu = \{X\mapsto c'(\ntup{X}{n})\}$ for some possible choice of $c'$ such that $\simrel(c,c') \geq \alpha$.

As an optimization, our prototype system allows the user to use a directive whose effect is that the system avoids the backtracking search just discussed and implements just the effect of the transformation $X == t, S \sep \alpha \Longrightarrow S\{X \mapsto t\}\sep \alpha$.
When including this directive, the user runs the risk of losing some valid solutions.
We conjecture that no incompleteness occurs in the case of $\sqclp{\simrel}{\qdom}{\cdom}$-programs based on a transitive $\simrel$ and whose clauses do not use attenuation factors other than $\tp$; i.e. SLP programs enriched with constraint solving.

%% file: J3_0.tex
\section{The Schemes QCLP \& CLP as Specializations of SQCLP}
\label{sec:qclpclp}

As discussed in the concluding section of \cite{RR10}, 
several specializations of the SQCLP scheme can be obtained by partial instantiation of its parameters.
In particular, QCLP and CLP can be defined as schemes with instances:
$$
\begin{array}{r@{\hspace{1mm}}c@{\hspace{1mm}}l}
\qclp{\qdom}{\cdom} &\eqdef& \sqclp{\sid}{\qdom}{\cdom}\\
\clp{\cdom} &\eqdef& \sqclp{\sid}{\B}{\cdom} = \qclp{\B}{\cdom}\\
\end{array}
$$
where $\sid$ is the {\em identity} proximity relation
and $\B$ is the qualification domain including just the two classical boolean values.
As explained in the introduction, QCLP and CLP are the targets of the two program transformations to 
be developed in Section \ref{sec:implemen}.
In this  brief section we provide an explicit description of the syntax and semantics of these two schemes,
derived from their behavior as specializations of SQCLP.

\input{J3_1} 
\input{J3_2} 

%% file: J3_1.tex
\subsection{Presentation of the QCLP Scheme}
\label{sec:cases:qclp}


As already explained, the instances of QCLP can be defined by the equation QCLP($\qdom$,$\cdom$) = SQCLP($\sid$,$\qdom$,$\cdom$).
Due to the admissibility of the parameter triple $\langle \sid, \qdom, \cdom \rangle$, the qualification domain $\qdom$ must be (existentially) expressible in the constraint domain $\cdom$.
Technically, the QCLP scheme can be seen as a common extension of the classical CLP scheme for Constraint Logic Programming \cite{JL87,JMM+98} and the QLP scheme for Qualified Logic Programming originally introduced in \cite{RR08}.
Intuitively,  QCLP programming behaves like SQCLP programming, except that proximity information other than the identity is not available for proving equalities.


Program clauses and observable qc-atoms in QCLP are defined in the same way as in  SQCLP.
The library program $\Prog_{\!s}$ in Figure \ref{fig:library} becomes a $\qclp{\,\U}{\rdom}$-program $\Prog'_{\!s}$ just by replacing $\sid$ for $\simrel$.
Of course, $\Prog'_{\!s}$ does not support flexible unification as it was the case with $\Prog_{\!s}$.


As explained in Section \ref{sec:sqclp}, the proof system consisting of the three displayed in Figure \ref{fig:sqchl} characterizes the declarative semantics of
a given $\sqclp{\simrel}{\qdom}{\cdom}$-program $\Prog$. In the particular case $\simrel = \sid$, the inference rules specialize to those  displayed in Figure \ref{fig:qchl}, yielding a formal proof system called \emph{Qualified Constrained Horn Logic} -- in symbols, $\QCHL(\qdom,\cdom)$ -- 
which characterizes the declarative semantics of a given $\qclp{\qdom}{\cdom}$-program $\Prog$.
Note that rule {\bf SQEA} depends on a relation $\approx_{\Pi}$ between terms that
is defined to behave the same as the specialization of $\approx_{d,\Pi}$ to the case $\simrel = \sid$.
It is easily checked that $t \approx_{\Pi} s$ does not depend on $d$ and holds iff $\Pi \model{\cdom} t == s$.
Both $\approx_{d,\Pi}$ and $\approx_{\Pi}$ allow to use the constraints within $\Pi$ when 
deducing equations. However, $c(\ntup{t}{n}) \approx_{\Pi} c'(\ntup{s}{n})$ never holds in the case that $c$ and $c'$ are not syntactically identical.

 \begin{figure}[h]
  \figrule
  \centering
  \begin{tabular}{l}\\
  \textbf{QDA} ~ $\displaystyle\frac
    {~ (~ \cqat{(t'_i == t_i\theta)}{d_i}{\Pi} ~)_{i = 1 \ldots n} \quad (~ \cqat{B_j\theta}{e_j}{\Pi} ~)_{j = 1 \ldots m} ~}
    {\cqat{p(\ntup{t'}{n})}{d}{\Pi}}$ \\ \\
  \hspace{11mm} if $(p(\ntup{t}{n}) \qgets{\alpha} \qat{B_1}{w_1}, \ldots, \qat{B_m}{w_m}) \in \Prog$, $\theta$ subst., \\
  \hspace{11mm} $e_j \dgeq^? w_j ~ (1 \le j \le m)$ and $d \dleq \bigsqcap_{i = 1}^{n}d_i \sqcap \alpha \circ \bigsqcap_{j = 1}^m e_j$.\\ 
  \\
  \textbf{QEA} ~ $\displaystyle\frac
    {}
    {\quad \cqat{(t == s)}{d}{\Pi} \quad}$ ~
  if $t \approx_{\Pi} s$.
  \hspace{8mm} \textbf{QPA} ~ $\displaystyle\frac
    {}
    {\quad \cqat{\kappa}{d}{\Pi} \quad}$ ~
  if $\Pi \model{\cdom} \kappa$. \\
  \\
  \end{tabular}
  \caption{Qualified Constrained Horn Logic}
  \label{fig:qchl}
  \figrule
\end{figure}

$\SQCHL(\simrel,\qdom,\cdom)$ proof trees and the notations related to them can be naturally specialized to $\QCHL(\qdom,\cdom)$.
In particular, we will use the notation $\Prog \qchldc \varphi$ (resp. $\Prog \qchldcn{k} \varphi$)
to indicate that the qc-atom $\varphi$ can be inferred in $\QCHL(\qdom,\cdom)$ from the program $\Prog$
(resp. it can be inferred by using exactly $k$ \textbf{QDA} inference steps).
Theorem \ref{thm:SQCHL-leastmodel} also specializes to QCHL, yielding the following result:

\begin{thm}[Logical characterization of least program models in QCHL]
\label{thm:QCHL-leastmodel} For any $\qclp{\qdom}{\cdom}$-program $\Prog$, its least
model can be characterized as: $$\Mp = \{\varphi \mid \varphi \mbox{ is an observable defined
qc-atom and }\Prog \qchldc \varphi\}. \enspace \mathproofbox$$
\end{thm}


Concerning goals and their solutions, their specialization to the particular case $\simrel = \sid$ leaves the syntax of goals $G$ unaffected and leads to the following definition, almost identical to Definition \ref{dfn:goalsol}:

\begin{defn}[Goal Solutions in QCLP]
\label{dfn:qclp-goalsol} Assume a given $\qclp{\simrel}{\qdom}{\cdom}$-program $\Prog$ and a goal
$G: (~ \qat{A_i}{W_i}, W_i \dgeq^? \!\beta_i ~)_{i = 1 \ldots m}$. Then:
\begin{enumerate}
 \item
 A {\em solution} for $G$ is any triple $\langle \sigma, \mu, \Pi \rangle$ such that
 $\sigma$ is a $\cdom$-substitution, $W\mu \in \aqdom$  for all $W \in \domset{\mu}$,
 $\Pi$ is a satisfiable and finite set of atomic $\cdom$-constraints,
 and the following two conditions hold for all $i = 1 \ldots m$:
$W_i\mu = d_i \dgeq^? \!\beta_i$ and $\Prog  \qchldc \cqat{A_i\sigma}{W_i\mu}{\Pi}$.
The set of all solutions for $G$ is noted $\Sol{\Prog}{G}$.
\item
A solution $\langle \eta, \rho, \Pi \rangle$ for $G$ is called {\em ground} iff $\Pi = \emptyset$ and
$\eta \in \mbox{Val}_\cdom$ is a variable valuation such that $A_i\eta$ is a ground atom for all $i = 1 \ldots m$.
The set of all ground solutions for $G$ is noted $\GSol{\Prog}{G} \subseteq \Sol{\Prog}{G}$.
\item
A ground solution $\langle \eta, \rho, \emptyset \rangle \in \GSol{\Prog}{G}$ is {\em subsumed} by
$\langle \sigma, \mu, \Pi \rangle$ iff
there is some $\nu \in \Solc{\Pi}$ s.t. $\eta =_{\varset{G}} \sigma\nu$ and $W_i\rho \dleq W_i\mu$ for $i = 1 \ldots m$.
\mathproofbox
\end{enumerate}
\end{defn}


Finally, the notion  of correct abstract goal solving system for $\mbox{SQCLP}$ given in Definition
\ref{dfn:goalsolsys} specializes to $\mbox{QCLP}$ without any formal change. Therefore, we state no
new definition at this point.

%% file: J3_2.tex
\subsection{Presentation of the CLP Scheme}
\label{sec:cases:clp}


As already explained, the instances of CLP can be defined by the equation CLP($\cdom$) = SQCLP($\sid,\B,\cdom$), 
or equivalently, CLP($\cdom$) = QCLP($\B,\cdom$).
Due to the fixed  choice $\qdom = \B$, the only qualification value $d \in \aqdom$ available for use as attenuation factor or threshold value is $d = \tp$.
Therefore, CLP can only include threshold values equal to `?' and attenuation values equal to the top element $\tp = true$ of $\B$.
As explained in Section \ref{sec:sqclp}, such trivial threshold and attenuation values can be omitted,
and CLP clauses can  be written with the simplified syntax  $A \gets B_1, \ldots, B_m$.


Since $\tp = true$ is the only non-trivial qualification value available in CLP, qc-atoms $\cqat{A}{d}{\Pi}$ 
are always of the form $\cqat{A}{true}{\Pi}$ and can be written as $\cat{A}{\Pi}$.
Moreover, all the side conditions for the inference rule {\bf QDA} in Figure \ref{fig:qchl}
become trivial when specialized to the case $\qdom = \B$. 
Therefore, the specialization of $\QCHL(\qdom,\cdom)$ to the case $\qdom = \B$ leads to the
formal proof system called \emph{Constrained Horn Logic} -- in symbols, $\CHL(\cdom)$ --
consisting of the three inference rules displayed in Figure \ref{fig:chl},
which characterizes the declarative semantics of a given $\clp{\cdom}$-program $\Prog$.

\begin{figure}[h]
  \figrule
  \centering
  \begin{tabular}{l}\\
  \textbf{DA} ~ $\displaystyle\frac
    {~ (~ \cat{(t'_i == t_i\theta)}{\Pi} ~)_{i = 1 \ldots n} \quad (~ \cat{B_j\theta}{\Pi} ~)_{j = 1 \ldots m} ~}
    {\cat{p(\ntup{t'}{n})}{\Pi}}$ \\ \\
  \hspace{9mm} if $(p(\ntup{t}{n}) \gets B_1, \ldots, B_m) \in \Prog$ and $\theta$ subst. \\
   \\
  \textbf{EA} ~ $\displaystyle\frac
    {}
    {\quad \cat{(t == s)}{\Pi} \quad}$ ~
  if $t \approx_{\Pi} s$.
  \hspace{10mm} \textbf{PA} ~ $\displaystyle\frac
    {}
    {\quad \cat{\kappa}{\Pi} \quad}$ ~
  if $\Pi \model{\cdom} \kappa$. \\
  \\
  \end{tabular}
  \caption{Constrained Horn Logic}
  \label{fig:chl}
  \figrule
\end{figure}

$\QCHL(\qdom,\cdom)$ proof trees and the notations related to them can be naturally specialized to $\CHL(\cdom)$.
In particular, we will use the notation $\Prog \chl \varphi$ (resp. $\Prog \chl{k} \varphi$)
to indicate that the qc-atom $\varphi$ can be inferred in $\CHL(\cdom)$ from the program $\Prog$
(resp. it can be inferred by using exactly $k$ \textbf{DA} inference steps).
Theorem \ref{thm:QCHL-leastmodel} also specializes to CHL, yielding the following result:

\begin{thm}[Logical characterization of least program models in CHL]
\label{thm:CHL-leastmodel} For any $\clp{\cdom}$-program $\Prog$, its least
model can be characterized as: $$\Mp = \{\varphi \mid \varphi \mbox{ is an observable defined
qc-atom and }\Prog \chl \varphi\}. \enspace \mathproofbox$$
\end{thm}



Concerning goals and their solutions, their specialization to the scheme CLP leads to the following definition:

\begin{defn}[Goals and their Solutions in $\mbox{CLP}$]
\label{dfn:clp-goalsol}
Assume a given $\clp{\cdom}$-program $\Prog$.  Then:
\begin{enumerate}
\item
Goals for $\Prog$ have the form $G:\, A_1, \ldots, A_m$,
abbreviated as $(~ A_i ~)_{i = 1 \ldots m}$, where $A_i ~ (1 \leq i \leq m)$ are atoms.
 \item
 A {\em solution} for a goal $G$ is any pair $\langle \sigma, \Pi \rangle$ such that
 $\sigma$ is a $\cdom$-substitution,
 $\Pi$ is a satisfiable and finite set of atomic $\cdom$-constraints,
 and $\Prog \chlc \cat{A_i\sigma}{\Pi}$ holds for all $i = 1 \ldots m$.
The set of all solutions for $G$ is noted $\Sol{\Prog}{G}$.
\item
A solution $\langle \eta, \Pi \rangle$ for $G$ is called {\em ground} iff $\Pi = \emptyset$ and
$\eta \in \mbox{Val}_\cdom$ is a variable valuation such that $A_i\eta$ is a ground atom for all $i = 1 \ldots m$.
The set of all ground solutions for $G$ is noted $\GSol{\Prog}{G}$.
Obviously, $\GSol{\Prog}{G} \subseteq \Sol{\Prog}{G}$.
\item
A ground solution $\langle \eta, \emptyset \rangle \in \GSol{\Prog}{G}$ is {\em subsumed} by
$\langle \sigma, \Pi \rangle$ iff there is some $\nu \in \Solc{\Pi}$ s.t.
$\eta =_{\varset{G}} \sigma\nu$. \mathproofbox
\end{enumerate}
\end{defn}

The notion  of correct abstract goal solving system for $\mbox{SQCFLP}$ given in Definition \ref{dfn:goalsolsys}
specializes to $\mbox{CLP}$ with only minor  formal changes, as follows:

\begin{defn}[Correct Abstract Goal Solving Systems for $\mbox{CLP}$]
\label{dfn:clp-goalsolsys}
A {\em goal solving system} for $\clp{\cdom}$
is any effective procedure which takes a program $\Prog$ and a goal $G$ as input
and yields various pairs $\langle \sigma, \Pi \rangle$,
called {\em computed answers}, as outputs.
Such a goal solving system is called:
\begin{enumerate}
 \item
 {\em Sound} iff every computed answer  is a solution $\langle \sigma,  \Pi \rangle \in \mbox{Sol}_\Prog(G)$.
 \item
 {\em Weakly complete} iff every ground solution $\langle \eta,  \emptyset \rangle \in \mbox{GSol}_\Prog(G)$
 is subsumed by some computed answer.
 \item
 {\em Correct} iff it is both sound and weakly complete.
\enspace \mathproofbox
\end{enumerate}
\end{defn}


We close this Subsection with a technical lemma that will be useful for proving some results in Subsection \ref{sec:implemen:QCLP2CLP}:

\begin{lem}
\label{lema:dec}
Assume an existential $\cdom$-constraint $\pi(\ntup{X}{n}) = \exists Y_1\ldots\exists Y_k(B_1 \land \ldots \land B_m)$
with free variables $\ntup{X}{n}$ and a given $\clp{\cdom}$-program $\Prog$ including the clause
$C:\, p(\ntup{X}{n}) \gets B_1, \ldots, B_m$, where $p \in DP^n$ does not occur at the head of any other clause of $\Prog$.
Then, for any n-tuple $\ntup{t}{n}$ of $\cdom$-terms and any finite and satisfiable $\Pi \subseteq \Con{\cdom}$,
one has: 
\begin{enumerate}
\item
$\Prog \chlc (\cat{p(\ntup{t}{n})}{\Pi}) \Longrightarrow \Pi \models_{\cdom} \pi(\ntup{t}{n})$,
where $\pi(\ntup{t}{n})$ stands for the result of applying the substitution $\{\ntup{X}{n} \mapsto \ntup{t}{n}\}$
to $\pi(\ntup{X}{n})$.
\item
The opposite implication
$\Pi \models_{\cdom} \pi(\ntup{t}{n}) \Longrightarrow  \Prog \chlc (\cat{p(\ntup{t}{n})}{\Pi})$
holds if $\ntup{t}{n}$ is a ground term tuple.
Note that for ground $\ntup{t}{n}$ the constraint entailment $\Pi \models_{\cdom} \pi(\ntup{t}{n})$
simply means that $\pi(\ntup{t}{n})$ is true in $\cdom$.
\item 
$\Pi \models_{\cdom} \pi(\ntup{t}{n}) \Longrightarrow  \Prog \chlc (\cat{p(\ntup{t}{n})}{\Pi})$
may fail if $\ntup{t}{n}$ is not a ground term tuple.
\end{enumerate}
\end{lem}
\begin{proof*}
We prove each item separately:
\begin{enumerate}
\item
Assume $\Prog \chlc (\cat{p(\ntup{t}{n})}{\Pi})$. 
Note that $C$ is the only clause for $p$ in $\Prog$ and that each atom $B_j$ in $C$'s body is an atomic constraint.
Therefore, the $\CHL(\cdom)$ proof must use a \textbf{DA} step based on an instance
$C\theta$ of clause $C$ such that $\Pi \models_{\cdom} t_i == X_i\theta$ holds for all $1 \leq i \leq n$ and 
$\Pi \models B_j\theta$ holds for all $1 \leq j \leq m$.
These conditions and the syntactic form of $\pi(\ntup{X}{n})$ obviously imply $\Pi \models_{\cdom} \pi(\ntup{t}{n})$.
\item
Assume now $\Pi \models_{\cdom} \pi(\ntup{t}{n})$ and $\ntup{t}{n}$ ground. 
Then $\pi(\ntup{t}{n})$ is true in $\cdom$, and 
due to the syntactic form of $\pi(\ntup{X}{n})$, there must be some substitution $\theta$ such that
$X_i\theta = t_i$ (syntactic identity) for all $1 \leq i \leq n$ and
$B_j\theta$ is ground and true in $\cdom$ for all $1 \leq j \leq m$.
Trivially, $\Pi \models_{\cdom} t_i == X_i\theta$ holds for all  $1 \leq i \leq n$ and
$\Pi \models_{\cdom} B_j\theta$ also holds for all $1 \leq j \leq m$.
Then, it is obvious that $\Prog \chlc (\cat{p(\ntup{t}{n})}{\Pi})$ can be proved 
by using a  \textbf{DA} step based on the instance $C\theta$ of clause $C$.
\item
We prove that $\Pi \models_{\cdom} \pi(\ntup{t}{n}) \Longrightarrow  \Prog \chlc (\cat{p(\ntup{t}{n})}{\Pi})$
can fail if $\ntup{t}{n}$ is not ground by presenting a counterexample based on the constraint domain $\rdom$, using the syntax for $\rdom$-constraints explained in \cite{RR10TR}.
Consider the existential $\rdom$-constraint $\pi(X) = \exists Y(op_{+}(Y,Y,X))$,
and a $\clp{\rdom}$-program $\Prog$ including the clause $C:\, p(X) \gets op_{+}(Y,Y,X)$ and no other 
occurrence of the defined predicate symbol $p$. Consider also $\Pi = \{cp_{\geq}(X,0.0)\}$ and $t = X$.
Then  $\Pi \models_{\rdom} \pi(X)$ is obviously true, because any real number $x \geq 0.0$ 
satisfies $\exists Y(op_{+}(Y,Y,x))$ in $\rdom$.
However, there is no $\rdom$-term $s$ such that $\Pi \models_{\rdom} op_{+}(s,s,X)$,
and therefore there is no instance $C\theta$ of clause $C$ that can be used to prove $\Prog \chlc (\cat{p(X)}{\Pi})$
by  applying a  \textbf{DA} step. \mathproofbox
\end{enumerate}
\end{proof*}

 

%% file: J4_0.tex
\section{Implementation by Program Transformation}
\label{sec:implemen}

The purpose of this section is to introduce a program transformation that transforms $\sqclp{\simrel}{\qdom}{\cdom}$ programs and goals into semantically equivalent $\clp{\cdom}$ programs and goals.
This transformation is performed as the composition of the two following specific transformations:
\begin{enumerate}
\item
elim$_\simrel$ --- Eliminates the proximity relation $\simrel$ of arbitrary $\sqclp{\simrel}{\qdom}{\cdom}$ programs and goals, producing equivalent $\qclp{\qdom}{\cdom}$ programs and goals.
\item
elim$_\qdom$ --- Eliminates the qualification domain $\qdom$ of arbitrary $\qclp{\qdom}{\cdom}$ programs and goals, producing equivalent $\clp{\cdom}$ programs and goals.
\end{enumerate}

Thus, given a $\sqclp{\simrel}{\qdom}{\cdom}$-program $\Prog$---resp. $\sqclp{\simrel}{\qdom}{\cdom}$-goal $G$---, the composition of the two transformations will produce an equivalent $\clp{\cdom}$-program $\elimD{\elimS{\Prog}}$---resp. $\clp{\cdom}$-goal $\elimD{\elimS{G}}$---.

\begin{exmp}[Running example: $\sqclp{\simrel_r}{\,\U{\otimes}\W}{\rdom}$-program $\Prog_r$]
\label{exmp:pr}
As a running example for this section, consider the $\sqclp{\simrel_r}{\,\U{\otimes}\W}{\rdom}$-program $\Prog_r$ as follows:
\begin{center}
\footnotesize\it
\renewcommand{\arraystretch}{1.4}
\begin{tabular}{rl}
\tiny $R_1$ & famous(sha) $\qgets{(0.9,1)}$ \\
\tiny $R_2$ & wrote(sha, kle) $\qgets{(1,1)}$ \\
\tiny $R_3$ & wrote(sha, hamlet) $\qgets{(1,1)}$ \\
\tiny $R_4$ & good\_work(G) $\qgets{(0.75,3)}$ famous(A)\#(0.5,100), authored(A, G) \\[3mm]
\tiny $S_1$ & $\simrel_r$(wrote, authored) = $\simrel_r$(authored, wrote) = (0.9,0)\\
\tiny $S_2$ & $\simrel_r$(kle, kli) = $\simrel_r$(kli, kle) = (0.8,2)\\
\end{tabular}
\end{center}
where the constants $shakespeare$, $king\_lear$ and $king\_liar$ have been respectively replaced, for clarity purposes in the subsequent examples, by $sha$, $kle$ and $kli$.

In addition, consider the $\sqclp{\simrel_r}{\,\U{\otimes}\W}{\rdom}$-goal $G_r$ as follows:
\begin{center}
\it good\_work(X)\#W $\sep$ \!W $\dgeq^?$ \!\!(0.5,10)
\end{center}

We will illustrate the two transformation by showing, in subsequent examples, the program clauses of $\elimS{\Prog_r}$ and $\elimD{\elimS{\Prog_r}}$ and the goals $\elimS{G_r}$ and $\elimD{\elimS{G_r}}$. \mathproofbox
\end{exmp}

The next two subsections explain each transformation in detail.

\input{J4_1} 
\input{J4_2} 
\input{J4_3} 

%% file: J4_1.tex

\subsection{Transforming SQCLP into QCLP}
\label{sec:implemen:SQCLP2QCLP}

In this subsection we assume that the triple $\langle \simrel,\qdom,\cdom \rangle$ is admissible.
In the sequel we say that a defined predicate symbol $p \in DP^n$ is {\em affected} by a $\sqclp{\simrel}{\qdom}{\cdom}$-program $\Prog$ iff $\simrel(p,p') \neq \bt$ for some $p'\!$ occurring in $\Prog$.
We also say that an atom $A$ is {\em relevant} for $\Prog$ iff some of the three following cases hold: a) $A$ is an equation $t == s$; b) $A$ is a primitive atom $\kappa$; or c) $A$ is a defined atom $p(\ntup{t}{n})$ such that $p$ is affected by $\Prog$.

As a first step towards the definition of the first program transformation elim$_\simrel$, we define a set $EQ_\simrel$ of $\qclp{\qdom}{\cdom}$ program clauses that emulates the behavior of equations in $\sqclp{\simrel}{\qdom}{\cdom}$.
The following definition assumes that the binary predicate symbol $\sim\ \in DP^2$ (used in infix notation) and the nullary predicate symbols $\mbox{pay}_\lambda \in DP^0$ are not affected by $\Prog$\!.

\begin{defn}\label{def:EQ}
We define $EQ_\simrel$ as the following $\qclp{\qdom}{\cdom}$-program:
$$
\begin{array}{l@{\hspace{0mm}}c@{\hspace{0mm}}l}
EQ_\simrel & \eqdef & \{~ X \sim Y \qgets{\tp} \qat{(X == Y)}{?} ~\} \\
&& ~\bigcup~ \{~ u \sim u' \qgets{\tp} \qat{\mbox{pay}_{\lambda}}{?} \mid u,u' \in B_\cdom \mbox{ and } \simrel(u,u') = \lambda \neq \bt ~\} \\
&& ~\bigcup~ \{~ c(\ntup{X}{n}) \sim c'(\ntup{Y}{n}) \qgets{\tp} \qat{\mbox{pay}_{\lambda}}{?},
(~ \qat{(X_i \sim Y_i)}{?} ~)_{i = 1 \ldots n} \mid c, c' \in DC^{n} \\
&& \qquad \mbox{and } \simrel(c,c') = \lambda \neq \bt ~\}\\
&& ~\bigcup~ \{~ \mbox{pay}_\lambda \qgets{\lambda} \ \mid \mbox{ for each } \lambda \in \aqdom ~\}. \mathproofbox \\
\end{array}
$$
\end{defn}

The following lemma shows the relation between the semantics of equations in $\SQCHL(\simrel,\qdom,\cdom)$
and the behavior of the binary predicate symbol `$\sim$' defined by $EQ_\simrel$ in $\QCHL(\qdom,\cdom)$.

\begin{lem} \label{lema:equiv}
Consider any two arbitrary terms $t$ and $s$; $EQ_\simrel$ defined as in Definition \ref{def:EQ}; and a satisfiable finite set $\Pi$ of $\cdom$-constraints.
Then, for every $d \in \aqdom$:
$$t \approx_{d,\Pi} s \Longleftrightarrow EQ_\simrel \qchldc \cqat{(t \sim s)}{d}{\Pi} \enspace .$$
\end{lem}
\begin{proof*}
We separately prove each implication.

\smallskip\noindent [$\Longrightarrow$]
Assume $t \approx_{d,\Pi} s$. Then, there are two terms $\hat{t}$, $\hat{s}$ such that:
$$
(1)~ t \approx_{\Pi} \hat{t} \qquad
(2)~ s \approx_{\Pi} \hat{s} \qquad
(3)~ \hat{t} \approx_{d} \hat{s}
$$
We use structural induction on the form of the term $\hat{t}$.
\begin{itemize}
\item 
$\hat{t} = Z$, $Z \in \Var$. 
From (3) we have $\hat{s} = Z$. Then (1) and (2) become $t \approx_\Pi Z$ and $s \approx_\Pi Z$, therefore $t \approx_\Pi s$. 
Now $EQ_\simrel \qchldc \cqat{(t \sim s)}{d}{\Pi}$ can be proved with a proof tree rooted by a {\bf QDA} step of the form:
$$
\displaystyle\frac
    {~ \cqat{(t == X\theta)}{\tp}{\Pi}\qquad \cqat{(s == Y\theta)}{\tp}{\Pi} \qquad \cqat{(X==Y)\theta}{\tp}{\Pi} ~}
    {\cqat{(t \sim s)}{d}{\Pi}}
$$
using the clause $X \sim Y \qgets{\tp} \qat{(X == Y)}{?} \in EQ_\simrel$ instantiated by the substitution $\theta = \{X \!\mapsto t,\ Y \!\mapsto s \}$. Therefore the three premises can be derived from $EQ_\simrel$ with {\bf QEA} steps since $t \approx_\Pi t$, $s \approx_\Pi s$ and $t \approx_\Pi s$, respectively. Checking the side conditions of all inference steps is straightforward.

\item $\hat{t} = u$, $u \in B_{\cdom}$.
From (3) we have $\hat{s} = u'$ for some $u' \in B_\cdom$ such that $d \dleq \lambda = \simrel(u,u')$.
Then (1) and (2) become  $t \approx_\Pi u$ and $s \approx_\Pi u'$, which allow to build a proof of $EQ_\simrel \qchldc \cqat{(t \sim s)}{d}{\Pi}$ by means of a {\bf QDA} step using the clause $u \sim u' \qgets{\tp} \qat{\mbox{pay}_{\lambda}}{?}$.

\item $\hat{t} = c$, $c \in DC^0$\!. 
From (3) we have $\hat{s} = c'$ for some $c' \in DC^0$\! such that $d \dleq \lambda = \simrel(c,c')$.
Then (1) and (2) become  $t \approx_\Pi c$ and $s \approx_\Pi c'$\!, which allow us to build a proof of $EQ_\simrel \qchldc \cqat{(t \sim s)}{d}{\Pi}$ by means of a {\bf QDA} step using the clause $c \sim c' \qgets{\tp} \qat{\mbox{pay}_{\lambda}}{?}$.

\item $\hat{t} = c(\ntup{t}{n})$, $c \in DC^n$ with $n > 0$.
In this case, and because of (3), we can assume $\hat{s} = c'(\ntup{s}{n})$ for some $c' \in DC^n$ satisfying $d \dleq d_0 \eqdef \simrel(c,c')$ and $d \dleq d_i \eqdef \simrel(t_i,s_i)$ for $i=1 \dots n$.
Then $EQ_\simrel \qchldc \cqat{(t \sim s)}{d}{\Pi}$ with a proof tree rooted by a {\bf QDA} step of the form:
$$
\displaystyle\frac
  {~ 
    \begin{array}{l@{\hspace{1cm}}l}
      \cqat{(t == c(\ntup{t}{n}))}{\tp}{\Pi} & \cqat{\mbox{pay}_{d_0}}{d_0}{\Pi} \\
      \cqat{(s == c'(\ntup{s}{n}))}{\tp}{\Pi} & (~ \cqat{(t_i \sim s_i)}{d_i}{\Pi} ~)_{i = 1 \ldots n} \\
    \end{array}
   ~}
  {\cqat{(t \sim s)}{d}{\Pi}}
$$
using the $EQ_\simrel$ clause $C : c(\ntup{X}{n}) \sim c'(\ntup{Y}{n}) \qgets{\tp} \qat{\mbox{pay}_{d_0}}{?}, (\qat{(X_i \sim Y_i)}{?})_{i = 1 \ldots n}$ instantiated by the substitution $\theta = \{ X_1 \mapsto t_1,\ Y_1 \mapsto s_1,\ \dots,\ X_n \mapsto t_n,\ Y_n \mapsto s_n \}$.
Note that $C$ has attenuation factor $\tp$ and threshold values $?$ at the body.
Therefore, the side conditions of the {\bf QDA} step boil down to $d \dleq d_i ~ (1 \leq i \leq n)$ which are true by assumption.
It remains to prove that each premise of the {\bf QDA} step can be derived from $EQ_\simrel$ in QCHL($\qdom,\cdom$):
\begin{itemize}
\item $EQ_\simrel \qchldc \cqat{(t == c(\ntup{t}{n}))}{\tp}{\Pi}$ and $EQ_\simrel \qchldc \cqat{(s == c'(\ntup{s}{n}))}{\tp}{\Pi}$ are trivial consequences of $t \approx_\Pi c(\ntup{t}{n})$ and $s \approx_\Pi c'(\ntup{s}{n})$, respectively.
In both cases, the QCHL($\qdom$,$\cdom$) proofs consist of one single {\bf QEA} step.
\item $EQ_\simrel \qchldc  \cqat{\mbox{pay}_{d_0}}{d_0}{\Pi}$ can be proved using the clause $\mbox{pay}_{d_0} \!\qgets{d_0}\ \in EQ_\simrel$ in one single {\bf QDA} step.
\item $EQ_\simrel \qchldc \cqat{(t_i \sim s_i)}{d_i}{\Pi}$ for ${i = 1 \ldots n}$.
For each $i$, we observe that $t_i \approx_{d_i,\Pi} s_i$ holds because of $\hat{t}_i = t_i$, $\hat{s}_i = s_i$ which satisfy $t_i \approx_\Pi \hat{t}_i$, $s_i \approx_\Pi \hat{s}_i$ and $\hat{t}_i \approx_{d_i} \hat{s}_i$.
Since $\hat{t}_i = t_i$ is a subterm of $\hat{t} = c(\ntup{t}{n})$, the inductive hypothesis can be applied.
\end{itemize}
\end{itemize}

\smallskip\noindent [$\Longleftarrow$]
Let $T$ be a $\QCHL(\qdom,\cdom)$-proof tree witnessing $EQ_\simrel \qchldc \cqat{(t \sim s)}{d}{\Pi}$.
We prove $t \approx_{d,\Pi} s$ reasoning  by induction on the number $n = \Vert T \Vert$ of nodes in $T$ that represent conclusions of  \textbf{QDA} inference steps.
Note that all the program clauses belonging to $EQ_\simrel$ define either the binary predicate symbol `$\sim$' or the nullary predicates $\mbox{pay}_\lambda$.

\begin{description}
\item[\bf Basis ($n = 1$).] \hfill

In this case we have for the {\bf QDA} inference step that there can be used three possible $EQ_\simrel$ clauses:
\begin{enumerate}
\item
The program clause is  $X \sim Y \qgets{\tp} \qat{(X == Y)}{?}$.
Then the {\bf QDA} inference step must be of the form:
$$
\displaystyle\frac
  {~ \cqat{(t == t')}{d_1}{\Pi} \quad  \cqat{(s == s')}{d_2}{\Pi} \quad \cqat{(t' == s')}{e_1}{\Pi} ~}
  {\cqat{(t \sim s)}{d}{\Pi}}
$$
with $d \dleq  d_1 \sqcap d_2 \sqcap e_1$.
The proof of the three premises must use the {\bf QEA} inference rule.
Because of the conditions of this inference rule we have $t \approx_\Pi t'$,  $s \approx_\Pi s'$ and $ t' \approx_\Pi s'$. Therefore $t \approx_\Pi s$ is clear. Then $t \approx_{d,\Pi} s$ holds by taking $\hat{t} = \hat{s} = t$ because, trivially, $t \approx_\Pi \hat{t}$, $s \approx_\Pi \hat{s}$ and $\hat{t} \approx_d \hat{s}$.
\item
The program clause is $u \sim u' \qgets{\tp} \qat{\mbox{pay}_{\lambda}}{?}$ with $u, u' \in B_\cdom$ such that $\simrel(u,u') = \lambda \neq \bt$.
The {\bf QDA} inference step must be of the form:
$$
\displaystyle\frac
    {~ \cqat{(t == u)}{d_1}{\Pi} \quad  \cqat{(s == u')}{d_2}{\Pi} \quad \cqat{\mbox{pay}_{\lambda}}{e_1}{\Pi}~ }
    {\cqat{(t \sim s)}{d}{\Pi}}
$$
with $d \dleq d_1 \sqcap d_2 \sqcap  e_1$.
Due to the forms of the {\bf QEA} inference rule and the $EQ_\simrel$ clause $\mbox{pay}_{\lambda} \qgets{\lambda}$, we can assume without loss of generality that $d_1 = d_2 = \tp$ and $e_1 = \lambda$.
Therefore $d \dleq \lambda$.
Moreover, the QCHL($\qdom$,$\cdom$) proofs of the first two premises must use {\bf QEA} inferences.
Consequently we have $t \approx_\Pi u$ and $s \approx_\Pi u'$.
These facts and $u \approx_d u'$ imply $t \approx_{d,\Pi} s$.
\item
The program clause is $c \sim c' \qgets{\tp} \qat{\mbox{pay}_{\lambda}}{?}$ with $c, c' \in DC^0$ such that $\simrel(c,c') = \lambda \neq \bt$.
The {\bf QDA} inference step must be of the form:
$$
\displaystyle\frac
    {~ \cqat{(t == c)}{d_1}{\Pi} \quad  \cqat{(s == c')}{d_2}{\Pi} \quad \cqat{\mbox{pay}_{\lambda}}{e_1}{\Pi} ~}
    {\cqat{(t \sim s)}{d}{\Pi}}
$$
with $d \dleq d_1 \sqcap d_2 \sqcap  e_1$. Due to the forms of the {\bf QEA} inference rule and the $EQ_\simrel$ clause $\mbox{pay}_{\lambda} \qgets{\lambda}$, we can assume without loss of generality that $d_1 = d_2 = \tp$ and $e_1 = \lambda$.
Therefore $d \dleq \lambda$.
Moreover, the QCHL($\qdom$,$\cdom$) proofs of the first two premises must use {\bf QEA} inferences.
Consequently we have $t \approx_\Pi c$ and $s \approx_\Pi c'$.
These facts and $c \approx_d c'$ imply $t \approx_{d,\Pi} s$.
\end{enumerate}

\item[\bf Inductive step ($n > 1$).] \hfill

In this case $t$ and $s$ must be of the form $t = c(\ntup{t}{n})$ and $s=c'(\ntup{s}{n})$.
The $EQ_\simrel$ clause used in the {\bf QDA} inference step at the root must be of the form: 
$$c(\ntup{X}{n}) \sim c'(\ntup{Y}{n}) \qgets{\tp} \qat{\mbox{pay}_{d_0}}{?},\ (\qat{(X_i \sim Y_i)}{?})_{i = 1 \ldots n}$$
with $\simrel(c,c') = d_0 \neq \bt$. The inference step at the root will be:
$$
\displaystyle\frac
  {~
    \begin{array}{l@{\hspace{1cm}}l}
      \cqat{(t == c(\ntup{t}{n}))}{d_1}{\Pi} & \cqat{pay_{d_0}}{e_0}{\Pi} \\
      \cqat{(s == c'(\ntup{s}{n}))}{d_2}{\Pi} & (~ \cqat{(t_i \sim s_i)}{e_i}{\Pi} ~)_{i = 1 \ldots n} \\
    \end{array}
  ~}
  {\cqat{(t \sim s)}{d}{\Pi}}
$$
with $d \dleq d_1 \sqcap d_2 \sqcap  \bigsqcap_{i = 0}^n e_i$.
Due to the forms of the $EQ_\simrel$ clause $\mbox{pay}_{d_0} \!\qgets{d_0}$ and the {\bf QEA} inference rule there is no loss of generality in assuming $d_1 = d_2 = \tp$ and $e_0 = d_0$, therefore we have $d \dleq d_0 \sqcap \bigsqcap_{i = 1}^n e_i$.
By the inductive hypothesis $t_i  \approx_{e_i, \Pi} s_i ~ (1 \le i \le n)$, i.e. there are constructor terms $\hat{t}_i$, $\hat{s}_i$ such that $t_i \approx_\Pi \hat{t_i}$,  $s_i \approx_\Pi \hat{s}_i$ and $\hat{t}_i \approx_{e_i} \hat{s}_i$ for $i = 1 \ldots n$.
Thus, we can build $\hat{t} = c(\hat{t}_1, \ldots, \hat{t}_n)$ and $\hat{s} = c'(\hat{s}_1, \ldots, \hat{s}_n)$ having $t \approx_{d,\Pi} s$ because:
\begin{itemize}
\item 
$t \approx_\Pi \hat{t}$, i.e. $c(\ntup{t}{n}) \approx_\Pi c(\ntup{\hat{t}}{n})$, by decomposition since $t_i \approx_\Pi \hat{t}_i$.
\item
$s \approx_\Pi \hat{s}$, i.e. $c'(\ntup{s}{n}) \approx_\Pi c'(\ntup{\hat{s}}{n})$, again by decomposition since $s_i \approx_\Pi \hat{s}_i$.
\item
$\hat{t} \approx_d \hat{s}$, since
$d \dleq  d_0 \sqcap \bigsqcap_{i = 1}^n e_i \dleq \simrel(c,c') \sqcap \bigsqcap_{i = 1}^n  \simrel(\hat{t}_i,\hat{s}_i) = \simrel(\hat{t},\hat{s}) \enspace .\mathproofbox
$
\end{itemize}
\end{description}
\end{proof*}

We are now ready to define elim$_\simrel$ acting over programs and goals.

\begin{defn}\label{def:sqclptransform}
Assume a $\sqclp{\simrel}{\qdom}{\cdom}$-program $\Prog$ and a $\sqclp{\simrel}{\qdom}{\cdom}$-goal $G$ for $\Prog$ whose atoms are all relevant for $\Prog$.
Then we define:
\begin{enumerate}
\item
For each atom $A$, let $A_\sim$ be $t \sim s$ if $A : t == s$; otherwise let $A_\sim$ be  $A$.
\item
For each clause $C : (p(\ntup{t}{n}) \qgets{\alpha} \tup{B}) \in \Prog$
let $\hat{\mathcal{C}}_\simrel$ be the set of $\qclp{\qdom}{\cdom}$ clauses consisting of:
\begin{itemize}
\item[---]
The clause $\hat{C} : (\transRen{p}_{C}(\ntup{t}{n}) \qgets{\alpha} \ntup{B}{\sim})$, where $\transRen{p}_{C} \in DP^n$ is not affected by $\Prog$ (chosen in a different way for each $C$) and $\ntup{B}{\sim}$ is obtained from $\tup{B}$ by replacing each atom $A$ occurring in $\tup{B}$ by $A_\sim$.
\item[---]
A clause $p'(\ntup{X}{n}) \qgets{\tp} \qat{\mbox{pay}_\lambda}{?},\ (\qat{(X_i \sim t_i)}{?})_{i = 1 \ldots n},\ \qat{ \transRen{p}_C(\ntup{t}{n})}{?}$ for each $p' \in DP^n$ such that $\simrel(p,p') = \lambda \neq \bt$.
Here, $\ntup{X}{n}$ must be chosen as $n$ pairwise different variables not occurring in the clause $C$.
\end{itemize}
\item
$\elimS{\Prog}$ is the $\qclp{\qdom}{\cdom}$-program $EQ_\simrel \cup \hat{\Prog}_\simrel$ where $\hat{\Prog}_\simrel \eqdef \bigcup_{C \in \Prog} \hat{\mathcal{C}}_\simrel$. 
\item $\elimS{G}$  is the $\qclp{\qdom}{\cdom}$-goal $G_\sim$ obtained from $G$ by replacing each atom $A$ occurring in $G$ by $A_\sim$.
\mathproofbox
\end{enumerate}
\end{defn}

The following example illustrates the transformation elim$_\simrel$.
\begin{exmp}[Running example: $\qclp{\U\!\otimes\!\W}{\,\rdom}$-program $\elimS{\Prog_r}$]
\label{exmp:elims-pr}
Consider the $\sqclp{\simrel_r}{\,\U{\otimes}\W}{\rdom}$-program $\Prog_r$ and the goal $G_r$ for $\Prog_r$ as presented in Example \ref{exmp:pr}. 
The transformed $\qclp{\U{\otimes}\W}{\rdom}$-program $\elimS{\Prog_r}$ is as follows:
\begin{center}
\footnotesize\it
\renewcommand{\arraystretch}{1.6}
\begin{tabular}{rl}
\tiny $\hat{R}_1$ & \^{f}amous$_{R_1}$(sha) $\qgets{(0.9,1)}$ \\
\tiny $R_{1.1}$ & famous(X) $\gets$ pay$_\tp$, X$\sim$sha, \^{f}amous$_{R_1}$(sha) \\
\tiny $\hat{R}_2$ & \^{w}rote$_{R_2}$(sha, kle) $\qgets{(1,1)}$ \\
\tiny $R_{2.1}$ & wrote(X, Y) $\gets$ pay$_\tp$, X$\sim$sha, Y$\sim$kle, \^{w}rote$_{R_2}$(sha, kle) \\
\tiny $R_{2.2}$ & authored(X, Y) $\gets$ pay$_{(0.9,0)}$, X$\sim$sha, Y$\sim$kle, \^{w}rote$_{R_2}$(sha, kle) \\
\tiny $\hat{R}_3$ & \^{w}rote$_{R_3}$(sha, hamlet) $\qgets{(1,1)}$ \\
\tiny $R_{3.1}$ & wrote(X, Y) $\gets$ pay$_\tp$, X$\sim$sha, Y$\sim$hamlet, \^{w}rote$_{R_3}$(sha, hamlet) \\
\tiny $R_{3.2}$ & authored(X, Y) $\gets$ pay$_{(0.9,0)}$, X$\sim$sha, Y$\sim$hamlet, \^{w}rote$_{R_3}$(sha, hamlet) \\
\tiny $\hat{R}_4$ & \^{g}ood\_work$_{R_4}$(G) $\qgets{(0.75,3)}$ famous(A)\#(0.5,100), authored(A, G) \\
\tiny $R_{4.1}$ & good\_work(X) $\gets$ pay$_\tp$, X$\sim$G, \^{g}ood\_work$_{R_4}$(G) \\[4mm]
\end{tabular} \\
\begin{tabular}{l@{\hspace{1.2cm}}l}
\% Program clauses for $\sim$: & \% Program clauses for pay: \\
X\,$\sim$Y $\gets$ X==Y & pay$_\tp$ $\gets$ \\
kle\,$\sim$\,kli $\gets$ pay$_{(0.8,2)}$ & pay$_{(0.9,0)}$ $\qgets{(0.9,0)}$ \\
$[\ldots]$ & pay$_{(0.8,2)}$ $\qgets{(0.8,2)}$ \\[2mm]
\end{tabular}
\end{center}

Finally, the goal $\elimS{G_r}$ for $\elimS{\Prog_r}$ is as follows:
\begin{center}
\it good\_work(X)\#W $\sep$ \!W $\dgeq^?$ \!\!(0.5,10) \mathproofbox
\end{center}
\end{exmp}

The next theorem proves the semantic correctness of the program transformation.

\begin{thm}
\label{thm:SQCLP2QCLP:programs}
Consider a $\sqclp{\simrel}{\qdom}{\cdom}$-program $\Prog$\!, an atom $A$ relevant for $\Prog$\!, a qualification value $d \in \aqdom$ and a satisfiable finite set of $\cdom$-constraints $\Pi$.
Then, the following two statements are equivalent:
\begin{enumerate}
  \item $\Prog \sqchlrdc \cqat{A}{d}{\Pi}$
  \item $\elimS{\Prog} \qchldc \cqat{A_\sim}{d}{\Pi}$
\end{enumerate}
where $A_\sim$ is understood as in Definition \ref{def:sqclptransform}(1).
\end{thm}

\begin{proof*}
We separately prove each implication.

\smallskip\noindent
[1. $\Rightarrow$ 2.] {\em (the transformation is complete).}
Assume that $T$ is a $\SQCHL(\simrel,\qdom,\cdom)$ proof tree witnessing $\Prog \sqchlrdc \cqat{A}{d}{\Pi}$.
We want to show the existence of a $\QCHL(\qdom,\cdom)$ proof tree $T'$
witnessing $\elimS{\Prog} \qchldc \cqat{A_\sim}{d}{\Pi}$. We reason by complete induction on $\Vert T \Vert$.
There are three possible cases according to the syntactic form of the atom $A$.
In each case we argue how to build the desired proof tree $T'$.

\noindent --- $A$ is a primitive atom $\kappa$.
In this case $A_\sim$ is also $\kappa$ and $T$ contains only one {\bf SQPA} inference node. Because of the inference rules {\bf SQPA} and {\bf QPA}, both $\Prog \sqchlrdc \cqat{\kappa}{d}{\Pi}$ and $\elimS{\Prog} \qchldc \cqat{\kappa}{d}{\Pi}$ are equivalent to $\Pi \model{\cdom} \kappa$, therefore $T'$ trivially contains just one {\bf QPA} inference node.

\noindent --- $A$ is an equation $t == s$.
In this case $A_\sim$ is $t \sim s$ and $T$ contains just one {\bf SQEA} inference node. We know $\Prog \sqchlrdc \cqat{(t == s)}{d}{\Pi}$ is equivalent to $t \approx_{d,\Pi} s$ because of the inference rule {\bf SQEA}. From this equivalence follows $EQ_\simrel \qchldc \cqat{(t \sim s)}{d}{\Pi}$ due to Lemma \ref{lema:equiv} and hence $\elimS{\Prog} \qchldc \cqat{(t \sim s)}{d}{\Pi}$ by construction of $\elimS{\Prog}$. In this case, $T'$ will be a proof tree rooted by a {\bf QDA} inference step.

\noindent --- $A$ is a defined atom $p'(\ntup{t'}{n})$ with $p' \in DP^n$\!.
In this case $A_\sim$ is $p'(\ntup{t'}{n})$ and the root inference of $T$ must be a {\bf SQDA} inference step of the form:
$$
\displaystyle\frac
{~
  (~ \cqat{(t'_i == t_i\theta)}{d_i}{\Pi} ~)_{i = 1 \ldots n}
  \quad
  (~ \cqat{B_j\theta}{e_j}{\Pi} ~)_{j = 1 \ldots m}
~}
{~ \cqat{p'(\ntup{t'}{n})}{d}{\Pi} ~}
~ (\clubsuit)
$$
with $C : (p(\ntup{t}{n}) \qgets{\alpha} \qat{B_1}{w_1}, \ldots, \qat{B_m}{w_m}) \in \Prog$, $\theta$ substitution, $\simrel(p',p) = d_0 \neq \bt$, $e_j \dgeq^? w_j ~ (1 \leq j \leq m)$, $d \dleq d_i ~ (0 \leq i \leq n)$ and $d \dleq \alpha \circ e_j ~ (1 \leq j \leq m)$---which means $d \dleq \alpha$ in the case  $m = 0$.
We can assume that the first $n$ premises at ($\clubsuit$) are proved in $\sqclp{\simrel}{\qdom}{\cdom}$ w.r.t. $\Prog$ by proof trees $T_{1i} ~ (1 \leq i \leq n)$ satisfying $\Vert T_{1i} \Vert < \Vert T \Vert ~ (1 \leq i \leq n)$, and the last $m$ premises at ($\clubsuit$) are proved in $\sqclp{\simrel}{\qdom}{\cdom}$ w.r.t. $\Prog$ by proof trees $T_{2j} ~ (1 \leq j \leq m)$ satisfying $\Vert T_{2j} \Vert < \Vert T \Vert ~ (1 \leq j \leq m)$.

By Definition \ref{def:sqclptransform}, we know that the transformed program $\elimS{\Prog}$ contains two clauses of the following form:
$$
\begin{array}{c@{\hspace{1mm}}cl}
\hat{C}&:&
\hat{p}_C(\ntup{t}{n}) \qgets{\alpha} \qat{B_\sim^1}{w_1},~ \ldots,~ \qat{B_\sim^m}{w_m} \\
\hat{C}_{p'}&:&
p'(\ntup{X}{n}) \qgets{\tp} \qat{\mbox{pay}_{d_0}}{?},~ (~ \qat{(X_i \sim t_i)}{?} ~)_{i = 1 \ldots n},~ \qat{\hat{p}_C(\ntup{t}{n})}{?} \\
\end{array}
$$
where $X_i ~ (1 \leq i \leq n)$ are fresh variables not occurring in $C$ and $B_\sim^j ~ (1 \leq j \leq m)$ is the result of replacing `$\sim$' for `==' if $B_j$ is equation; and $B_j$ itself otherwise.
Given that the $n$ variables $X_i$ do not occur in $C$, we can assume that $\sigma \eqdef \theta' \uplus \theta$ with $\theta' \eqdef \{ X_1 \mapsto t'_1,~ \ldots,~ X_n \mapsto t'_n\}$ is a well-defined substitution.
We claim that $\elimS{\Prog} \qchldc \cqat{A_\sim}{d}{\Pi}$ can be proved with 
a proof tree $T'$ rooted by the {\bf QDA} inference step ($\spadesuit$.1), which uses the clause $\hat{C}_{p'}$ instantiated by $\sigma$ and having $d_{n+1} = d$.
$$
\displaystyle\frac
{~
  \begin{array}{l}
  (~ \cqat{(t'_i == X_i\sigma)}{\tp}{\Pi} ~)_{i = 1 \ldots n} \\
  \cqat{\mbox{pay}_{d_0}\sigma}{d_0}{\Pi} \\
  (~ \cqat{(X_i \sim t_i)\sigma}{d_i}{\Pi} ~)_{i = 1 \ldots n} \\
  \cqat{\hat{p}_C(\ntup{t}{n})\sigma}{d_{n+1}}{\Pi} \\
  \end{array}
~}
{~ \cqat{p'(\ntup{t'}{n})}{d}{\Pi} ~}
~ (\spadesuit.1)
\quad
\displaystyle\frac
{~
  \begin{array}{l}
  (~ \cqat{(t'_i == X_i\theta')}{\tp}{\Pi} ~)_{i = 1 \ldots n} \\
  \cqat{\mbox{pay}_{d_0}}{d_0}{\Pi} \\
  (~ \cqat{(X_i\theta' \sim t_i\theta)}{d_i}{\Pi} ~)_{i = 1 \ldots n} \\
  \cqat{\hat{p}_C(\ntup{t}{n}\theta)}{d_{n+1}}{\Pi} \\
  \end{array}
~}
{~ \cqat{p'(\ntup{t'}{n})}{d}{\Pi} ~}
~ (\spadesuit.2)
$$
By construction of $\sigma$, ($\spadesuit$.1) can be rewritten as ($\spadesuit$.2), and in order to build the rest of $T'$\!, we show that each premise of ($\spadesuit$.2) admits a proof in $\QCHL(\qdom,\cdom)$ w.r.t. the transformed program $\elimS{\Prog}$:
\begin{itemize}
\item
$\elimS{\Prog} \qchldc \cqat{(t'_i == X_i\theta')}{\tp}{\Pi}$ for $i = 1 \ldots n$.
Straightforward using a single {\bf QEA} inference step since $X_i\theta' = t'_i$ and $t'_i \approx_\Pi t'_i$ is trivially true.
\item $\elimS{\Prog} \qchldc \cqat{\mbox{pay}_{d_0}}{d_0}{\Pi}$.
Immediate using the clause $(\mbox{pay}_{d_0} \qgets{d_0}) \in \elimS{\Prog}$ with a single {\bf QDA} inference step.
\item $\elimS{\Prog} \qchldc \cqat{(X_i\theta' \sim t_i\theta)}{d_i}{\Pi}$ for $i = 1 \ldots n$.
From the first $n$ premises of ($\clubsuit$) we know $\Prog \sqchlrdc \cqat{(t'_i == t_i\theta)}{d_i}{\Pi}$ with a proof tree $T_{1i}$ satisfying $\Vert T_{1i} \Vert < \Vert T \Vert$ for $i = 1 \ldots n$. Therefore, for $i = 1 \ldots n$, $\elimS{\Prog} \qchldc \cqat{(t'_i \sim t_i\theta)}{d_i}{\Pi}$ with some QCHL($\qdom$,$\cdom$) proof tree $T'_{1i}$ by inductive hypothesis. Since $(X_i\theta' \sim t_i\theta) = (t'_i \sim t_i\theta)$ for $i = 1 \ldots n$, we are done.
\item $\elimS{\Prog} \qchldc \cqat{\hat{p}_C(\ntup{t}{n}\theta)}{d}{\Pi}$.
This is proved by a $\QCHL(\qdom,\cdom)$ proof tree with a {\bf QDA} inference step node at its root of the following form:
$$
\displaystyle\frac
{~
  (~ \cqat{(t_i\theta == t_i\theta)}{d_i}{\Pi} ~)_{i = 1 \ldots n}
  \quad
  (~ \cqat{B_\sim^j\theta}{e_j}{\Pi} ~)_{j = 1 \ldots m}
~}
{~ \cqat{\hat{p}_C(\ntup{t}{n}\theta)}{d}{\Pi} ~}
~ (\heartsuit)
$$
which uses the program clause $\hat{C}$ instantiated by the substitution $\theta$.
Once more, we have to check that the premises can be derived in $\QCHL(\qdom,\cdom)$ from the transformed program $\elimS{\Prog}$ and that the side conditions of ($\heartsuit$) are satisfied:
\begin{itemize}
  \item The first $n$ premises can be trivially proved using {\bf QEA} inference steps.
  \item The last $m$ premises can be proved w.r.t. $\elimS{\Prog}$ with some $\QCHL(\qdom,\cdom)$ proof trees $T'_{2j} ~ (1 \leq j \leq m)$ by the inductive hypothesis, since we have premises $(~ \cqat{B_j\theta}{e_j}{\Pi} ~)_{j = 1 \ldots m}$ at ($\clubsuit$) that can be proved in $\sqclp{\simrel}{\qdom}{\cdom}$ w.r.t. $\Prog$ with proof trees $T_{2j}$ of size $\Vert T_{2j} \Vert < \Vert T \Vert ~ (1 \leq j \leq m)$.
  \item The side conditions---namely: $e_j \dgeq^? w_j ~ (1 \leq j \leq m)$, $d \dleq d_i ~ (1 \leq i \leq n)$ 
   and $d \dleq \alpha \circ e_j ~ (1 \leq j \leq m)$---trivially hold because they are also satisfied by ($\clubsuit$). 
\end{itemize}
\end{itemize}

Finally, we complete the construction of $T'$ by checking that ($\spadesuit$.2) satisfies
the side conditions of the inference rule {\bf QDA}:
\begin{itemize}
  \item All threshold values at the body of $\hat{C}_{p'}$ are `?'\!, therefore the first group of side conditions becomes $d_i \dgeq^?\,\,? ~ (0 \leq i \leq n+1)$, which are trivially true.
  \item The second side condition reduces to $d \dleq \tp$, which is also trivially true.
  \item The third, and last, side condition is $d \dleq \tp \circ d_i ~ (0 \leq i \leq n+1)$, or equivalently $d \dleq d_i ~ (0 \leq i \leq n+1)$. In fact, $d \dleq d_i ~ (0 \leq i \leq n)$ holds due to the side conditions in ($\clubsuit$), and $d \dleq d_{n+1}$ holds because $d_{n+1} = d$ by construction of ($\spadesuit$.1) and ($\spadesuit$.2).
\end{itemize}

\smallskip\noindent
[2. $\Rightarrow$ 1.] {\em (the transformation is sound).}
Assume that $T'$ is a $\QCHL(\qdom,\cdom)$ proof tree witnessing $\elimS{\Prog} \qchldc \cqat{A_\sim}{d}{\Pi}$.
We want to show the existence of a $\SQCHL(\simrel,\qdom,\cdom)$ proof tree $T$ witnessing $\Prog \sqchlrdc \cqat{A}{d}{\Pi}$. We reason by complete induction of $\Vert T' \Vert$.
There are three possible cases according to the syntactic form of the atom $A_\sim$.
In each case we argue how to build the desired proof tree $T$.

\noindent --- $A_\sim$ is a primitive atom $\kappa$.
In this case $A$ is also $\kappa$ and $T'$ contains only one {\bf QPA} inference node. Both $\elimS{\Prog} \qchldc \cqat{\kappa}{d}{\Pi}$ and $\Prog \sqchlrdc \cqat{\kappa}{d}{\Pi}$ are equivalent to $\Pi \model{\cdom} \kappa$ because of the inference rules {\bf QPA} and {\bf SQPA}, therefore $T$ trivially contains just one {\bf SQPA} inference node.

\noindent --- $A_\sim$ is of the form $t \sim s$.
In this case $A$ is $t == s$ and $T'$ is rooted by a {\bf QDA} inference step. From $\elimS{\Prog} \qchldc \cqat{(t \sim s)}{d}{\Pi}$ and by construction of $\elimS{\Prog}$ we have $EQ_\simrel \qchldc \cqat{(t \sim s)}{d}{\Pi}$. By Lemma \ref{lema:equiv} we get $t \approx_{d,\Pi} s$ and, by the definition of the {\bf SQEA} inference step, we can build $T$ as a proof tree with only one {\bf SQEA} inference node proving $\Prog \sqchlrdc  \cqat{(t == s)}{d}{\Pi}$.

\noindent --- $A_\sim$ is a defined atom $p'(\ntup{t}{n})$ with $p'\in DP^n$ and $p' \neq\,\,\sim$.
In this case $A = A_\sim$ and the step at the root of $T'$ must be a  {\bf QDA} inference step using a clause $C' \in \elimS{\Prog}$ with head predicate $p'$ and a substitution $\theta$. Because of Definition \ref{def:sqclptransform} and the fact that $p'$ is relevant for $\Prog$, there must be some clause $C : (p(\ntup{t}{n}) \qgets{\alpha} \tup{B}) \in \Prog$ such that $\simrel(p,p') = d_0 \neq \bt$, and $C'$ must be of the form:
$$
C' :
p'(\ntup{X}{n}) \qgets{\tp} \qat{\mbox{pay}_{d_0}}{?},~ (\qat{(X_i \sim t_i)}{?})_{i = 1 \ldots n},~ \qat{\hat{p}_C(\ntup{t}{n})}{?}
$$
where the variables $\ntup{X}{n}$ do not occur in $C$.
Thus the {\bf QDA} inference step at the root of $T'$ must be of the form:
$$
\displaystyle\frac
{~
  \begin{array}{l}
  (~ \cqat{(t'_i == X_i\theta)}{d_{1i}}{\Pi} ~)_{i = 1 \ldots n} \\
  \cqat{\mbox{pay}_{d_0}\theta}{e_{10}}{\Pi} \\
  (~ \cqat{(X_i \sim t_i)\theta}{e_{1i}}{\Pi} ~)_{i = 1 \ldots n} \\
  \cqat{\hat{p}_C(\ntup{t}{n})\theta}{e_{1(n+1)}}{\Pi} \\
  \end{array}
~}
{~ \cqat{p'(\ntup{t'}{n})}{d}{\Pi} ~}
~ (\spadesuit)
$$
and the proof of the last premise must use the only clause for $\hat{p}_C$ introduced in $\elimS{\Prog}$ according to Definition  \ref{def:sqclptransform}, i.e.:
$$
\hat{C}:
\hat{p}_C(\ntup{t}{n}) \qgets{\alpha} \qat{B_\sim^1}{w_1},~ \ldots,~ \qat{B_\sim^m}{w_m} \enspace .
$$
Therefore, the proof of this premise must be of the form:
$$
\displaystyle\frac
{~
  (~ \cqat{(t_i\theta == t_i\theta')}{d_{2i}}{\Pi} ~)_{i = 1 \ldots n}
  \quad
  (~ \cqat{B_\sim^j\theta'}{e_{2j}}{\Pi} ~)_{j = 1 \ldots m}
~}
{~ \cqat{\hat{p}_C(\ntup{t}{n})\theta}{e_{1(n+1)}}{\Pi} ~}
~ (\heartsuit)
$$
for some substitution $\theta'$ not affecting $\ntup{X}{n}$.
We can assume that the last $m$ premises in ($\heartsuit$) are proved in $\QCHL(\qdom,\cdom)$ w.r.t. $\elimS{\Prog}$ by proof trees $T'_{j}$ satisfying $\Vert T'_{j} \Vert < \Vert T' \Vert ~ (1 \leq j \leq m)$.
Then we use the substitution $\theta'$ and clause $C$ to build a $\SQCHL(\simrel,\qdom,\cdom)$ proof tree $T$
with a {\bf SQDA} inference step at the root of the form:
$$
\displaystyle\frac
{~
  (~ \cqat{(t'_i == t_i\theta')}{e_{1i}}{\Pi} ~)_{i = 1 \ldots n}
  \quad
  (~ \cqat{B_j\theta'}{e_{2j}}{\Pi} ~)_{j = 1 \ldots m}
~}
{~ \cqat{p'(\ntup{t'}{n})}{d}{\Pi} ~}
~ (\clubsuit)
$$
Next we check that the premises of this inference step admit proofs in $\SQCHL(\simrel,\qdom,$ $\cdom)$ and that
$(\clubsuit)$ satisfies the side conditions of a valid {\bf SQDA} inference step.
\begin{itemize}
\item
$\Prog \sqchlrdc \cqat{(t'_i == t_i\theta')}{e_{1i}}{\Pi}$ for $i = 1 \ldots n$.
\begin{itemize}
\item From the premises $(\cqat{(X_i \sim t_i)\theta}{e_{1i}}{\Pi})_{i = 1 \ldots n}$ of $(\spadesuit)$ and by construction of $\elimS{\Prog}$ we know $EQ_\simrel \qchldc \cqat{(X_i \sim t_i)\theta}{e_{1i}}{\Pi} ~ (1 \leq i \leq n)$.
Therefore by Lemma \ref{lema:equiv} we have $X_i\theta \approx_{e_{1i},\Pi} t_i\theta$ for $i=1 \dots n$.
\item Consider now the premises $(\cqat{(t'_i == X_i\theta)}{d_{1i}}{\Pi})_{i = 1 \ldots n}$ of $(\spadesuit)$.
Their proofs  must rely on {\bf QEA} inference steps, and therefore $t'_i \approx_{\Pi} X_i\theta$ holds for $i=1 \dots n$.
\item Analogously, from the proofs of the premises $(\cqat{(t_i\theta == t_i\theta')}{d_{2i}}{\Pi})_{i = 1 \ldots n}$ we have $t_i\theta \approx_{\Pi} t_i\theta'$ (or equivalently $t_i\theta' \approx_{\Pi} t_i\theta$) for $i = 1 \ldots n$.
\end{itemize} 
From the previous points we have $X_i\theta \approx_{e_{1i},\Pi} t_i\theta$, $t'_i \approx_{\Pi} X_i\theta$ and $t_i\theta' \approx_{\Pi} t_i\theta$, which by Lemma 2.7(1) of \cite{RR10TR} imply $ t'_i \approx_{e_{1i},\Pi} t_i\theta' ~ (1 \leq i \leq n)$. 
Therefore the premises $(\cqat{(t'_i == t_i\theta')}{e_{1i}}{\Pi})_{i = 1 \ldots n}$ can be proven in $\SQCHL(\simrel,\qdom,\cdom)$ using a {\bf SQEA} inference step.
\item
$\Prog \sqchlrdc \cqat{B_j\theta'}{e_{2j}}{\Pi}$ for $j = 1 \ldots m$.
We know $\elimS{\Prog} \qchldc \cqat{B_\sim^j\theta'}{e_{2j}}{\Pi}$ with a proof tree $T'_j$ satisfying $\Vert T'_j \Vert < \Vert T' \Vert ~ (1 \leq j \leq m)$ because of ($\heartsuit$).
Therefore we have, by inductive hypothesis, $\Prog \sqchlrdc \cqat{B_j\theta'}{e_{2j}}{\Pi}$ for some $\SQCHL(\simrel,\qdom,\cdom)$ proof tree $T_j ~ (1 \leq j \leq m)$.
\item
$\simrel(p,p') = d_0 \neq \bt$.
As seen above.
\item
$e_{2j} \dgeq^? w_j$ for $j = 1 \ldots m$.
This is a side condition of the {\bf QDA} step in $(\heartsuit)$.
\item
$d \dleq e_{1i}$ for $i = 1 \ldots n$.
Straightforward from the side conditions of $(\spadesuit)$, which include $d \dleq \tp \circ e_{1i}$ for $(0 \le i \le n+1)$.
\item
$d \dleq \alpha \circ e_{2j}$ for $j = 1 \ldots m$.
This follows from the side conditions of $(\spadesuit)$ and $(\heartsuit)$, since we have $d \dleq \tp \circ e_{1i}$ for $i = 0 \ldots n+1$ (in particular $d \dleq e_{1(n+1)}$) and $e_{1(n+1)} \dleq \alpha \circ e_{2j}$ for $j = 1 \ldots m$. \mathproofbox
\end{itemize}
\end{proof*}

Finally, the next theorem extends the previous result to goals.

\begin{thm}
\label{thm:SQCLP2QCLP:goals} 
Let $G$ be a goal for a $\sqclp{\simrel}{\qdom}{\cdom}$-program $\Prog$ whose atoms are all relevant for $\Prog$.
Assume $\Prog' = \elimS{\Prog}$ and $G' = \elimS{G}$. 
Then, $\Sol{\Prog}{G} = \Sol{\Prog'}{G'}$.
\end{thm}
\begin{proof*}
According to the definition of goals in Section \ref{sec:sqclp}, and Definition \ref{def:sqclptransform}, $G$ and $G'$ must be of the form $(\qat{A_i}{W_i}, W_i \,{\dgeq}^? \beta_i)_{i = 1 \ldots m}$ and $(\qat{A_\sim^i}{W_i}, W_i \,{\dgeq}^? \beta_i)_{i = 1 \ldots m}$, respectively.
By Definitions \ref{dfn:goalsol} and \ref{dfn:qclp-goalsol}, both $\Sol{\Prog}{G}$ and $\Sol{\Prog'}{G'}$ are sets of triples $\langle \sigma, \mu, \Pi \rangle$ where $\sigma$ is a $\cdom$-substitution, $\mu : \warset{G} \to \aqdomd{\qdom}$ (note that $\warset{G} = \warset{G'}$) and $\Pi$ is a satisfiable finite set of $\cdom$-constraints.
Moreover:
\begin{enumerate}
\item
$\langle \sigma, \mu, \Pi \rangle \in \Sol{\Prog}{G}$ iff $W_i\mu = d_i \dgeq^? \!\beta_i$ and $\Prog \sqchlrdc \cqat{A_i\sigma}{W_i\mu}{\Pi}$ $(1 \leq i \leq m)$.
\item
$\langle \sigma, \mu, \Pi \rangle \in \Sol{\Prog'}{G'}$ iff $W_i\mu = d_i \dgeq^? \!\beta_i$ and $\Prog' \qchldc \cqat{A_\sim^i\sigma}{W_i\mu}{\Pi}$ $(1 \leq i \leq m)$.
\end{enumerate}
Because of Theorem \ref{thm:SQCLP2QCLP:programs}, conditions (1) and (2) are equivalent. \mathproofbox
\end{proof*}

%% file: J4_2.tex
\subsection{Transforming QCLP into CLP}
\label{sec:implemen:QCLP2CLP}

The results presented in this subsection are dependant on the assumption that the qualification domain $\qdom$ is existentially expressible in the constraint domain $\cdom$ via an injective mapping $\imath : \aqdomd{\qdom} \to C_\cdom$ and two existential $\cdom$-constraints of the following form:
\begin{itemize}
\item[] $\qval{X} = \exists U_1\ldots\exists U_k(B_1 \land \ldots \land B_m)$
\item[] $\qbound{X,Y,Z} = \exists V_1\ldots\exists V_l(C_1 \land \ldots \land C_q)$
\end{itemize}

Our aim is to present semantically correct transformations from $\qclp{\qdom}{\cdom}$ into $\clp{\cdom}$,
working both for programs and goals. In order to compute with the encodings of $\qdom$ values in $\cdom$,
we will use  the $\clp{\cdom}$-program $E_\qdom$ consisting of the following two clauses:
\begin{itemize}
\item[] $qV\!al(X) ~\gets~ B_1,\ \ldots,\ B_m$
\item[] $qBound(X,Y,Z) ~\gets~ C_1,\ \ldots,\ C_q$
\end{itemize}
where $qV\!al \in DP^1$ and $qBound \in DP^3$ do not occur in the $\qclp{\qdom}{\cdom}$ programs and goals to be transformed.

The lemma stated below is an immediate consequence of Lemma \ref{lema:dec} and Definition \ref{dfn:expressible}.

\begin{lem}
\label{lema:expr}
For any satisfiable finite set $\Pi$ of $\cdom$-constraints one has:
\begin{enumerate}
\item
For any ground term $t \in C_\cdom$:
$$t \in \mbox{ran}(\imath) \iff \qval{t} \mbox{ true in } \cdom \iff E_\qdom \chlc \cat{qV\!al(t)}{\Pi}$$
\item
For any 
ground terms $r = \imath(x)$, $s = \imath(y)$, $t = \imath(z)$ with $x,y,z \in \aqdomd{\qdom}$:
$$x \dleq y \circ z \iff \qbound{r,s,t} \mbox{ true in } \cdom \iff E_\qdom \chlc \cat{qBound(r,s,t)}{\Pi}$$
\end{enumerate}
The two items above are also valid if $E_\qdom$ is replaced by any $\clp{\cdom}$-program including the two clauses in $E_\qdom$ and having no additional occurrences of $qV\!al$ and $qBound$ at the head of clauses. \mathproofbox
\end{lem}

\begin{figure}[ht]
\figrule
\centering
\begin{tabular}{@{\hspace{4mm}}l@{\hspace{1mm}}l}
\multicolumn{2}{l}{\textbf{Transforming Atoms}} \\&\\
\bf TEA & $\transform{({t == s})} \!= (t == s, ~\imath(\tp))$. \\
&\\
\bf TPA & $\transform{({\kappa})} \!= (\kappa, ~\imath(\tp))$ with $\kappa$ primitive atom. \\
&\\
\bf TDA & $\transform{({p(\ntup{t}{n})})} \!= (p'(\ntup{t}{n},W), ~W)$ with $p \in DP^n$ and $W$\! a fresh CLP variable. \\
&\\
\multicolumn{2}{l}{\textbf{Transforming qc-Atoms}\vspace{2mm}} \\
\bf TQCA &  $\displaystyle\frac
  {\transform{A} = (A',w) }
  {\quad \transform{(\qat{A}{d} \Leftarrow \Pi)} = (A' \Leftarrow \Pi, ~\{\qval{w},\ \qbound{\imath(d),\imath(\tp),w}\}) \quad}$ \\
&\\
\multicolumn{2}{l}{\textbf{Transforming Program Clauses}\vspace{2mm}} \\
\bf TPC & $\displaystyle\frac
  { (~ \transform{B_j} = (B_j',  w'_j) ~)_{j=1 \dots m}}
  {\quad
        \transform{C} = p'(\ntup{t}{n},W) ~\gets~ qV\!al(W),\ \left(
          \begin{array}{l}
            qV\!al(w_j'),\ \encode{w'_j \dgeq^? \!\imath(w_j)}, \\
            qBound(W, \imath(\alpha), w'_j),\ B'_j
          \end{array}
        \right)_{j = 1 \ldots m}
  \quad}$ \\
&\\
\multicolumn{2}{l}{$\qquad$ where $C : p(\ntup{t}{n}) \qgets{\alpha} \qat{B_1}{w_1}, \ldots, \qat{B_m}{w_m}$, $W$\! is a fresh CLP variable and} \\
\multicolumn{2}{l}{$\qquad$ $\encode{w'_j \dgeq^? \imath(w_j)}$ is omitted if $w_j =\ ?$, i.o.c. abbreviates $qBound(\imath(w_j),\imath(\tp),w'_j)$.}\\
&\\
\multicolumn{2}{l}{\textbf{Transforming Goals}\vspace{2mm}} \\
\bf TG & $\displaystyle\frac
  {(~ \transform{B_j} = (B_j',  w'_j) ~)_{j=1 \dots m}}
  {\quad
          \elimD{G} = \left(
            \begin{array}{l}
              qV\!al(W_j),\ \encode{W_j\dgeq^? \imath(\beta_j)}, \\
              qV\!al(w'_j),\ qBound(W_j, \imath(\tp), w'_j),\  B_j'
            \end{array}
          \right)_{j = 1 \ldots m}
  \quad}$ \\
&\\
\multicolumn{2}{l}{$\qquad$ where $G : (\qat{B_j}{W_j}, W_j \dgeq^? \beta_j)_{j=1 \dots m}$ and $\encode{W_j\dgeq^? \imath(\beta_i)}$ as in \textbf{TPC} above.} \\
\end{tabular}
\caption{Transformation rules}
\label{fig:transformation}
\figrule
\end{figure}

Now we are ready to define the  transformations from $\qclp{\qdom}{\cdom}$ into $\clp{\cdom}$.

\begin{defn}\label{def:qclptransform}
Assume that $\qdom$ is existentially expressible in $\cdom$, and let $\qval{X}$, $\qbound{X,Y,Z}$ and  $E_\qdom$ be as explained above.
Assume also a $\qclp{\qdom}{\cdom}$-program $\Prog$ and a $\qclp{\qdom}{\cdom}$-goal $G$ for $\Prog$ without occurrences of the defined predicate symbols $qV\!al$ and $qBound$.
Then:
\begin{enumerate}
\item
$\Prog$ is transformed into the $\clp{\cdom}$-program  $\elimD{\Prog}$ consisting of the two clauses in $E_\qdom$ and the
transformed $\transform{C}\!$ of each clause $C \in \Prog$\!, built as specified in Figure \ref{fig:transformation}. The transformation rules of this figure assume a different choice of $p' \in DP^{n+1}$ for each $p \in DP^n$\!.
\item
$G$ is transformed into the $\clp{\cdom}$-goal $\elimD{G}$ built as specified in Figure \ref{fig:transformation}.
Note that the qualification variables $\ntup{W}{\!n}$ occurring in $G$ become normal CLP variables in the transformed goal.
\mathproofbox
\end{enumerate}
\end{defn}

The following example illustrates the transformation elim$_\qdom$.
\begin{exmp}[Running example: $\clp{\rdom}$-program $\elimD{\elimS{\Prog_r}}$]
\label{exmp:elimd-elims-pr}

Consider the $\qclp{\U{\otimes}\W}{\rdom}$-program $\elimS{\Prog_r}$ and the goal $\elimS{G_r}$ for the same program as presented in Example \ref{exmp:elims-pr}. The transformed $\clp{\rdom}$-program $\elimD{\elimS{\Prog_r}}$ is as follows:
\begin{center}
\footnotesize\it
\renewcommand{\arraystretch}{1.5}
\begin{tabular}{rl}
\tiny $\hat{R}_1$ & \^{f}amous$_{R_1}$(sha, W) $\gets$ qVal(W), qBound(W, $\tp$, (0.9,1)) \\
\tiny $R_{1.1}$ & famous(X, W) $\gets$ qVal(W), qVal(W$_1$), qBound(W, $\tp$, \!W$_1$), pay$_\tp$(W$_1$), \\
& $\quad$ qVal(W$_2$), qBound(W, $\tp$, \!W$_2$), $\sim$(X, sha, W$_2$), \\
& $\quad$ qVal(W$_3$), qBound(W, $\tp$, \!W$_3$), \^{f}amous$_{R_1}$(sha, W$_3$) \\
\tiny $\hat{R}_2$ & \^{w}rote$_{R_2}$(sha, kle, W) $\gets$ qVal(W), qBound(W, $\tp$, (1,1)) \\
\tiny $R_{2.1}$ & wrote(X, Y, W) $\gets$ qVal(W), qVal(W$_1$), qBound(W, $\tp$, \!W$_1$), pay$_\tp$(W$_1$), \\
& $\quad$ qVal(W$_2$), qBound(W, $\tp$, \!W$_2$), $\sim$(X, sha, W$_2$), \\
& $\quad$ qVal(W$_3$), qBound(W, $\tp$, \!W$_3$), $\sim$(Y, kle, W$_3$), \\
& $\quad$ qVal(W$_4$), qBound(W, $\tp$, \!W$_4$), \^{w}rote$_{R_2}$(sha, kle, W$_4$) \\
\tiny $R_{2.2}$ & authored(X, Y, W) $\gets$ qVal(W), qVal(W$_1$), qBound(W, $\tp$, \!W$_1$), pay$_{(0.9,0)}$(W$_1$), \\
& $\quad$ qVal(W$_2$), qBound(W, $\tp$, \!W$_2$), $\sim$(X, sha, W$_2$), \\
& $\quad$ qVal(W$_3$), qBound(W, $\tp$, \!W$_3$), $\sim$(Y, kle, W$_3$), \\
& $\quad$ qVal(W$_4$), qBound(W, $\tp$, \!W$_4$), \^{w}rote$_{R_2}$(sha, kle, W$_4$) \\
\tiny $\hat{R}_3$ & \^{w}rote$_{R_3}$(sha, hamlet, W) $\gets$ qVal(W), qBound(W, $\tp$, (1,1)) \\
\tiny $R_{3.1}$ & wrote(X, Y, W) $\gets$ qVal(W), qVal(W$_1$), qBound(W, $\tp$, \!W$_1$), pay$_\tp$(W$_1$), \\
& $\quad$ qVal(W$_2$), qBound(W, $\tp$, \!W$_2$), $\sim$(X, sha, W$_2$), \\
& $\quad$ qVal(W$_3$), qBound(W, $\tp$, \!W$_3$), $\sim$(Y, hamlet, W$_3$), \\
& $\quad$ qVal(W$_4$), qBound(W, $\tp$, \!W$_4$), \^{w}rote$_{R_3}$(sha, hamlet, W$_4$) \\
\end{tabular}
\end{center}
\begin{center}
\footnotesize\it
\renewcommand{\arraystretch}{1.5}
\begin{tabular}{rl}
\tiny $R_{3.2}$ & authored(X, Y, W) $\gets$ qVal(W), qVal(W$_1$), qBound(W, $\tp$, \!W$_1$), pay$_{(0.9,0)}$(W$_1$), \\
& $\quad$ qVal(W$_2$), qBound(W, $\tp$, \!W$_2$), $\sim$(X, sha, W$_2$), \\
& $\quad$ qVal(W$_3$), qBound(W, $\tp$, \!W$_3$), $\sim$(Y, hamlet, W$_3$), \\
& $\quad$ qVal(W$_4$), qBound(W, $\tp$, \!W$_4$), \^{w}rote$_{R_3}$(sha, hamlet, W$_4$) \\
\tiny $\hat{R}_4$ & \^{g}ood\_work$_{R_4}$(X, W) $\gets$ qVal(W), \\
& $\quad$ qVal(W$_1$), qBound((0.5,100), $\tp$, W$_1$), qBound(W, (0.75,3), W$_1$), famous(Y, W$_1$), \\
& $\quad$ qVal(W$_2$), qBound(W, (0.75,3), W$_2$), authored(Y, X, W$_2$) \\
\tiny $R_{4.1}$ & good\_work(G, W) $\gets$ qVal(W), qVal(W$_1$), qBound(W, $\tp$, \!W$_1$), pay$_\tp$(W$_1$), \\
& $\quad$ qVal(W$_2$), qBound(W, $\tp$, \!W$_2$), $\sim$(G, X, W$_2$), \\
& $\quad$ qVal(W$_3$), qBound(W, $\tp$, \!W$_3$), \^{g}ood\_work$_{R_4}$(X, W$_3$) \\[3mm]
& \% Program clauses for $\sim$: \\
& $\sim$(X, Y, W) $\gets$ qVal(W), qVal($\tp$), qBound(W, $\tp$, $\tp$), X==Y \\
& $\sim$(kle, kli, W) $\gets$ qVal(W), qVal(W$_1$), qBound(W, $\tp$, \!W$_1$), pay$_{(0.8,2)}$(W$_1$) \\
& $[\ldots]$ \\[3mm]
& \% Program clauses for pay: \\
& pay$_\tp$(W) $\gets$ qVal(W), qBound(W, $\tp$, $\tp$) \\
& pay$_{(0.9,0)}$(W) $\gets$ qVal(W), qBound(W, $\tp$, (0.9,0)) \\
& pay$_{(0.8,2)}$(W) $\gets$ qVal(W), qBound(W, $\tp$, (0.8,2)) \\[3mm]
& \% Program clauses for qVal \& qBound: \\
& qVal((X$_1$,X$_2$)) $\gets$ X$_1$ $>$ 0, X$_1$ $\le$ 1, X$_2$ $\ge$ 0 \\
& qBound((W$_1$,W$_2$), (Y$_1$,Y$_2$), (Z$_1$,Z$_2$)) $\gets$ W$_1$ $\le$ Y$_1$ $\times$ Z$_1$, W$_2$ $\ge$ Y$_2$ $+$ Z$_2$ \\[2mm]
\end{tabular}
\end{center}

Finally, the goal $\elimD{\elimS{G_r}}$ for $\elimD{\elimS{\Prog_r}}$ is as follows:
\begin{center}
\footnotesize\it qVal(W), qBound((0.5,10), $\tp$, \!W), qVal(W'), qBound(W, $\tp$, \!W'), good\_work(X, W')
\end{center}

Note that, in order to improve the clarity of the program clauses of this example, the qualification value $(1,\!0)$---top value in $\U{\otimes}\W$---has been replaced by $\tp$. \mathproofbox
\end{exmp}

The next theorem proves the semantic correctness of the program transformation.

\begin{thm}
\label{thm:QCLP2CLP:programs}
Let $A$ be an atom such that $qV\!al$ and $qBound$ do not occur in $A$. Assume $d \in \aqdom$ such that $\transform{(\cqat{A}{d}{\Pi})} = (\cat{A'}{\Pi}, \Omega)$.
Then, the two following statements are equivalent:
\begin{enumerate}
  \item  $\Prog \qchldc \cqat{A}{d}{\Pi}$
  \item $\elimD{\Prog} \chlc \cat{A'\!\rho}{\Pi}$ for some $\rho \in \Sol{\cdom}{\Omega}$ such that $\domset{\rho} = \varset{\Omega}$.
\end{enumerate}
\end{thm}
\begin{proof*}
We separately prove each implication.

\smallskip\noindent
[1. $\Rightarrow$ 2.] {\em (the transformation is complete).}
We assume that $T$ is a $\QCHL(\qdom,\cdom)$ proof tree witnessing $\Prog \qchldc \cqat{A}{d}{\Pi}$.
We want to show the existence of a $\clp{\cdom}$ proof tree  $T'$
witnessing $\elimD{\Prog} \chlc \cat{A'\!\rho}{\Pi}$
for some $\rho \in \Sol{\cdom}{\Omega}$ such that $\domset{\rho} = \varset{\Omega}$.
We reason  by complete induction on $\Vert T \Vert$.
There are three possible cases,  according to the the syntactic form of the atom $A$.
In each case we argue how to build the desired proof tree $T'$\!.

\noindent --- $A$ is a primitive atom $\kappa$.
In this case {\bf TQCA} and {\bf TPA} compute $A' = \kappa$ and $\Omega = \{\qval{\imath(\tp)},\ \qbound{\imath(d),\imath(\tp),\imath(\tp)}\}$. Now, from $\Prog \qchldc \cqat{\kappa}{d}{\Pi}$ follows $\Pi \model{\cdom} \kappa$ due to the {\bf QPA} inference, and therefore taking $\rho = \varepsilon$ we can prove $\elimD{\Prog} \chlc \cat{\kappa\varepsilon}{\Pi}$ with a proof tree $T'$\! containing only one {\bf PA} node. Moreover, $\varepsilon \in \Solc{\Omega}$ is trivially true because the two constraints belonging to $\Omega$ are obviously true in $\cdom$.

\noindent --- $A$ is an equation $t == s$.
In this case {\bf TQCA} and {\bf TEA} compute $A' = (t == s)$ and $\Omega = \{\qval{\imath(\tp)},\ \qbound{\imath(d),\imath(\tp),\imath(\tp)}\}$. Now, from $\Prog \qchldc \cqat{(t==s)}{d}{\Pi}$ follows $t \approx_\Pi s$ due to the {\bf QEA} inference, and therefore taking $\rho = \varepsilon$ we can prove $\elimD{\Prog} \chlc \cat{(t == s)\varepsilon}{\Pi}$ with a proof tree $T'$\! containing only one {\bf EA} node. Moreover, $\varepsilon \in \Solc{\Omega}$ is trivially true because the two constraints belonging to $\Omega$ are obviously true in $\cdom$.

\noindent --- $A$ is a defined atom $p(\ntup{t'}{n})$ with $p \in DP^n$.
In this case {\bf TQCA} and {\bf TDA} compute $A' = p'(\ntup{t'}{n},W)$ and $\Omega = \{\qval{W},\ \qbound{\imath(d),\imath(\tp),W}\}$ where $W$ is a fresh CLP variable.
On the other hand, $T$ must be rooted by a {\bf QDA} step of the form:
$$
\displaystyle\frac
    {~ (~ \cqat{(t'_i == t_i\theta)}{d_i}{\Pi} ~)_{i = 1 \ldots n} \quad (~ \cqat{B_j\theta}{e_j}{\Pi} ~)_{j = 1 \ldots m} ~}
    {\cqat{p(\ntup{t'}{n})}{d}{\Pi}} \quad (\clubsuit)
$$
using a clause $C : (p(\ntup{t}{n}) \qgets{\alpha} \qat{B_1}{w_1}, \ldots, \qat{B_m}{w_m}) \in \Prog$ instantiated by a substitution $\theta$ and such that the side conditions $e_j \dgeq^? w_j ~ (1 \le j \le m)$, $d \dleq d_i ~ (1 \le i \le n)$ and $d \dleq \alpha \circ e_j ~ (1 \le j \le m)$ are fulfilled.

For $j = 1 \ldots m$ we can assume $\transform{B_j} = (B'_j, w'_j)$ and thus $\transform{(\cqat{B_j\theta}{e_j}{\Pi})} = (\cat{B'_j\theta}{\Pi}, \Omega_j)$ where $\Omega_j = \{\qval{w'_j},\ \qbound{\imath(e_j),\imath(\tp),w'_j}\}$. The proof trees $T_j$ of the last $m$ premises of $(\clubsuit)$ will have less than $\Vert T \Vert$ nodes,
and hence the induction hypothesis can be applied to each $(\cqat{B_j\theta}{e_j}{\Pi})$ with $1 \leq j \leq m$, obtaining CHL($\cdom$) proof trees $T'_j$ proving $\elimD{\Prog} \chlc \cat{B'_j\theta\rho_j}{\Pi}$ for some $\rho_j \in \Solc{\Omega_j}$ with $\domset{\rho_j} = \varset{\Omega_j}$.

Consider $\rho = \{W \mapsto \imath(d)\}$ and $\transform{C} \in \elimD{\Prog}$ of the form:
$$
\transform{C} :~ p'(\ntup{t}{n},W') ~\gets~ qV\!al(W'),\
  \left( \begin{array}{l}
    qV\!al(w_j'), ~\encode{w'_j \dgeq^? \imath(w_j)}, \\
    qBound(W', \imath(\alpha), w'_j), ~B'_j \\
  \end{array}\right)_{j = 1 \ldots m.} \\
$$
Obviously, $\rho \in \Solc{\Omega}$ and $\domset{\rho} = \varset{\Omega}$. To finish the proof we must prove $\elimD{\Prog} \chlc \cat{A'\!\rho}{\Pi}$. We claim that this can be done with a CHL($\cdom$) proof tree $T'$ whose root inference is a {\bf DA} step of the form:
$$
\displaystyle\frac
{~
  \begin{array}{l}
     ~~~~(~ \cat{(t'_i\rho == t_i\theta')}{\Pi} ~)_{i = 1 \ldots n} \\
     ~~~~\cat{(W\rho == W'\theta')}{\Pi} \\
     ~~~~\cat{qV\!al(W')\theta'}{\Pi}  \\
     \left( \begin{array}{l}
      \cat{qV\!al(w'_j)\theta'}{\Pi} \\
      \cat{\encode{w'_j \dgeq^? \imath(w_j)}\theta'}{\Pi} \\
      \cat{qBound(W',\imath(\alpha),w'_j)\theta'}{\Pi} \\
      \cat{B'_j\theta'}{\Pi}
    \end{array}\right)_{j = 1 \ldots m} \\
  \end{array}
~}
{\cat{p'(\ntup{t'}{n},W)\rho}{\Pi}} ~ (\spadesuit)
$$
using $\transform{C}$ instantiated by the substitution $\theta' = \theta \uplus \rho_1 \uplus \dots \uplus \rho_m \uplus \{W' \mapsto \imath(d)\}$.
We check that the premises of ($\spadesuit$) can be derived from $\elimD{\Prog}$ in CHL($\cdom$):
\begin{itemize}
\item
$\elimD{\Prog} \chlc \cat{(t'_i\rho == t_i\theta')}{\Pi}$ for $i = 1 \ldots n$.
By construction of $\rho$ and $\theta'$\!, these are equivalent to prove $\elimD{\Prog} \chlc \cat{(t'_i == t_i\theta)}{\Pi}$ for $i = 1 \ldots n$ and these hold with CHL($\cdom$) proof trees of only one {\bf EA} node because of $t'_i \approx_\Pi t_i\theta$, which is a consequence of the first $n$ premises of ($\clubsuit$).
\item
$\elimD{\Prog} \chlc \cat{(W\rho == W'\theta')}{\Pi}$.
By construction of $\rho$ and $\theta'$\!, this is equivalent to prove $\elimD{\Prog} \chlc \cat{(\imath(d) == \imath(d))}{\Pi}$ which results trivial.
\item
$\elimD{\Prog} \chlc \cat{qV\!al(W')\theta'}{\Pi}$.
By construction of $\theta'$, this is equivalent to prove $\elimD{\Prog} \chlc \cat{qV\!al(\imath(d))}{\Pi}$. We trivially have that $\imath(d) \in \mbox{ran}(\imath)$. Then, by Lemma \ref{lema:expr}, this premise holds.
\item
$\elimD{\Prog} \chlc \cat{qV\!al(w'_j)\theta'}{\Pi}$ for $j = 1 \ldots m$.
By construction of $\theta'$ and Lemma \ref{lema:expr} we must prove, for any fixed $j$, that $\qval{w'_j\rho_j}$ is true in $\cdom$. As $\rho_j \in \Solc{\Omega_j}$ we know $\rho_j \in \Solc{\qval{w'_j}}$, therefore $\qval{w'_j\rho_j}$ is trivially true in $\cdom$.
\item
$\elimD{\Prog} \chlc \cat{\encode{w'_j \dgeq^? \imath(w_j)}\theta'}{\Pi}$ for $j = 1 \ldots m$.
We reason for any fixed $j$.
If $w_j =\ ?$ this results trivial.
Otherwise, it amounts to $\qbound{\imath(w_j),\imath(\tp),w'_j\rho_j}$ being true in $\cdom$, by construction of $\theta'$ and Lemma \ref{lema:expr}.
As seen before, $\qval{w'_j\rho_j}$ is true in $\cdom$, therefore $w'_j\rho_j = \imath(e'_j)$ for some $e'_j \in \aqdom$. From the side conditions of ($\clubsuit$) we have $w_j \dleq e_j$.
On the other hand, $\rho_j \in \Solc{\Omega_j}$ and, in particular, $\rho_j \in \Solc{\qbound{\imath(e_j),\imath(\tp),w'_j}}$. This, together with $w'_j\rho_j = \imath(e'_j)$, means $e_j \dleq e'_j$, which with $w_j \dleq e_j$ implies $w_j \dleq e'_j$, i.e. $\qbound{\imath(w_j),\imath(\tp),w'_j\rho_j}$ is true in $\cdom$.
\item
$\elimD{\Prog} \chlc \cat{qBound(W',\imath(\alpha),w'_j)\theta'}{\Pi}$ for $j = 1 \ldots m$.
We reason for any fixed $j$. By construction of $\theta'$ and Lemma \ref{lema:expr}, we must prove that $\qbound{\imath(d),\imath(\alpha),w'_j\rho_j}$ is true in $\cdom$. As seen before, $\qval{w'_j\rho_j}$ is true in $\cdom$, therefore $w'_j\rho_j = \imath(e'_j)$ for some $e'_j \in \aqdom$. From the side conditions of ($\clubsuit$) we have $d \dleq \alpha \circ e_j$. On the other hand, $\rho_j \in \Solc{\Omega_j}$ and, in particular, $\rho_j \in \Solc{\qbound{\imath(e_j), \imath(\tp), w'_j}}$. This, together with $w'_j\rho_j = \imath(e'_j)$, means $e_j \dleq e'_j$. Now, $d \dleq \alpha \circ e_j$ and $e_j \dleq e'_j$ implies $d \dleq \alpha \circ e'_j$, i.e. $\qbound{\imath(d),\imath(\alpha),w'_j\rho_j}$ is true in $\cdom$.
\item
$\elimD{\Prog} \chlc \cat{B'_j\theta'}{\Pi}$ for $j = 1 \ldots m$.
In this case, it is easy to see that $B'_j\theta' = B'_j\theta\rho_j$ by construction of $\theta'$ and because of the program transformation rules.
On the other hand, proof trees $T'_j$ proving $\elimD{\Prog} \chlc \cat{B'_j\theta\rho_j}{\Pi}$ can be obtained by inductive hypothesis as seen before.
\end{itemize}

\smallskip\noindent
[2. $\Rightarrow$ 1.] {\em (the transformation is sound).}
We assume that $T'$ is a a CHL($\cdom$) proof tree
witnessing $\elimD{\Prog} \chlc \cat{A'\rho}{\Pi}$
for some $\rho \in \Sol{\cdom}{\Omega}$ such that $\domset{\rho} = \varset{\Omega}$.
We want to to show the existence of a $\QCHL(\qdom,\cdom)$ proof tree $T$
witnessing $\Prog \qchldc \cqat{A}{d}{\Pi}$.
We reason  by complete induction on $\Vert T' \Vert$.
There are three possible cases according to the the syntactic form of the atom $A'$\!.
In each case we argue how to build the desired proof tree $T$.

\noindent --- $A'$ is a primitive atom $\kappa$.
In this case due to {\bf TQCA} and {\bf TPA} we can assume $A = \kappa$ and $\Omega = \{\qval{\imath(\tp)},\ \qbound{\imath(d),\imath(\tp),\imath(\tp)}\}$.
Note that $\domset{\rho} = \varset{\Omega} = \emptyset$ implies $\rho = \varepsilon$.
Now, from $\elimD{\Prog} \chlc \cat{\kappa\varepsilon}{\Pi}$ follows $\Pi \model{\cdom} \kappa$ due to the {\bf PA} inference, and therefore we can prove $\Prog \qchldc \cqat{\kappa}{d}{\Pi}$ with a proof tree $T$ containing only one {\bf QPA} node.

\noindent --- $A'$ is an equation $t == s$.
In this case due to {\bf TQCA} and {\bf TEA} we can assume $A = (t == s)$ and $\Omega = \{\qval{\imath(\tp)},\ \qbound{\imath(d),\imath(\tp),\imath(\tp)}\}$.
Note that $\domset{\rho} = \varset{\Omega} = \emptyset$ implies $\rho = \varepsilon$.
Now, from $\elimD{\Prog} \chlc \cat{(t == s)\varepsilon}{\Pi}$ follows $t \approx_\Pi s$ due to the {\bf EA} inference, and therefore we can prove $\Prog \qchldc \cqat{(t == s)}{d}{\Pi}$ with a proof tree $T$ containing only one {\bf QEA} node.

\noindent --- $A'$ is a defined atom $p'(\ntup{t'}{n},W)$ with $p' \in DP^{n+1}$.
In this case due to {\bf TQCA} and {\bf TDA} we can assume $A = p(\ntup{t'}{n})$ and $\Omega = \{\qval{W},\ \qbound{\imath(d),\imath(\tp),W}\}$.
On the other hand, $T'$ must be rooted by a {\bf DA} step ($\spadesuit$) using a clause $\transform{C} \in \elimD{\Prog}$ instantiated by a substitution $\theta'$. We can assume that ($\spadesuit$), $\transform{C}$ and the corresponding clause $C \in \Prog$ have the form already displayed in [1. $\Rightarrow$ 2.].

By construction of $\transform{C}$\!, we can assume $\transform{B_j} = (B'_j,\ w'_j)$.
Let $\theta = \theta'{\upharpoonright}\varset{C}$ and $\rho_j = \theta'{\upharpoonright}\varset{w'_j} ~ (1 \ge j \ge m)$.
Then, due to the premises $\cat{qV\!al(w'_j)\theta'}{\Pi}$ of ($\spadesuit$) and Lemma \ref{lema:expr} we can assume $e'_j \in \aqdom ~ (1 \leq j \leq m)$ such that $w'_j\rho_j = \imath(e'_j)$.

To finish the proof, we must prove $\Prog \qchldc \cqat{A}{d}{\Pi}$.
We claim that this can be done with a $\QCHL(\qdom,\cdom)$ proof tree $T$ whose root inference is a {\bf QDA} step of the form of ($\clubsuit$), as displayed in [1. $\Rightarrow$ 2.], using clause $C$ instantiated by $\theta$.
In the premises of this inference we choose $d_i = \tp ~ (1 \leq i \leq n)$ and $e_j = e'_j ~ (1 \leq j \leq m)$.
Next we check that these premises can be derived from $\Prog$ in $\QCHL(\qdom,\cdom)$ and that the side conditions are fulfilled:
\begin{itemize}
\item
$\Prog \qchldc \cqat{(t'_i == t_i\theta)}{d_i}{\Pi}$ for $i = 1 \ldots n$.
This amounts to $t'_i \approx_\Pi t_i\theta$ which follows from the first $n$ premises of ($\spadesuit$) given that $t'_i\rho = t'_i$ and $t_i\theta' = t_i\theta$.
\item
$\Prog \qchldc \cqat{B_j\theta}{e_j}{\Pi}$ for $j = 1 \ldots m$.
From $\transform{B_j} = (B_j',  w'_j)$ and due to rule {\bf TQCA}, we have $\transform{(\cqat{(B_j\theta)}{e_j}{\Pi})} = (\cat{B_j\theta}{\Pi}, \Omega_j)$ where $\Omega_j = \{\qval{w'_j},\ \qbound{\imath(e_j), \imath(\tp),$ $w'_j}\}$. From the premises of ($\spadesuit$) and the fact that $B'_j\theta' = B'_j\theta\rho_j$ we know that $\elimD{\Prog} \chlc \cat{B'_j\theta\rho_j}{\Pi}$ with a CHL($\cdom$) proof tree $T'_j$ such that $\Vert T'_j \Vert < \Vert T' \Vert$. Therefore $\Prog \qchldc \cqat{B_j\theta}{e_j}{\Pi}$ follows by inductive hypothesis provided that $\rho_j \in \Solc{\Omega_j}$. In fact, due to the form of $\Omega_j$, $\rho_j \in \Solc{\Omega_j}$ holds iff $w'_j\rho_j = \imath(e'_j)$ for some $e'_j$ such that $e_j \dleq e'_j$, which is the case because of the choice of $e_j$.
\item
$e_j \dgeq^? w_j$ for $j = 1 \ldots m$.
Trivial in the case that $w_j =\ ?$.
Otherwise they are equivalent to $w_j \dleq e'_j$ which follow from premises $\cat{\encode{w'_j \dgeq^? \imath(w_j)}\theta'}{\Pi}$ (i.e. $\cat{\encode{w'_j\rho_j \dgeq^? \imath(w_j)}}{\Pi}$) of ($\spadesuit$) and Lemma \ref{lema:expr}.
\item
$d \dleq d_i$ for $i = 1 \ldots n$.
Trivially hold due to the choice of $d_i = \tp$.
\item
$d \dleq \alpha \circ e_j$ for $j = 1 \ldots m$.
Note that $\rho \in \Solc{\Omega}$ implies the existence of $d' \in \aqdom$ such that $\imath(d') = W\rho$ and $d \dleq d'$.
On the other hand, $e_j = e'_j$ by choice.
It suffices to prove $d' \dleq \alpha \circ e'_j$ for $j = 1 \ldots m$. Premises of ($\spadesuit$) and Lemma \ref{lema:expr} imply that $\qbound{W'\theta',\imath(\alpha),w'_j\theta'}$ is true in $\cdom$. Moreover, $W'\theta' = W\rho = \imath(d')$ because of another premise of ($\spadesuit$) and $w'_j\theta' = \imath(e'_j)$ as explained above.
Therefore $\qbound{W'\theta',\imath(\alpha),w'_j\theta'}$ amounts to $\qbound{\imath(d'),\imath(\alpha),\imath(e'_j)}$ which guarantees $d' \dleq \alpha \circ e'_j ~ (1 \leq j \leq m)$. \mathproofbox
\end{itemize}
\end{proof*}

The goal transformation correctness is established by the next theorem, which will rely on the previous result:

\begin{thm}
\label{thm:QCLP2CLP:goals}
Let $G$ be a goal for a $\qclp{\qdom}{\cdom}$-program $\Prog$ such that $qV\!al$ and $qBound$ do not occur in $G$.
Let $\Prog' = \elimD{\Prog}$ and $G' = \elimD{G}$.
Assume a $\cdom$-subtitution $\sigma$, a mapping $\mu : \warset{G} \to \aqdomd{\qdom}$ and a satisfiable finite set of $\cdom$-constraints $\Pi$.
Then, the following two statements are equivalent:
\begin{enumerate}
\item
$\langle \sigma, \mu, \Pi \rangle \in \Sol{\Prog}{G}$.
\item
$\langle \theta, \Pi \rangle \in \Sol{\Prog'}{G'}$ for some $\theta$ that verifies the following requirements:
  \begin{enumerate}
  \item
  $\theta =_{\varset{G}} \sigma$,
  \item
  $\theta =_{\warset{G}} \mu\imath$ and
  \item
  $W\theta \in \mbox{ran}(\imath)$ for each $W \in \varset{G'} \setminus (\varset{G} \cup \warset{G})$.
  \end{enumerate}
\end{enumerate}
\end{thm}

\begin{proof*}
As explained in Subsection \ref{sec:cases:qclp} the syntax of goals in $\qclp{\qdom}{\cdom}$-programs is the same as that of goals for $\sqclp{\simrel}{\qdom}{\cdom}$-programs, which is described in Section \ref{sec:sqclp}.
Therefore $G$, and $G'$ due to rule {\bf TG}, must have the following form:
$$
\begin{array}{c@{\hspace{1mm}}c@{\hspace{1mm}}l}
G &:& (~\qat{B_j}{W_j},\ W_j \dgeq^? \!\beta_j~)_{j = 1 \ldots m} \\
G' &:& (~qV\!al(W_j),~  \encode{W_j\dgeq^? \imath(\beta_j)},~  qV\!al(w'_j),~ qBound(W_j,$ $\imath(\tp),w'_j),~  B_j'~)_{j=1 \dots m} \\
\end{array}
$$
with $\transform{B_j} = (B'_j,  w'_j) ~ (1 \leq j \leq m)$.
Note that, because of rule {\bf TQCA}, we have $\transform{(\cqat{B_j\sigma}{W_j\mu}{\Pi})} = (\cat{B_j'\sigma}{\Pi}, \Omega_j)$ with $\Omega_j = \{\qval{w'_j},\ \qbound{\imath(W_j\mu),\imath(\tp),$ $w'_j}\}$ for $j = 1 \ldots m$.
We now prove each implication.

\smallskip\noindent
[1. $\Rightarrow$ 2.]
Let  $\langle \sigma, \mu, \Pi \rangle \in \Sol{\Prog}{G}$.
This means, by Definition \ref{dfn:qclp-goalsol}, $W_j\mu \dgeq^? \!\beta_j$ and $\Prog \qchldc \cqat{B_j\sigma}{W_j\mu}{\Pi}$ for $j = 1 \ldots m$.
In these conditions, Theorem \ref{thm:QCLP2CLP:programs} guarantees $\Prog' \chlc \cat{B'_j\sigma\rho_j}{\Pi} ~ (1 \leq j \leq m)$ for some $\rho_j \in \Solc{\Omega_j}$ such that $\domset{\rho_j} = \varset{\Omega_j}$.
It is easy to see that $\varset{G'} \setminus (\varset{G} \cup \warset{G}) = \varset{\Omega_1} \uplus \cdots \uplus \varset{\Omega_m}$.
Therefore it is possible to define a substitution $\theta$ verifying $\theta =_{\varset{G}} \sigma$, $\theta =_{\warset{G}} \mu\imath$ and $\theta =_{\domset{\rho_j}} \rho_j ~ (1 \leq j \leq m)$. Trivially, $\theta$ satisfies conditions 2.(a) and 2.(b). It also satisfies condition 2.(c) because for any $j$ and any variable $X$ such that $X \in \varset{\Omega_j}$, we have a constraint $\qval{X} \in \Omega_j$ implying, due to Lemma \ref{lema:expr}, $X\rho_j \in \mbox{ran}(\imath)$ (because $\rho_j \in \Solc{\Omega_j}$).

In order to prove $\langle \theta, \Pi \rangle \in \Sol{\Prog'}{G'}$ in the sense of Definition \ref{dfn:clp-goalsol} we check the following items:
\begin{itemize}
\item
By construction, $\theta$ is a $\cdom$-substitution.
\item
By the theorem's assumptions, $\Pi$ is a satisfiable and finite set of $\cdom$-constraints.
\item
$\Prog' \chlc \cat{A\theta}{\Pi}$ for every atom $A$ in $G'$.
Because of the form of $G'$ we have to prove the following for any fixed $j$:
\begin{itemize}
\item
$\Prog' \chlc \cat{qV\!al(W_j)\theta}{\Pi}$.
By construction of $\theta$ and Lemma \ref{lema:expr}, this amounts to $\qval{\imath(W_j\mu)}$ being true in $\cdom$, which is trivial consequence of $W_j\mu \in \aqdom$.
\item
$\Prog' \chlc \cat{\encode{W_j \dgeq^? \!\imath(\beta_j)}\theta}{\Pi}$.
If $\beta_j =\ ?$ this becomes trivial. Otherwise, $W_j\theta = \imath(W_j\mu)$ by construction of $\theta$, and by Lemma \ref{lema:expr} it suffices to prove $\qbound{\imath(\beta_j),\imath(\tp),\imath(W_j\mu)}$ is true in $\cdom$. This follows from $W_j\mu \dgeq^? \!\beta_j$, that is ensured by $\langle \sigma, \mu, \Pi \rangle \in \Sol{\Prog}{G}$.
\item
$\Prog' \chlc \cat{qV\!al(w'_j)\theta}{\Pi}$.
By construction of $\theta$ and Lemma \ref{lema:expr}, this amounts to $\qval{w'_j\rho_j}$ being true in $\cdom$, that is guaranteed by $\rho_j \in \Solc{\Omega_j}$.
\item
$\Prog' \chlc \cat{qBound(W_j,\imath(\tp),w'_j)\theta}{\Pi}$.
By construction of $\theta$ and Lemma \ref{lema:expr}, this amounts to $\qbound{\imath(W_j\mu),\imath(\tp),w'_j\rho_j}$ being true in $\cdom$, that is also guaranteed by $\rho_j \in \Solc{\Omega_j}$.
\item
$\Prog' \chlc \cat{B'_j\theta}{\Pi}$.
Note that, by construction of $\theta$, $B'_j\theta = B'_j\sigma\rho_j$. On the other hand, $\rho_j$ has been chosen above to verify $\Prog' \chlc \cat{B'_j\sigma\rho_j}{\Pi}$.
\end{itemize}
\end{itemize}

\smallskip\noindent
[2. $\Rightarrow$ 1.]
Let $\langle \theta, \Pi \rangle \in \Sol{\Prog'}{G'}$ and assume that $\theta$ verifies 2.(a), 2.(b) and 2.(c).
In order to prove $\langle \sigma, \mu, \Pi \rangle \in \Sol{\Prog}{G}$ in the sense of Definition \ref{dfn:qclp-goalsol} we must prove the following items:
\begin{itemize}
\item
By the theorem's assumptions, $\sigma$ is a $\cdom$-substitution, $\mu : \warset{G} \to \aqdomd{\qdom}$ and $\Pi$ is a satisfiable finite set of $\cdom$-constraints.
\item
$W_j\mu \dgeq^? \!\beta_j$.
We reason for any fixed $j$.
If $\beta_j =\ ?$ this results trivial.
Otherwise, we have $\Prog' \chlc \cat{\encode{W_j \dgeq^? \imath(\beta_j)}\theta}{\Pi}$ which, by condition 2.(b) and Lemma \ref{lema:expr} amounts to $\qbound{\imath(\beta_j),\imath(\tp),\imath(W_j\mu)}$ is true $\cdom$, i.e. $W_j\mu \dgeq \beta_j$.
\item
$\Prog \qchldc \cqat{B_j\sigma}{W_j\mu}{\Pi}$ for $j = 1 \ldots m$.
We reason for any fixed $j$. Let $\rho_j$ be the restriction of $\theta$ to $\varset{\Omega_j}$.
Then, $\Prog' \chlc \cat{B'_j\sigma\rho_j}{\Pi}$ follows from $\langle \theta,\Pi \rangle \in \Sol{\Prog'}{G'}$ and $B'_j\theta = B'_j\sigma\rho_j$.
Therefore, $\Prog \qchldc \cqat{B_j\sigma}{W_j\mu}{\Pi}$ follows from Theorem 5.3 provided that $\rho_j \in \Solc{\Omega_j}$.
By Lemma \ref{lema:expr} and the form of $\Omega_j$, $\rho_j \in \Solc{\Omega_j}$ holds iff $\Prog' \chlc \cat{qV\!al(w'_j\rho_j)}{\Pi}$ and $\Prog' \chlc \cat{qBound(\imath(W_j\mu),\imath(\tp),w'_j\rho_j)}{\Pi}$, which is true because $\langle \theta,\Pi \rangle \in \Sol{\Prog'}{G'}$ and construction of $\rho_j$. \mathproofbox
\end{itemize}
\end{proof*}

%% file: J4_3.tex
\subsection{Solving SQCLP Goals}
\label{sec:implemen:solving}

In this subsection we show that the transformations from the two previous subsections can be used to define abstract goal solving systems for SQCLP and arguing about their correctness.
In the sequel we consider a given $\sqclp{\simrel}{\qdom}{\cdom}$-program $\Prog$\! and a goal $G$ for $\Prog$ whose atoms are all relevant for $\Prog$.
We also consider $\Prog' \!= \elimS{\Prog}$, $G' = \elimS{G}$, $\Prog'' \!= \elimD{\Prog'}$ and $G'' = \elimD{G'}$.
Due to the definition of both elim$_\simrel$ and elim$_\qdom$, we can assume:
$$
\begin{array}{c@{\hspace{1mm}}c@{\hspace{1mm}}l}
G &:& (~ \qat{A_i}{W_i},~ W_i \dgeq^? \!\beta_i ~)_{i = 1 \ldots m} \\
G' &:& (~ \qat{A^i_\sim}{W_i},~ W_i \dgeq^? \!\beta_i ~)_{i = 1 \ldots m} \\
G'' &:& (~ qV\!al(W_i),~ \encode{W_i \dgeq^? \imath(\beta_i)},~ qV\!al(w'_i),~ qBound(W_i, \imath(\tp), w'_i),~ A_i' ~)_{i=1 \dots m} \\
&& \mbox{ where } \transform{A_i} = (A'_i, w'_i).
\end{array}
$$

We start by presenting an auxiliary result.

\begin{lem}
\label{lem:soltoground}
Assume $\Prog$\!, $G$, $\Prog'$\!, $G'\!$, $\Prog''$ and $G''$\! as above.
Let $\langle \sigma', \Pi \rangle \in \Sol{\Prog''}{G''}$, $\nu \in \Solc{\Pi}$ and $\theta = \sigma'\nu$.
Then $\langle \theta, \Pi \rangle \in \Sol{\Prog''}{G''}$. Moreover, $W \theta \in \mbox{ran}(\imath)$ for every $W \in \varset{G''} \setminus \varset{G}$.\footnote{Note that $\warset{G} \subseteq \varset{G''} \setminus \varset{G}$.}
\end{lem}
\begin{proof}
Consider an arbitrary atom $A''$\! occurring in $G''$\!.
Because of $\langle \sigma', \Pi \rangle \in \Sol{\Prog''}{G''}$ we have $\Prog \chlc \cat{A''\sigma'}{\Pi}$.
On the other hand, because of $\nu \in \Solc{\Pi}$ we have $\emptyset \model{\cdom} \Pi\nu$ and therefore also $\Pi \model{\cdom} \Pi\nu$.
This and Definition 3.1(4) of \cite{RR10TR} ensure $\cat{A''\sigma'}{\Pi} \entail{\cdom} \cat{A''\sigma'\nu}{\Pi}$, i.e. $\cat{A''\sigma'}{\Pi} \entail{\cdom} \cat{A''\theta}{\Pi}$.
This fact, $\Prog'' \chlc \cat{A''\sigma'}{\Pi}$ and the Entailment Property for Programs in $\clp{\cdom}$ imply $\Prog'' \chlc \cat{A''\theta}{\Pi}$.
Therefore, $\langle \theta,\Pi \rangle \in \Sol{\Prog''}{G''}$.

Consider now any $W \in \varset{G''} \setminus \varset{G}$.
By construction of $G''$\!, one of the atoms occurring in $G''$ is $qV\!al(W)$.
Then, due to $\langle \sigma'\Pi \rangle \in \Sol{\Prog''}{G''}$ we have $\Prog'' \chlc \cat{qV\!al(W\sigma')}{\Pi}$.
Because of Lemma \ref{lema:dec}(1) this implies $\Pi \model{\cdom} \qval{W\sigma'}$, i.e. $\Solc{\Pi} \subseteq \Solc{\qval{W\sigma'}}$.
Since $\nu \in \Solc{\Pi}$ we get $\nu \in \Solc{\qval{W\sigma'}}$, i.e. $W\sigma'\nu \in \mbox{ran}(\imath)$.
Since $W\sigma'\nu = W\theta$, we are done.
\end{proof}

Next, we explain how to define an abstract goal solving system for SQCLP from a given abstract goal solving system for CLP.

\begin{defn}
\label{dfn:SQCLPGSS}
Let CLP-AGSS be an abstract goal solving system for $\clp{\cdom}$ (in the sense of Definition \ref{dfn:clp-goalsolsys}).
Then we define {\em SQCLP-AGSS} as an abstract goal solving system for $\sqclp{\simrel}{\qdom}{\cdom}$ that works as follows:
\begin{enumerate}
\item
Given a goal $G$ for the $\sqclp{\simrel}{\qdom}{\cdom}$-program $\Prog$, consider $\Prog'$\!, $G'$\!, $\Prog''$\! and $G''$ as explained at the beginning of the subsection.
\item
For each solution $\langle \sigma',\Pi \rangle$ computed by CLP-AGSS for $G''$\!, $\Prog''$\! and for any $\nu \in \Solc{\Pi}$, SQCLP-AGSS computes $\langle \sigma, \mu, \Pi \rangle$ where $\theta = \sigma'\nu$, $\sigma = \theta{\upharpoonright}\varset{G}$ and $\mu = \theta\imath^{-1}{\upharpoonright}\warset{G}$.
Note that $\mu$ is well-defined thanks to Lemma \ref{lem:soltoground}. \mathproofbox
\end{enumerate}
\end{defn}

The next theorem ensures that SQCLP-AGSS is correct provided that CLP-AGSS is also correct.
The proof relies on the semantic results of the two previous subsections.

\begin{thm}
Assume that CLP-AGSS is correct (in the sense of Definition \ref{dfn:clp-goalsolsys}).
Let SQCLP-AGSS be as in the previous definition.
Then SQCLP-AGSS is correct in the sense of Definition \ref{dfn:goalsolsys}.
\end{thm}
\begin{proof*}
We separately prove that SQCLP-AGSS is {\em sound} and {\em weakly complete}.

\smallskip\noindent
--- {\em SQCLP-AGSS is sound.}
Let $\langle \sigma, \mu, \Pi \rangle$ be an answer computed by SQCLP-AGSS for $G,\Prog$. We must prove that
$\langle \sigma, \mu, \Pi \rangle \in \Sol{\Prog}{G}$.
By Definition \ref{dfn:SQCLPGSS} we can assume $\langle \sigma', \Pi \rangle \in \Sol{\Prog''}{G''}$ and $\nu \in \Solc{\Pi}$ such that $\sigma = \theta{\upharpoonright}\varset{G}$ and $\mu = \theta\imath^{-1}{\upharpoonright}\warset{G}$ with $\theta = \sigma'\nu$.
Because of Lemma \ref{lem:soltoground} we have $\langle \theta, \Pi \rangle \in \Sol{\Prog''}{G''}$ and $W\theta \in \mbox{ran}(\imath)$ for every $W \in \varset{G''} \setminus \varset{G}$.
Note that:
\begin{itemize}
\item
$\theta =_{\varset{G'}} \sigma$.
This follows from $\varset{G'} = \varset{G}$ and the construction of $\sigma$.
\item
$\theta =_{\warset{G'}} \mu\imath$.
This follows from $\warset{G'} = \warset{G}$ and $\theta =_{\warset{G}} \mu\imath$, that is obvious from the construction of $\mu$.
\item
$W\theta \in \mbox{ran}(\imath)$ for each $W \in \varset{G''} \setminus (\varset{G'} \cup \warset{G'})$.
This is a consequence of Lemma \ref{lem:soltoground} since $\varset{G''} \setminus (\varset{G'} \cup \warset{G'}) \subseteq \varset{G''} \setminus \varset{G'}$ and $\varset{G'} = \varset{G}$.
\end{itemize}
From the previous items and Theorem \ref{thm:QCLP2CLP:goals} we get $\langle \sigma, \mu, \Pi \rangle \in \Sol{\Prog'}{G'}$, which trivially implies $\langle \sigma, \mu, \Pi \rangle \in \Sol{\Prog}{G}$ because of Theorem \ref{thm:SQCLP2QCLP:goals}.

\smallskip\noindent
--- {\em SQCLP-AGSS is weakly complete.}
Let $\langle \eta, \rho, \emptyset \rangle \in \GSol{\Prog}{G}$ be  a ground solution for $G$ w.r.t. $\Prog$.
We must prove that it is subsumed---in the sense of Definition \ref{dfn:goalsol}(3)---by some answer $\langle \sigma, \mu, \Pi \rangle$ computed by SQCLP-AGSS for $G,\Prog$.

By Theorem \ref{thm:SQCLP2QCLP:goals} we have that $\langle \eta, \rho, \emptyset \rangle$ is also a ground solution for $G'$ w.r.t. $\Prog'$.
In addition, by Theorem \ref{thm:QCLP2CLP:goals} $\langle \eta', \emptyset \rangle \in \Sol{\Prog''}{G''}$ for some $\eta'$ such that
\begin{itemize}
  \item (1) $\eta' =_{\varset{G'}} \eta$,
  \item (2) $\eta' =_{\warset{G'}} \rho\imath$ and hence $\eta'(\imath^{-1}) =_{\warset{G'}} \rho$, and
  \item $W\eta' \in \mbox{ran}(\imath)$ for each $W \in \varset{G''} \setminus (\varset{G'} \cup \warset{G'})$ (i.e. $w'_i\eta' \in \mbox{ran}(\imath)$ for each $i = 1 \ldots m$ such that $w'_i$ is a variable).
\end{itemize}
By construction of $\eta'$, it is clear that $\langle \eta', \emptyset \rangle$ is ground.
Now, by the weak completeness of CLP-AGSS, there is some computed answer $\langle \sigma',\Pi \rangle$ subsuming $\langle \eta',\emptyset \rangle$, therefore satisfying
\begin{itemize}
  \item (3) there is some $\nu \in \Solc{\Pi}$, and
  \item (4) $\eta' =_{\varset{G''}} \sigma'\nu$.
\end{itemize}

Because of Definition \ref{dfn:SQCLPGSS} one can build a SQCLP-AGSS computed answer $\langle \sigma, \mu, \Pi \rangle$ as follows:
\begin{itemize}
  \item (5) $\sigma = \sigma'\nu{\upharpoonright}\varset{G}$
  \item (6) $\mu = \sigma'\nu\imath^{-1}{\upharpoonright}\warset{G}$
\end{itemize}

We now check that $\langle \sigma, \mu, \Pi \rangle$ subsumes $\langle \eta, \rho, \emptyset \rangle$:
\begin{itemize}
  \item $W_i\rho \dleq W_i\mu$ and even $W_i\rho = W_i\mu$ because:
  $$W_i\rho =_{(2)} W_i\eta'(\imath^{-1}) =_{(4)} W_i\sigma'\nu(\imath^{-1}) =_{(6)} W_i\mu \enspace .$$
  \item $\nu \in \Solc{\Pi}$ by (3) and, moreover, for any $X \in \varset{G}$:
  $$ X\eta =_{(1)} X\eta' =_{(4)} X\sigma'\nu =_{(\dagger)} X\sigma'\nu\nu =_{(5)} X\sigma\nu$$
  therefore $\eta =_{\varset{G}} \sigma\nu$. 
  
  The step ($\dagger$) is justified because $\nu \in \mbox{Val}_\cdom$ implies $\nu = \nu\nu$. \mathproofbox
\end{itemize}
\end{proof*}

%% file: J5_0.tex
\section{A Practical Implementation}
\label{sec:practical}

This section is devoted to the more practical aspects of the SQCLP programming scheme and it is developed in three subsections: Subsection \ref{sec:practical:SQCLP} explains what steps must be given when implementing a programming scheme like this and why the theoretic results presented in the previous sections---with special emphasis in those in Subsection \ref{sec:implemen:solving}---become useful for implementation.
Subsection \ref{sec:practical:prototype} introduces a prototype implementation and explains how to write programs and how to solve goals.
Finally, in Subsection \ref{sec:practical:efficiency} we study the unavoidable overload introduced in the system by qualifications and proximity relations when comparing the execution of programs without any explicit use of such resources.

\input{J5_1}

\input{J5_2}

\input{J5_3} 

%% file: J5_1.tex
\subsection{SQCLP over a CLP Prolog System}
\label{sec:practical:SQCLP}

Assume an available CLP Prolog System, a $\sqclp{\simrel}{\qdom}{\cdom}$-program $\Prog$ and a goal $G$ for $\Prog$.
Our purpose is to implement a goal solving system for SQCLP following Definition \ref{dfn:SQCLPGSS}.
We will examine each step in this schema, discussing the necessary implementation details for putting theory into practice.

The first step is to obtain the transformed programs $\Prog' \!= \elimS{\Prog}$ and $\Prog'' \!= \elimD{\Prog'}$;
and the transformed goals $G' = \elimS{G}$ and $G'' = \elimD{G'}$.
According to Definition \ref{def:sqclptransform}(3), $\Prog' \!= \elimS{\Prog}$ is of the form $EQ_\simrel \cup \hat{\Prog}_\simrel$, where $EQ_\simrel$ is obtained following Definition \ref{def:EQ} and $\hat{\Prog}_\simrel$ is obtained following Definition \ref{def:sqclptransform}(3,2).
When implementing  $EQ_\simrel$ a first difficulty arises, namely the implementation of $\sim\ \in DP^2$\!, which apparently requires one clause of the form:
$$u \sim u' ~\qgets{\tp}~ \qat{\mbox{pay}_{\lambda}}{?}$$
for each pair $u,u' \in B_\cdom$ such that $\simrel(u,u') = \lambda \neq \bt$, and one clause of the form:
$$
c(\ntup{X}{n}) \sim c'(\ntup{Y}{n}) ~\qgets{\tp}~ \qat{\mbox{pay}_{\lambda}}{?},\
(\qat{(X_i \sim Y_i)}{?})_{i = 1 \ldots n}
$$
for each pair $c, c' \in DC^{n}$ such that $\simrel(c,c') = \lambda \neq \bt$.
While this should obviously require an infinite number of clauses (because $DC^n$ is infinite and $\simrel(c,c) = \tp \neq \bt$ for all $c\in DC^{n}$; and also $B_\cdom$ is infinite---in general---and $\simrel(u,u) = \tp \neq \bt$ for every $u \in B_\cdom$), in practice, it is enough to limit the number of clauses to the finite number of different basic values $u \in B_\cdom$ and constructors $c \in DC^n$ that can be found either in $\Prog$\!, $G$ or $\simrel$.

A similar difficulty arises when codifying the clauses for predicates $\mbox{pay}_\lambda \in DP^0$\!, which according to Definition \ref{def:EQ} there should be a clause of the form:
$$\mbox{pay}_\lambda ~\qgets{\lambda}~$$
in $EQ_\simrel$ for each $\lambda \in \aqdomd{\qdom}$.
In this case, the solution is also similar because it suffices to generate enough $\mbox{pay}_\lambda$ clauses for the finite $\lambda \in \aqdomd{\qdom}$ that can be found occurring either in the clauses of $\hat{\Prog}_\simrel$ or in the clauses implementing the predicate $\sim\ \in DP^2$\!.

The construction of $\hat{\Prog}_\simrel$, following Definition \ref{def:sqclptransform}, presents no particular difficulties.
For each clause $C : (p(\ntup{t}{n}) \qgets{\alpha} \tup{B}) \in \Prog$ we will generate a finite set $\hat{\mathcal{C}}_\simrel$ of clauses, because the number of symbols $p'$ such that $\simrel(p,p') = \lambda \neq \bt$ will be also finite in practice.
Finally, the construction of $G'$ is merely the straightforward replacement of all the occurrences of `==' in $G$ by `$\sim$'.

The transformation elim$_\qdom$ from $\qclp{\qdom}{\cdom}$ into $\clp{\cdom}$, is defined in Definition \ref{def:qclptransform}.
$\Prog'' \!= \elimD{\Prog'}$ is obtained by incorporating the two clauses of the program $E_\qdom$ to the result of applying the transformation rules in Figure \ref{fig:transformation} to the $\qclp{\qdom}{\cdom}$-program $\Prog'$\!.
Applying the transformation rules is straightforward, but the codification of constraints $\qval{X}$ and $\qbound{X,Y,Z}$ in $E_\qdom$ requires some clarification.
In our implementation we have considered the constraint domain  $\rdom$, as well as any qualification domain that can be built from $\B$, $\U$ and $\W$ by means of the strict cartesian product operation ${\otimes}$ including, in particular, $\U{\otimes}\W$.
These qualification domains are existentially expressible in $\rdom$, therefore the constraints can be implemented by defined predicates as explained in Section \ref{sec:implemen:QCLP2CLP}. In particular in our prototype implementation these predicates are:
\begin{verbatim}
% qval( +QDom, ?W ):
qval(b, 1).
qval(u, W) :- {W > 0, W =< 1}.
qval(w, W) :- {W >= 0}.
qval((D1,D2), (W1,W2)) :- qval(D1, W1), qval(D2, W2).

% qbound( +QDom, ?X, ?Y, ?Z ):
qbound(b, 1, 1, 1).
qbound(u, X, Y, Z) :- {X =< Y * Z}.
qbound(w, X, Y, Z) :- {X >= Y + Z}.
qbound((D1,D2), (X1,X2), (Y1,Y2), (Z1,Z2)) :- qbound(D1, X1, Y1, Z1),
    qbound(D2, X2, Y2, Z2).
\end{verbatim}
Instead of using different $qV\!al $ and $qBound $ predicates for each allowable $\qdom$, our prototype implementation just uses two predicates $qV\!al $ and $qBound$ with an extra first argument, used to encode an identifier of some specific allowable $\qdom$. 
This parameter can take either the value {\tt b} (for $\B$), {\tt u} (for $\U$), {\tt w} (for $\W$) or a pair {\tt (D$_1$,D$_2$)}
(for $\qdom_1 \otimes \qdom_2$), where each {\tt D}$_i$ can be either {\tt b}, {\tt u}, {\tt w} or another pair representing a product.
For instance {\tt ((u,w),w)} represents the qualification domain $(\U{\otimes}\W){\otimes}\W$.
The compiler ensures that this argument takes the correct value for each transformed program and goal depending on the specific instance of the SQCLP scheme the program is written for.

After obtaining $\Prog''$\! and $G''$\!, the CLP Prolog System is used to solve $G''$ \!w.r.t. $\Prog''$\!.
This yields computed answers of the form $\langle \sigma',\Pi \rangle$.
Now, instead of obtaining particular substitutions $\theta = \sigma'\nu$, $\sigma = \theta{\upharpoonright}\varset{G}$ and $\mu = \theta\imath^{-1}{\upharpoonright}\warset{G}$ for any $\nu \in \Solc{\Pi}$ as explained in Definition \ref{dfn:SQCLPGSS}(2), our prototype implementation limits itself to display $\langle \sigma',\Pi \rangle$ as the computed answer in SQCLP.
The reason behind this behavior is that, in general (and particularly in $\rdom$), it is impossible to enumerate the possible solutions $\nu \in \Solc{\Pi}$.
Thus, it results impossible to implement a technique for obtaining all the possible triples $\langle \sigma, \mu, \Pi \rangle$.
Note, however, that for a user it will not be difficult to distinguish, in the shown computed answers, what variable bindings correspond to the substitution $\sigma$ of the triple and what to the substitution $\mu$, even when the qualification variables are not bound but constrained, which is a common behavior in the context of CLP programming.

However, for the SQCLP-AGSS of Definition \ref{dfn:SQCLPGSS}, it results mandatory to define the computed answers in terms of $\nu \in \Solc{\Pi}$, because our SQCLP-semantics relies on proving instances of $G$ for some specific ground values of the variables in $\warset{G}$.

%% file: J5_2.tex
\subsection{{\tt (S)QCLP}: A Prototype System for SQCLP Programming}
\label{sec:practical:prototype}

The prototype implementation object of this subsection is publicly available, and can be found at:
\begin{center}
{\tt http://gpd.sip.ucm.es/cromdia/qclp}
\end{center}

The system currently requires the user to have installed either {\em SICStus Prolog} or {\em SWI-Prolog}, and it has been tested to work under Windows, Linux and MacOSX platforms.
The latest version available at the time of writing this paper is {\tt 0.6}.
If a latter version is available some things might have changed but in any case the main aspects of the system should remain the same. Please consult the {\em changelog} provided within the system itself for specific changes between versions.

SQCLP is a very general programming scheme and, as such, it supports different proximity relations, different qualification domains and different constraint domains when building specific instances of the scheme for any specific purpose.
As it would result impossible to provide an implementation for every admissible triple (or instance of the scheme), it becomes mandatory to decide in advance what specific instances will be available for writing programs in {\tt (S)QCLP}.
In essence:
\begin{enumerate}
\item
In its current state, the only available constraint domain is $\rdom$.
Thus, under both {\em SICStus Prolog} and {\em SWI-Prolog} the library {\tt clpr} will provide all the available primitives in {\tt (S)QCLP} programs.
\item
The available qualification domains are: `{\tt b}' for the domain $\B$; `{\tt u}' for the domain $\U$; `{\tt w}' for the domain $\W$; and any strict cartesian product of those, as e.g. `{\tt (u,w)}' for the product domain $\U{\otimes}\W$.
\item With respect to proximity relations, the user will have to provide, in addition to the two symbols and their proximity value, their {\em kind} (either predicate or constructor) and their {\em arity}. Both kind and arity must be the same for each pair of symbols having a proximity value different of $\bt$.
\end{enumerate}
Note, however, that when no specific proximity relation $\simrel$ is provided for a given program, $\sid$ is then assumed.
Under this circumstances, an obvious technical optimization consists on transforming the original program only with elim$_\qdom$, thus reducing the overload introduced in this case by elim$_\simrel$.
The reason behind this optimization is that for any given $\sqclp{\sid}{\qdom}{\cdom}$-program $\Prog$, it is also true that $\Prog$ is a $\qclp{\qdom}{\cdom}$-program, therefore $\elimD{\elimS{\Prog}}$ must semantically be equivalent to $\elimD{\Prog}$.
Nevertheless, $\elimD{\Prog}$ behaves more efficiently than $\elimD{\elimS{\Prog}}$ due to the reduced number of resulting clauses.
Thus, in order to improve the efficiency, the system will avoid the use of elim$_\simrel$ when no proximity relation is provided by the user.

The final available instances in the {\tt (S)QCLP} system are: $\sqclp{\simrel}{\tt b}{\tt clpr}$, $\sqclp{\simrel}{\tt u}{\tt clpr}$, $\sqclp{\simrel}{\tt w}{\tt clpr}$, $\sqclp{\simrel}{\tt (u,w)}{\tt clpr}$, \ldots{} and their counterparts in the QCLP scheme when $\simrel$ = $\sid$.

\subsubsection{Programming in {\tt (S)QCLP}}
\label{sec:practical:prototype:programming}

Programming in {\tt (S)QCLP} is straightforward if the user is accustomed to the Prolog programming style.
However, there are three syntactic differences with pure Prolog:
\begin{enumerate}
  \item Clauses implications are replaced by ``{\tt <-$d$-}'' where $d \in \aqdom$. If $d = \tp$, then the implication can become just ``{\tt <--}''. E.g. ``{\tt <-0.9-}'' is a valid implication in the domains $\U$ and $\W$; and ``{\tt <-(0.9,2)-}'' is a valid implication in the domain ${\U{\otimes}\W}$.
  \item Clauses in {\tt (S)QCLP} are not finished with a dot ({\tt .}). They are separated by layout, therefore all clauses in a {\tt (S)QCLP} program must start in the same column. Otherwise, the user will have to explicitly separate them by means of semicolons ({\tt ;}).
  \item After every body atom (even constraints) the user can provide a threshold condition using `{\tt \#}'. The notation `{\tt ?}' can also be used instead of some particular qualification value, but in this case the threshold condition `{\tt \#?}' can be omitted.
\end{enumerate}
Comments are as in Prolog:
\begin{verbatim}
% This is a line comment.
/* This is a multi-line comment, /* and they nest! */. */
\end{verbatim}
and the basic structure of a {\tt (S)QCLP} program is the following (line numbers are for reference):
\begin{center}
\small
\begin{tabular}{rp{11cm}}
\multicolumn{2}{@{\hspace{0mm}}l}{{\bf File:} {\em Peano.qclp}} \\[2mm]
\tiny 1 & \verb+% Directives...+ \\
\tiny 2 & \verb+# qdom w+ \\[2mm]
\tiny 3 & \verb+% Program clauses...+ \\
\tiny 4 & \verb+% num( ?Num )+ \\
\tiny 5 & \verb+num(z) <--+ \\
\tiny 6 & \verb+num(s(X)) <-1- num(X)+ \\
\end{tabular}
\end{center}
In the previous small program, lines {\tt 1}, {\tt 3} and {\tt 4} are line comments, line {\tt 2} is a program directive telling the compiler the specific qualification domain the program is written for, and lines {\tt 5} and {\tt 6} are program clauses defining the well-known Peano numbers.
As usual, comments can be written anywhere in the program as they will be completely ignored (remember that a line comment must necessarily end in a new line character, therefore the very last line of a file cannot contain a line comment),
and directives must be declared before any program clause.
There are three program directives in {\tt (S)QCLP}:
\begin{enumerate}
\item
The first one is ``{\tt \#qdom} {\em qdom}'' where {\em qdom} is any system available qualification domain, i.e. {\tt b}, {\tt u}, {\tt w}, {\tt (u,w)}\ldots{}
See line {\tt 2} in the previous program sample as an example.
This directive is mandatory because the user must tell the compiler for which particular qualification domain the program is written.
\item
The second one is ``{\tt \#prox} {\em file}'' where {\em file} is the name of a file (with extension {\tt .prox} containing a proximity relation.
If the name of the file starts with a capital letter, or it contains spaces or any special character, {\em file} will have to be quoted with single quotes.
For example, assume that with our program file we have another file called {\em Proximity.prox}.
Then, we would have to write ``{\tt \#prox `Proximity'}'' to link the program with such proximity relation.
This directive is optional, and if omitted, the system assumes that the program is of an instance of the QCLP scheme.
\item
The third one is ``{\tt \#optimized\_unif}''. This directive tells the compiler that the program is intended to be used with the optimized version of the unification algorithm, what improves the general efficiency of the goal solving process. However, as noted at the end of Section \ref{sec:sqclp}, this could have the effect of losing valid answers, although we conjecture that if the proximity relation is transitive and if the program clauses do not make use of attenuation factors other that $\tp$, this will not happen.
\end{enumerate}

Proximity relations are defined in files of extension {\tt .prox} with the following form:
\begin{center}
\small
\begin{tabular}{rp{11cm}}
\multicolumn{2}{@{\hspace{0mm}}l}{{\bf File:} {\em Work.prox}} \\[2mm]
\tiny 1 & \verb+% Predicates: pprox( S1, S2, Arity, Value ).+ \\
\tiny 2 & \verb+pprox(wrote, authored, 2, (0.9,0)).+ \\[2mm]
\tiny 3 & \verb+% Constructors: cprox( S1, S2, Arity, Value ).+ \\
\tiny 4 & \verb+cprox(king_lear, king_liar, 0, (0.8,2)).+ \\
\end{tabular}
\end{center}
where the file can contain {\tt pprox/4} Prolog facts, for defining proximity between predicate symbols of any arity;
or {\tt cprox/4} Prolog facts, for defining proximity between constructor symbols of any arity.
The arguments of both {\tt pprox/4} and {\tt cprox/4} are: the two symbols, their arity and its proximity value.
Note that, although it is not made explicit the qualification domain this proximity relation is written for, all values in it must be of the same specific qualification domain, and this qualification domain must be the same declared in every program using the proximity relation.
Otherwise, the solving of equations may produce unexpected results or even fail.

Reflexive and symmetric closure is inferred by the system, therefore, there is no need for writing reflexive proximity facts, nor the symmetric variants of proximity facts already provided.
You can notice this in the previous sample file in which neither reflexive proximity facts, nor the symmetric proximity facts to those at lines {\tt 2} and {\tt 4} are provided.
In the case of being explicitly provided, additional (repeated) solutions might be computed for the same given goal, although soundness and weak completeness of the system should still be preserved.
Transitivity is neither checked nor inferred so the user will be responsible for ensuring it if desired.

As the reader would have already guessed, the file  {\tt Work.prox} implements the proximity relation $\simrel_r$ of Example \ref{exmp:pr} in {\tt (S)QCLP}. Finally, the program $\Prog_r$ of Example \ref{exmp:pr} can be represented in {\tt (S)QCLP} as follows:
\begin{center}
\small
\begin{tabular}{rp{11cm}}
\multicolumn{2}{@{\hspace{0mm}}l}{{\bf File:} {\em Work.qclp}} \\[2mm]
\tiny 1 & \verb+# qdom (u,w)+ \\
\tiny 2 & \verb+# prox 'Work'+ \\[2mm]
\tiny 3 & \verb+% famous( ?Author )+ \\
\tiny 4 & \verb+famous(shakespeare) <-(0.9,1)-+ \\[2mm]
\tiny 5 & \verb+% wrote( ?Author, ?Book )+ \\
\tiny 6 & \verb+wrote(shakespeare, king_lear) <-(1,1)-+ \\
\tiny 7 & \verb+wrote(shakespeare, hamlet) <-(1,1)-+ \\[2mm]
\tiny 8 & \verb+% good_work( ?Work )+ \\
\tiny 9 & \verb+good_work(X) <-(0.75,3)- famous(Y)#(0.5,100), authored(Y,X)+ \\
\end{tabular}
\end{center}
Note that, at line {\tt 1} the qualification domain $\U{\otimes}\W$ is declared, and at line {\tt 2} the proximity relation at {\tt Work.prox} is linked to the program.
In addition, observe that one threshold constraint is imposed for a body atom in the program clause at line {\tt 9}, effectively requiring to prove {\tt famous(Y)} for a qualification value of {\em at least} {\tt (0.5,100)} to be able to use this program clause.

Finally, we explain how constraints are written in {\tt (S)QCLP}.
As it has already been said, only $\rdom$ is available, thus both in {\em SICStus Prolog} and {\em SWI-Prolog} the library {\tt clpr} is the responsible for providing the available primitive predicates.
Given that constraints are primitive atoms of the form {\tt r($\ntup{{\tt t}}{n}$)} where {\tt r} $\in PP^n$ and {\tt t$_i$} are terms; primitive atoms share syntax with usual Prolog atoms.
At this point, and having that many of the primitive predicates are syntactically operators (hence not valid identifiers), the syntax for predicate symbols has been extended to include operators, therefore predicate symbols like $op_+ \in PP^3$, which codifies the operation {\tt +} in a 3-ary predicate, will let us to build constraints of the form {\tt +(A,B,C)}, that must be understood as in $A+B=C$ or $C=A+B$.
Similarly, predicate symbols like $cp_> \in PP^2$, which codifies the comparison operator {\tt >} in a binary predicate, will let us to build constraints of the form {\tt >(A,B)}, that must be understood as in $A > B$.
Any other primitive predicate such as {\em maximize} $\in PP^1$, will let us to build constraints like {\tt maximize(X)}.
Valid primitive predicate symbols include {\tt +}, {\tt -}, {\tt *}, {\tt /}, {\tt >}, {\tt >=}, {\tt =<}, {\tt <}, {\tt maximize}, {\tt minimize}, etc.

Threshold constraints can also be provided for primitive atoms in the body of clauses with the usual notation.
Note, however, that due the semantics of SQCLP, all primitive atoms can be trivially proved with $\tp$ if they ever succeeds---so threshold constraints become, in this case, of no use.

The syntax for constraints explained above follows the standard syntax for atoms.
Nonetheless, the system also allows to write these constraints in a more natural infix notation.
More precisely, {\tt +(A,B,C)} can be also written in the infix form {\tt A+B=C} or {\tt C=A+B}, and {\tt >(X,Y)} in the infix form {\tt X>Y}; and similarly for other $op$ and $cp$ constraints.
When using infix notation, threshold conditions can be set by (optionally) enclosing the primitive atom between parentheses, therefore becoming {\tt (A+B=C)\#$\tp$}, {\tt (C=A+B)\#$\tp$} or {\tt (X>Y)\#$\tp$} (or any other valid qualification value or `?').
Using parentheses is recommended to avoid understanding that the threshold condition is set only for the last term in the constraint, which would not be the case.
Note that even in infix notation, operators cannot be nested, that is, terms {\tt A}, {\tt B}, {\tt C}, {\tt X} and {\tt Y} cannot have operators as main symbols (neither in prefix nor in infix notation), so the infix notation is just a syntactic sugar of its corresponding prefix notation.

As a final example for constraints, one could write the predicate {\tt double/2} in {\tt (S)QCLP}, for computing the double of any given number, with just the clause \verb+double(N,D) <-- *(N,2,D)+, or \verb+double(N,D) <-- N*2=D+ for a clause with a more natural syntax.

\subsubsection{The interpreter for {\tt (S)QCLP}}

The interpreter for {\tt (S)QCLP} has been implemented on top of both {\em SICStus Prolog} and {\em SWI-Prolog}.
To load it, one must first load her desired (and supported) Prolog system and then load the main file of the interpreter---i.e. {\tt qclp.pl}---, that will be located in the main {\tt (S)QCLP} folder among other folders.
Once loaded, one will see the welcome message and will be ready to compile and load programs, and to execute goals.
\begin{verbatim}
WELCOME TO (S)QCLP 0.6
(S)QCLP is free software and comes with absolutely no warranty.
Support & Updates: http://gpd.sip.ucm.es/cromdia/qclp.

Type ':help.' for help.
yes
| ?-
\end{verbatim}

From the interpreter for {\tt (S)QCLP} one can, in addition to making use of any standard Prolog goals, use the specific {\tt (S)QCLP} commands required for both interacting with the {\tt (S)QCLP} system, and  for compiling/loading SQCLP programs.
All these commands take the form:
\begin{center}
\tt :command.
\end{center}
if they do not require arguments, or:
\begin{center}
\tt :command({\em Arg}$_1$, \ldots, {\em Arg}$_n$).
\end{center}
if they do; where each argument {\tt\em Arg}$_i$ must be a prolog atom unless stated otherwise.
The most useful commands are:
\begin{itemize}
\item {\tt :cd(}{\em Folder}{\tt ).} \\
Changes the working directory to {\em Folder}. {\em Folder} can be an absolute or relative path.
\item {\tt :compile(}{\em Program}{\tt ).} \\
Compiles the {\tt (S)QCLP} program `{\em Program}{\tt .qclp}' producing the equivalent  Prolog program in the file `{\em Program}{\tt .pl}'.
\item {\tt :load(}{\em Program}{\tt ).} \\
Loads the already compiled {\tt (S)QCLP} program `{\em Program}{\tt .qclp}' (note that the file `{\em Program}{\tt .pl}' must exist for the program to correctly load).
\item {\tt :run(}{\em Program}{\tt ).} \\
Compiles the {\tt (S)QCLP} program `{\em Program}{\tt .qclp}' and loads it afterwards.
This command is equivalent to executing: {\tt :compile(}{\em Program}{\tt ), :load(}{\em Program}{\tt ).}
\end{itemize}

For illustration purposes, we will assume that you have the files {\tt Work.prox} and {\tt Work.qclp} (both as seen before) in the folder {\tt $\sim$/examples}.
Under these circumstances, after loading your preferred Prolog system and the interpreter for {\tt (S)QCLP}, one would only have to change the working directory to that where the files are located:
\begin{Verbatim}[commandchars=\\\{\}]
| ?- :cd('\prox/examples').
\end{Verbatim}
and run the program:
\begin{verbatim}
| ?- :run('Work').
\end{verbatim}
If no errors are encountered, one should see the output:
\begin{verbatim}
| ?- :run('Work').
<Work> Compiling...
<Work> QDom: 'u,w'.
<Work> Prox: 'Work'.
<Work> Translating to QCLP...
<Work> Translating to CLP...
<Work> Generating code...
<Work> Done.
<Work> Loaded.
yes
\end{verbatim}
and now everything is ready to execute goals for the program loaded.

\subsubsection{Executing SQCLP-Goals}

Recall that goals have the form
$\qat{A_1}{W_1},\ \ldots,\ \qat{A_m}{W_m} \sep W_1 \!\dgeq^? \beta_1,\ \ldots,\ W_m \dgeq^? \!\beta_m$
which in actual {\tt (S)QCLP} syntax becomes:
\begin{verbatim}
| ?- A1#W1, ..., Am#Wm :: W1 >= B1, ..., Wm >= Bm.
\end{verbatim}
Note the following:
\begin{enumerate}
\item Goals must end in a dot ({\tt .}).
\item The symbol `$\sep$' is replaced by `\verb+::+'.
\item The symbol `${\dgeq}^?$' is replaced by `\verb+>=+' (and this is independent of the qualification domain in use, so that it may mean $\le$ in $\W$).
\item Conditions of the form $W \dgeq^?\ ?$ {\em must be omitted}, therefore $\qat{A_1}{W_1}, \qat{A_2}{W_2} \sep W_1 \dgeq^?\ ?, W_2 \dgeq^? \!\beta_2$ becomes ``\verb+A1#W1, A2#W2 :: W2 >= B2.+'',
and $\qat{A}{W}\sep W \dgeq^?\ ?$ becomes just ``\verb+A#W.+''.
\end{enumerate}

Assuming now that we have loaded the program {\tt Work.qclp} as explained before, we can execute the goal $\qat{good\_work(king\_liar)}{W}\sep W \dgeq^? (0.5,100)$:
\begin{verbatim}
| ?- good_work(king_liar)#W::W>=(0.5,10).
W = (0.6,5.0) ?
yes
\end{verbatim}

\subsubsection{Examples}

To finish this subsection, we are now showing some additional goal executions using the interpreter for {\tt (S)QCLP} and the programs displayed along the paper.

\paragraph{Peano.}
Consider the program {\tt Peano.qclp} as displayed at the beginning of Subsection \ref{sec:practical:prototype:programming}.
Qualifications in this program are intended as a cost measure for obtaining a given number in the Peano representation, assuming that each use of the clause at line {\tt 6} requires to pay {\em at least} {\tt 1}.
In essence, threshold conditions will impose an upper bound over the maximum number obtainable in goals containing the atom {\tt num(X)}.
Therefore if we ask for numbers {\em up to} a cost of {\tt 3} we get the following answers:
\begin{center}
\small
\begin{tabular}{rp{11cm}}
\tiny Goal & \verb+?- num(X)#W::W>=3.+ \\[2mm]
\tiny Sol$_1$ & \verb+W = 0.0, X = z ? ;+ \\
\tiny Sol$_2$ & \verb+W = 1.0, X = s(z) ? ;+ \\
\tiny Sol$_3$ & \verb+W = 2.0, X = s(s(z)) ? ;+ \\
\tiny Sol$_4$ & \verb+W = 3.0, X = s(s(s(z))) ? ;+ \\
\tiny & \verb+no+ \\
\end{tabular}
\end{center}

\paragraph{Work.}
Consider now the program {\tt Work.qclp} and the proximity relation {\tt Work.prox}, both as displayed in Subsection \ref{sec:practical:prototype:programming} above.
In this program, qualifications behave as the conjunction of the certainty degree of the user confidence about some particular atom, and a measure of the minimum cost to pay for proving such atom.
In these circumstances, we could ask---just for illustration purposes---for famous authors with a minimum certainty degree---for them being actually famous---of {\tt 0.5}, and with a proof cost of no more than {\tt 30} (think of an upper bound for possible searches in different databases).
Such a goal would have, in this very limited example, only the following solution:
\begin{center}
\small
\begin{tabular}{rp{11cm}}
\tiny Goal & \verb+?- famous(X)#W::W>=(0.5,30).+ \\[2mm]
\tiny Sol$_1$ & \verb+W = (0.9,1.0), X = shakespeare ? ;+ \\
\tiny & \verb+no+ \\
\end{tabular}
\end{center}
meaning that we can have a confidence of {\tt shakespeare} being famous of {\tt 0.9}, and that we can prove it with a cost of {\tt 1}.

Now, in a similar fashion we could try to obtain different works that can be considered as good works by using the last clause in the example.
Limiting the search to those works that can be considered good with a qualification value better or equal to {\tt (0.5,100)} produce the following result:
\begin{center}
\small
\begin{tabular}{rp{11cm}}
\tiny Goal & \verb+?- good_work(X)#W::W>=(0.5,100).+ \\[2mm]
\tiny Sol$_1$ & \verb+W = (0.675,4.0), X = king_lear ? ;+ \\
\tiny Sol$_2$ & \verb+W = (0.6,5.0), X = king_liar ? ;+ \\
\tiny & \verb+no+ \\
\end{tabular}
\end{center}
It is important to remark here that the qualification value obtained for a particular computed answer is not guaranteed to be the best possible one; rather, different computed answers may compute different qualification values which can be observed by the user.
This is easy to see if we try to solve a more particular goal:
\begin{center}
\small
\begin{tabular}{rp{11cm}}
\tiny Goal & \verb+?- good_work(king_liar)#W::W>=(0.675,4.0).+ \\[2mm]
\tiny Sol$_1$ & \verb+W = (0.675,4.0) ? ;+ \\
\tiny & \verb+no+ \\
\end{tabular}
\end{center}
That is, not only {\tt good\_work(king\_liar)} can be proved for for {\tt W = (0.6,5.0)} as shown in Sol$_2$ above, but also with {\tt W = (0.675,4.0)}, which results a better qualification value (i.e. greater certainty degree and lower proof cost).
\paragraph{Library.}
Finally, consider the program $\Prog_s$ and the proximity relation $\simrel_s$, both as displayed in Figure \ref{fig:library} of Section \ref{sec:sqclp}.
As it has been said when this example was introduced, the predicate {\em guessRdrLvl} takes advantage of attenuation factors to encode heuristic rules to compute reader levels on the basis of vocabulary level and other book features.
As an illustration of use, consider the following goal:
\begin{center}
\small
\begin{tabular}{rp{11cm}}
\tiny Goal & \verb+?- guessRdrLvl(book(2, 'Dune', 'F. P. Herbert', english, sciFi,+\\
& \verb+   medium, 345), Level)#W.+ \\[2mm]
\tiny Sol$_1$ & \verb+W = 0.8, Level = intermediate ? ;+ \\
& $\cdots$ \\
\tiny Sol$_6$ & \verb+W = 0.7, Level = upper ?+ \\
\tiny & \verb+yes+ \\
\end{tabular}
\end{center}
Here we ask for possible ways of classifying the second book in the library according to reader levels.
We obtain as valid solutions, among others, {\tt intermediate} with a certainty factor of {\tt 0.8}; and {\tt upper} with a certainty factor of {\tt 0.7}.
These valid solutions show that the predicate {\em guessRdrLvl} tries with different levels for any certain book based on the heuristic implemented by the qualified clauses.

To conclude, consider now the goal proposed in Section \ref{sec:sqclp} for this program.
For such goal we obtain:
\begin{center}
\small
\begin{tabular}{rp{11cm}}
\tiny Goal & \verb+?- search(german, essay, intermediate, ID)#W::W>=0.65.+\\[2mm]
\tiny Sol$_1$ & \verb+W = 0.8, ID = 4 ?+ \\
\tiny & \verb+yes+ \\
\end{tabular}
\end{center}
What tells us that the forth book in the library is written in German, it can be considered to be an essay, and it is targeted for an intermediate reader level.
All this with a certainty degree of {\em at least} {\tt 0.8}.

%% file: J5_3.tex
\subsection{Efficiency}
\label{sec:practical:efficiency}
The minimum---and unavoidable---overload introduced by qualifications and proximity relations in the transformed programs manifests itself in the case of {\tt (S)QCLP} programs which use the identity proximity relation and have $\tp$ as the attenuation factor of all their clauses.
In order to measure this overload we have made some experiments using some program samples, taken from the {\em SICStus Prolog Benchmark} that can be found in:
\begin{center}
\tt http://www.sics.se/isl/sicstuswww/site/performance.html
\end{center}
and we have compared the time it took to repeatedly execute a significant number of times each program in both {\tt (S)QCLP} and {\em SICStus Prolog} making use of a {\em slightly} modified (to ensure a correct behavior in both systems) version of the harness also provided in the same site.

From all the programs available in the aforementioned site, we selected the following four:
\begin{itemize}
\item {\em naivrev:} naive implementation of the predicate that reverses the contents of a list.
\item {\em deriv:} program for symbolic derivation.
\item {\em qsort:} implementation of the well-known sorting algorithm {\em Quicksort}.
\item {\em query:} obtaining the population density of different countries.
\end{itemize}
No other program could be used because they included impure features such as cuts which are not currently supported by our system. In order to adapt these Prolog programs to our setting 
the following modifications were required:
\begin{enumerate}
\item All the program clause are assumed to have $\tp$ as attenuation factor.
      After including these attenuation factors, we obtain as results QCLP programs.
      More specifically we obtain two QCLP programs for each initial Prolog program,
      one using  the qualification domain $\B$ (because this domain uses trivial constraints), and another using  the qualification domain $\U$ (which uses $\rdom$-constraints).
\item We define an empty proximity relation, allowing us to obtain two additional 
SQCLP-programs.
\item By means of the program directive ``{\tt \#optimized\_unif}'' defined in Subsection \ref{sec:practical:prototype:programming}, each SQCLP program can be also executed 
    in this optimized mode. Therefore each original Prolog Program produces six {\tt (S)QCLP}
    programs, denoted as Q(b), Q(u), PQ(b), PQ(u), SQ(b) and SQ(u) in Table \ref{table:results}.
\end{enumerate}

Additionally some  minor modifications to the program samples have been introduced for compatibility reasons, i.e. additions using the predicate {\tt is/2} were replaced, both in the Prolog version of the benchmark and in the multiple {\tt (S)QCLP} versions, by {\tt clpr} constraints. 
In any case, all the program samples used for this benchmarks in this subsection can be found in the folder {\tt benchmarks/} of the {\tt (S)QCLP} distribution.

Finally, we proceeded to solve the same goals for every version of the benchmark programs, both in {\em SICStus Prolog} and in {\tt (S)QCLP}.
The benchmark results can be found in Table \ref{table:results}. All the experiments were performed in a computer with a Intel(R) Core(TM)2 Duo CPU at 2.19GHz and with 3.5 GB RAM.

\begin{table}[ht]
\caption{Time overload factor with respect to Prolog}
\label{table:results}
\begin{minipage}{\textwidth}
\begin{tabular}{lrrrrrr}
\hline\hline
Program &
Q(b)\footnote{$\qclp{\B}{\rdom}$ version (i.e. the program does not have the {\tt \#prox} directive).} &
Q(u)\footnote{$\qclp{\U}{\rdom}$ version (i.e. the program does not have the {\tt \#prox} directive).} &
PQ(b)\footnote{$\sqclp{\sid}{\B}{\rdom}$ version.} &
PQ(u)\footnote{$\sqclp{\sid}{\U}{\rdom}$ version.} &
SQ(b)\footnote{$\sqclp{\sid}{\B}{\rdom}$ version with directive {\tt \#optimized\_unif}.} &
SQ(u)\footnote{$\sqclp{\sid}{\U}{\rdom}$ version with directive {\tt \#optimized\_unif}.} \\
\hline
naivrev & 1.80  & 10.71 & 4289.79 & 4415.11 & 56.22& 65.75\\
deriv   & 1.94  & 10.60 &  331.45 &  469.67 & 29.63 & 39.32 \\
qsort   & 1.05  & 1.11 & 135.59  & 136.98 & 2.51 & 2.83\\
query   & 1.02  & 1.12 & 7.17  & 7.13 & 3.80 & 3.88\\
\hline\hline
\end{tabular}
\vspace{-2\baselineskip}
\end{minipage}
\end{table}

The results in the table indicate the slowdown factor obtained for each version of each program.
For instance, the first column indicates that the time required for evaluating the goal corresponding to the sample program {\em naivrev} in $\qclp{\B}{\rdom}$ is about 1.80 times the required time for the evaluation of the same goal in Prolog. Next we discuss the results:
\begin{itemize}
\item {\em Influence of the qualification domain.}
      In general the difference between the slowdown factors obtained for the two considered
      qualification domains is not large. However, in the case of QCLP-programs {\em naivrev} and  {\em deriv}  the difference increases notably. 
      This is due to the different ratios of the $\B$-constraints w.r.t. the program and 
      $\U$-constraints w.r.t. the program. It must be noticed that the transformed programs
      are the same in both cases, but for the implementation of {\tt qval} and {\tt qbound}
      constraints, which is more complex for $\U$ as one can see in Section \ref{sec:practical:SQCLP}.
      In the case of {\em naivrev} and  {\em deriv} this makes a big difference because
      the number of computation steps directly required by the programs is much smaller
      than in the other cases. Thus the slowdown factor becomes noticeable for the qualification domain $\U$ in computations that requires a large number of steps.
      
\item {\em Influence of the proximity relation.}
The introduction of a proximity relation, even 
of empty, is very significative. This is due to the introduction of the predicate $\sim$,
which replaces Prolog unification. The situation even worsens when the computation
introduces large constructor terms, as in the case of {\em naivrev} which deals with 
Prolog lists. The efficient Prolog unification is replaced by an explicit term decomposition.

\item {\em Influence of the optimized unification.}
As explained at the end of Section \ref{sec:sqclp} this optimization can lead to the loss of solutions in general.
However, this is not the case for the chosen examples.
As seen in the table, the use of the program directive {\tt \#optimized\_unif} causes a clear increase in the efficiency of goal solving for these examples.

\end{itemize}

%% file: J6_0.tex
\section{Conclusions}
\label{sec:conclusions}


In our recent work \cite{RR10} we extended the classical CLP scheme to a new programming scheme SQCLP whose instances $\sqclp{\simrel}{\qdom}{\cdom}$ were parameterized by a proximity relation $\simrel$, a qualification domain $\qdom$ and a constraint domain $\cdom$.
This new scheme offered extra facilities for dealing with expert knowledge representation and flexible query answering.
In this paper we have contributed to the aforementioned scheme providing, in a more practical sense, both a semantically correct transformation technique, in two steps, for transforming SQCLP programs and goals intro equivalent CLP programs and goals;
and a prototype implementation on top CLP($\rdom$) systems like {\em SICStus Prolog} and {\em SWI-Prolog} of some particularly interesting instances of the scheme.

The two-step transformation technique presented in Section \ref{sec:implemen} has provided us with the needed theoretical results for effectively showing how proximity relations can be reduced to qualifications and clause annotations by means of the transformation elim$_\simrel$; and how qualifications and clause annotations can be reduced to classical CLP programming by means of the transformation elim$_\qdom$.
These two transformations altogether, ultimately enables the use of the classical mechanism of SLD resolution to obtain computed answers for SQCLP goals w.r.t SQCLP programs, via their equivalent CLP programs and goals and the computed answers obtained from them by any capable CLP goal solving procedure.

The prototype implementation presented in Section \ref{sec:practical} has finally allowed us to execute all the examples showed in this paper---and in previous ones---, and a series of benchmarks for measuring the overload actually introduced by proximity relations---or by similarity relations---and by clause annotations and qualifications.
While we are aware that the prototype implementation presented in this paper has to be considered a research application (and as such, we have to admit that it cannot be used for industrial applications), we think that it can contribute to the field as a quite complete implementation of an extension of the CLP($\rdom$) scheme with proximity relations and qualifications.
Some related implementation techniques and systems have been cited in the introduction.
However, as far as we know, no other implementation in this field has ever provided support for proximity (and similarity) relations, qualifications via clause annotations and CLP($\rdom$) style programming.
Moreover, our results in Section \ref{sec:implemen} on the semantic correctness of our implementation technique are in our opinion another contribution of this paper which has no counterpart in related approaches.


In the future, and taking advantage of the prototype system we have already developed, we plan to investigate possible applications which can profit from proximity relations and qualifications, such as in the area of flexible query answering.
In particular, we plan to investigate application related to flexible answering of queries to XML documents, in the line of \cite{CDG+09} and other related papers.
As support for practical applications, we also plan to increase the repertoire of constraint and qualification domains which can be used in the {\tt (S)QCLP} prototype, adding the constraint domain $\mathcal{FD}$ and the qualification domain $\W_d$ defined in Section 2.2.3 of \cite{RR10TR}.
On a more theoretical line, other possible lines of future work include:
a) extension of the SLD($\qdom$) resolution procedure presented in \cite{RR08} to a SQCLP goal solving procedure able to work with constraints and a proximity relation;
b) investigation of the conjecture stated at the end of Section \ref{sec:sqclp};
and c) extension of the QCFLP ({\em qualified constraint functional logic programming}) scheme in \cite{CRR09} to work with a proximity relation and higher-order functions, as well as the implementation of the resulting scheme in the CFLP($\cdom$)-system {\sf Toy} \cite{toy}.